\documentclass{article}
\usepackage{amsmath}  % Custom, but needs to be loaded first!
\usepackage{amsthm}
\usepackage{arxiv}

\usepackage[utf8]{inputenc} % allow utf-8 input
\usepackage[T1]{fontenc}    % use 8-bit T1 fonts
\usepackage{hyperref}       % hyperlinks
\usepackage{url}            % simple URL typesetting
\usepackage{booktabs}       % professional-quality tables
\usepackage{amsfonts}       % blackboard math symbols
\usepackage{nicefrac}       % compact symbols for 1/2, etc.
\usepackage{microtype}      % microtypography
\usepackage{cleveref}       % smart cross-referencing
\usepackage{lipsum}         % Can be removed after putting your text content
\usepackage{graphicx}
\usepackage{natbib}
\usepackage{doi}

\usepackage{comment}
\usepackage{subfigure}
\usepackage{amssymb}
\usepackage{nicefrac}
\usepackage{mathtools}

\usepackage{enumitem}
\usepackage{physics}
\usepackage{algpseudocode}
\usepackage{dsfont}
\usepackage{amsfonts}

\makeatletter
\newcommand{\@giventhatstar}[2]{\left(#1\;\middle|\;#2\right)}
\newcommand{\@giventhatnostar}[3][]{#1#2\;#1|\;#3#1}
\newcommand{\giventhat}{\@ifstar\@giventhatstar\@giventhatnostar}
\makeatother

\newcommand*\diff{\mathop{}\!\mathrm{d}}

\newcommand{\argdot}{\,\cdot\,}

\newcommand{\FP}[1]{\textcolor{blue}{FP: #1}}

\newenvironment{sproof}{%
	\proof}{\endproof}

\theoremstyle{plain}
\newtheorem{theorem}{Theorem}[section]
\newtheorem{lemma}[theorem]{Lemma}

\newtheorem{assumption}[theorem]{Assumption}

\newtheorem{innercustomthm}{Theorem}
\newenvironment{customthm}[1]
{\renewcommand\theinnercustomthm{#1}\innercustomthm}
{\endinnercustomthm}

\newtheorem{innercustomlemma}{Lemma}
\newenvironment{customlemma}[1]
{\renewcommand\theinnercustomlemma{#1}\innercustomlemma}
{\endinnercustomlemma}

\theoremstyle{definition}
\newtheorem{definition}[theorem]{Definition}

\theoremstyle{remark}

\newtheorem{example}[theorem]{Example}

\newcommand\restr[2]{{% we make the whole thing an ordinary symbol
		\left.\kern-\nulldelimiterspace % automatically resize the bar with \right
		#1 % the function
		\littletaller % pretend it's a little taller at normal size
		\right|_{#2} % this is the delimiter
}}
\newcommand{\littletaller}{\mathchoice{\vphantom{\big|}}{}{}{}}
\title{Verifying Approximate Equilibrium in Auctions}

\date{}

\author{ Fabian~R.~Pieroth \\
	\texttt{fabian.pieroth@tum.de} \\
	Technical University of Munich\\
	\And
	Tuomas Sandholm \\
	\texttt{sandholm@cs.cmu.edu}\\
	Carnegie Mellon University\\
	Strategy Robot, Inc.\\
	Optimized Markets, Inc. \\
	Strategic Machine, Inc.
}

\renewcommand{\headeright}{}
\renewcommand{\undertitle}{}
\renewcommand{\shorttitle}{Verifying Approximate Equilibrium in Auctions}

\hypersetup{
pdftitle={Verifying Approximate Equilibrium in Auctions},
pdfsubject={cs.GT},
pdfauthor={Fabian R. Pieroth, Tuomas Sandholm},
pdfkeywords={equilibrium verification; auctions; interdependent distributions; dispersion},
}

\begin{document}
\maketitle

\begin{abstract}
In practice, most auction mechanisms are not strategy-proof, so equilibrium analysis is required to predict bidding behavior. In many auctions, though, an exact equilibrium is not known and one would like to understand whether---manually or computationally generated---bidding strategies constitute an approximate equilibrium. We develop a framework and methods for estimating the distance of a strategy profile from equilibrium, based on samples from the prior and either bidding strategies or sample bids. We estimate an agent's utility gain from deviating to strategies
from a constructed finite subset of the strategy space. We use PAC-learning to give error bounds, both for independent and interdependent prior distributions. The primary challenge is that one may miss large
utility gains by considering only a finite subset of the strategy space. Our work differs from
prior research in two critical ways. First, we explore the impact of bidding strategies on altering opponents’ perceived prior distributions---instead of assuming the other agents to bid truthfully. Second, we delve into reasoning with interdependent priors, where the type of one agent may imply a distinct distribution for other agents. Our main contribution lies in establishing sufficient conditions for strategy profiles and a closeness criterion for
conditional distributions to ensure that utility gains estimated through our finite subset closely approximate the maximum gains. To our knowledge, ours is the first method to verify approximate equilibrium in any auctions beyond single-item ones. Also, ours is the first sample-based method for approximate equilibrium verification.
\end{abstract}

% keywords can be removed

% Paper body
\section{Introduction}

A central problem in mechanism design is understanding the strategic incentives of participants--in order to design mechanisms that lead to desired outcomes. A \textit{Bayesian Nash Equilibrium (BNE)}~\citep{harsanyiGamesIncompleteInformation1967} represents a fixed point in strategy space, where no agent has an incentive to deviate. This concept constitutes the central solution concept for games with incomplete information, such as auctions.

Mechanism design devotes significant attention to designing incentive-compatible mechanisms, where truthful bidding constitutes a BNE~\citep{hurwiczInformationallyDecentralizedSystems1972}. When bidders are aware that it is in their best interest to report their true valuation for an item, this knowledge leads to several desired effects.
For example, one can guarantee efficient outcomes, ensuring that the item is allocated to the bidder who values it the most. Furthermore, it simplifies the strategic decision-making process for the bidders, thereby saving resources. 

Nonetheless, practitioners typically employ mechanisms that are not incentive compatible, referred to as \emph{manipulable} mechanisms. For instance, the first-price mechanism is commonly used in real-world auctions. In the context of multi-unit sales, the U.S. Treasury has utilized discriminatory auctions for selling treasury bills since 1929~\citep{krishna2009auction}. 
%Sponsored search auctions often employ variations of the generalized second-price auction~\citep{edelmanInternetAdvertisingGeneralized2007, varianPositionAuctions2007}. 
Moreover, combinatorial auctions in practice typically use manipulable mechanisms such as first-price payments for their simplicity and other desirable features, or core-selecting payment rules intended to ensure the winners' payments are sufficient to maintain envy-freeness~\citep{dayCoreselectingPackageAuctions2008}. 

Several factors were identified why manipulable mechanisms are prevalent in practice. First, their rules are typically more straightforward to communicate. Second, incentive-compatible mechanisms have the potential to more readily expose the bidders’ confidential private information~\citep{rothkopfWhyAreVickrey1990}. Third, if information acquisition is costly, even incentive-compatible mechanisms have no dominant strategy for making information-gathering or valuation-computation decisions~\citep{Sandholm00:Issues,Larson01:Costly}.
%Larson, K. and Sandholm, T. 2001. Costly Valuation Computation in Auctions. In Proceedings of the Theoretical Aspects of Reasoning about Knowledge (TARK).
%
%Sandholm, T. 2000. Issues in Computational Vickrey Auctions. International Journal of Electronic Commerce, 4(3), 107-129. Special issue on Intelligent Agents for Electronic Commerce. invited reviewed submission. Early version in Proceedings of ICMAS-96.
%
%Conitzer, V. and Sandholm, T. 2006. Failures of the VCG mechanism in combinatorial auctions and exchanges. In Proceedings of the International Conference on Autonomous Agents and Multi-Agent Systems (AAMAS). Early version in Agent-Mediated Electronic Commerce (AMEC) workshop, 2004.
%
%Gilpin, A. and Sandholm, T. 2004. Arbitrage in combinatorial exchanges. Agent-Mediated Electronic Commerce (AMEC) workshop.
%
Additionally, incentive-compatible mechanisms, like the VCG mechanism~\citep{vickreyCounterspeculationAuctionsCompetitive1961, clarkeMultipartPricingPublic1971, grovesIncentivesTeams1973}, exhibit significant drawbacks in combinatorial auction contexts. First, they can result in minimal or even null revenues in spite of intense competition for the items~\citep{Conitzer06:Failures,ausubelLovelyLonelyVickrey2006}. Second, they may encourage collusion~\citep{Conitzer06:Failures,dayCoreselectingPackageAuctions2008}. Third, they may beget arbitrage opportunities~\citep{Gilpin04:Arbitrage}. Fourth, mechanisms that are incentive compatible in single-shot settings---like the VCG---typically do not remain incentive compatible over time across auctions where complementary or substitutable items are sold~\citep{Sandholm00:Issues}. Fifth, in scenarios such as sourcing, the repeated application of an incentive-compatible mechanism is not incentive compatible as the bid taker uses bids from one auction to modify the parameters (reserve prices or more sophisticated parameters) of future auctions~\citep{sandholmVeryLargeScale2013}. 

%\FP{Include importance of interdependent priors. Use mineral rights as example.}
%Determining equilibrium in auctions with continuous types and actions can often be formulated as the solution to a system of differential equations. Unless there are simple (uniform) distributional assumptions and simple assumptions about the bidders’ utility functions and the goods, we typically do not have a general solution theory. Even setting up the differential equations can be challenging.
Despite the significant academic work in auction theory, equilibrium strategies for manipulable mechanisms are primarily known only for very restricted, simple market models, such as single-item auctions with independent prior distributions~\citep{krishna2009auction}. 
Even worse, equilibria are not known to exist in general, but only in specific settings~\citep{Athey01:Single, renyNashEquilibriumDiscontinuous2020}. 
Fortunately, every strategy profile can be considered a $\varepsilon$-Bayesian Nash Equilibrium ($\varepsilon$-BNE) for some approximation factor $\varepsilon > 0$. Intuitively, $\varepsilon$ measures the potential utility gain an agent could achieve by deviating from its current strategy, assuming the other agents’ strategies remain unchanged.

As a result, recent efforts have concentrated on identifying strategies with an $\varepsilon$ as small as possible. Several computational techniques have demonstrated promise in discovering strong bidding strategies (e.g.,\citep{bosshardComputingBayesNashEquilibria2020, bichlerLearningEquilibriaSymmetric2021, bichlerLearningEquilibriaAsymmetric2023, bichlerComputingBayesNash2023}). 
Although there is strong empirical evidence suggesting that the approximation factor $\varepsilon$ is small for the computed strategies, their theoretical guarantees are limited in settings where no analytical equilibrium is available.
\citet{bichlerComputingBayesNash2023} rely on significant assumptions, including complete knowledge of the joint and marginal prior distributions, and their results are restricted to single-item auctions. \citet{bosshardComputingBayesNashEquilibria2020} introduce error bounds based on the precise calculation of metrics that are typically intractable to compute, such as the best-response \textit{ex interim} utility. Meanwhile, \citet{bichlerLearningEquilibriaSymmetric2021, bichlerLearningEquilibriaAsymmetric2023} employ a sampling-based strategy but do not provide error bounds.

\subsection{Contributions}
We introduce techniques with provable guarantees that identify the smallest approximation factor $\varepsilon$ for a strategy profile. Our methods require only access to samples from the type and bid distribution. The bids can either be observed directly, or, given access to the strategies, one can map the sampled types to their corresponding bids. Our results are applicable to single- and multi-item auctions with independent and interdependent prior distributions.

We analyze both the \textit{ex interim} and \textit{ex ante} settings.\footnote{We exclude the study of \textit{ex post} approximate equilibrium from our analysis because these concepts are based on worst-case, distribution-independent notions, rendering it impractical to assess through sampling from agents' type distributions.} In the \textit{ex interim} case, we bound the amount any agent can improve its utility by deviating from its current strategy, in expectation over the other agents' types, regardless of its own true type. In the weaker \textit{ex ante} setting, the expectation also includes the agent's own true type.

Our estimate is simple. It measures the maximum utility an agent can gain by deviating from its current strategy, averaged over the samples, where the alternative strategies considered are from a finite subset of the strategy space. We present upper bounds in the \textit{ex interim} case, denoted by $\hat{\varepsilon}$, and in the \textit{ex ante} case, denoted by $\tilde{\varepsilon}$. Specifically, we offer \textit{ex interim} guarantees $\hat{\varepsilon}$ for scenarios with independent prior distributions and \textit{ex ante} guarantees $\tilde{\varepsilon}$ for interdependent prior distributions.

Prior sampling-based methods operated under the assumption that agents play truthfully and have independent prior distributions, meaning they were only capable of verifying the truthful strategy under independent priors~\citep{balcanEstimatingApproximateIncentive2019}.
We expand upon this in two significant ways. First, our results hold for a large class of bidding strategies (as long as bids can neither change too fast nor too slowly as a function of an agent's type). This class satisfies common assumptions on equilibrium strategies made in auctions, such as monotonicity~\citep{renyExistenceMonotonePureStrategy2011, renyNashEquilibriumDiscontinuous2020}.
Second, we introduce findings for interdependent prior distributions. To achieve this, we consider a partition $\mathcal{B}$ of an agent's type space and establish an upper bound on the estimation error that utilizes the maximum total variation distance between the opponents' conditional distribution within each element $B \in \mathcal{B}$. In the arXiv version of their EC-19-Exemplary-AI-Paper-Award-winning extended abstract, \citet{balcanEstimatingApproximateIncentive2019} also presented---among other results---\textit{ex ante} guarantees for interdependent prior distributions. However, they retracted that result after we pointed out that it is incorrect~\citep{balcanEstimatingApproximateIncentive2019a}. That approach was flawed because it did not consider bidding strategies that can be functions of a bidder's type. 

We apply our estimation technique across several important auction classes. For instance, in the first-price auction, our error bound for a $B \in \mathcal{B}$ is $\tilde{O}\left(\tau_B + \left(n +\left(\kappa_{B} L_{\beta^{-1}_{\text{max}}} \right)^{-1} \right) /\sqrt{N_{B}}\right)$,
where $n$ is the number of bidders, $N_{B}$ is the number of samples within $B$, $[0, \kappa_B]$ denotes the range of the prior density, $L_{\beta^{-1}_{\text{max}}}$ is the maximum Lipschitz constant of an agent's inverse bidding strategy, and $\tau_B$ denotes the maximum total variation distance among the conditional prior distributions for types from $B$. It is important to note that $\tau_B$ does not need to become small for every $B$ in order to provide meaningful \textit{ex ante} guarantees, as the overall bound for the entire partition can still be small in expectation. For the case of independent prior distributions, this bound improves to become an \textit{ex interim} guarantee with $\tau_B=0$ and $N_{B} = N$.

We present similar results for a variety of auction formats, including combinatorial first-price auctions, uniform-price auctions, and discriminatory auctions.

\textbf{Key challenges}
To prove our guarantees, we aim to estimate the maximum possible amount an agent can improve its utility by deviating from its current strategy in both the \textit{ex interim} and \textit{ex ante} cases, respectively. We determine our error bounds, $\hat{\varepsilon}$ or $\tilde{\varepsilon}$, by quantifying the extent to which an agent may improve its utility, averaged over the samples, when considering alternative strategies from a finite set. To achieve this, we encounter two major technical challenges.

The first challenge arises from the limitation of searching over a finite set, potentially causing an agent to miss strategies that could significantly improve its utility. This occurs from auctions often having discontinuities in the utility functions. 
For instance, in both first- and second-price auctions, a slight increase in an agent's bid from just below to just above the highest bid of other agents alters the allocation, resulting in a sudden jump in utility.
For a given type of an agent, we consider a grid with an edge length $w$ over the action space, assuming the action space is $[0, 1]^m$ for some integer $m$. The critical question then becomes how much potential utility might be missed when searching over this finite grid and the effect of $w$ on this potential loss.

To tackle this issue, we utilize the concept of \emph{dispersion}~\citep{balcanDispersionDataDrivenAlgorithm2018}. In broad terms, a set of piecewise Lipschitz functions is $(w, k)$-dispersed if every ball of radius $w$ in the domain contains no more than $k$ discontinuities of the functions. 
Given $N$ samples from the prior and bidding distributions, we examine the dispersion of a set of \textit{ex post} utility functions, each defined by a sample and varying over one agent's bid. We demonstrate that if this set of functions is sufficiently dispersed, it is possible to control the error by searching for a best response over a finite grid with edge length $w$, rather than in the infinite action space. 
Crucially, we establish sufficient conditions on both the prior distribution \emph{and} bidding strategies to ensure this approach is viable.

The second major challenge arises under interdependent prior distributions. In such contexts, an agent gains additional information about the opponents' prior distributions upon learning its type. Given the continuous nature of these distributions, the probability of drawing the identical type more than once is zero, leading to the expectation that one would not collect more than a single sample from the same conditional prior distribution. 
We tackle this issue by considering a partition $\mathcal{B}$ of the type space for each agent and grouping samples that fall into the same element $B \in \mathcal{B}$. We demonstrate that if the total variation distance for the conditional distributions from types within $B$ is sufficiently small, then the aggregated samples can provide valuable insights about the conditional prior distributions for all types from $B$.

Finally, provided that the intrinsic complexities of the agents' utility functions are manageable (as determined by the learning-theoretic concept of pseudo-dimension~\citep{pollardConvergenceStochasticProcesses1984}), our empirical estimates $\hat{\varepsilon}$ and $\tilde{\varepsilon}$ quickly converge to the true approximation factors as the sample size increases.

\section{Related research}
In this section, we discuss additional related work on equilibrium-verification methods, emphasizing the contributions and limitations of prior efforts.

\textbf{Estimating approximate incentive compatibility}
\citet{balcanEstimatingApproximateIncentive2019} introduce techniques to estimate the proximity of a mechanism to being incentive compatible, specifically addressing the utility loss associated with truthful strategies. By analyzing samples from agents' type distributions, their method evaluates potential utility gains from misreporting types, leveraging finite subsets of the type space. The work provides PAC-guarantees for the approximation of incentive compatibility, utilizing the pseudo-dimension and dispersion of utility functions, and applies these techniques across a variety of common auction formats. The paper received the Exemplary AI Track Paper award at the \textit{ACM Conference on Economics and Computation (EC)} in 2019. Building on these results, we derive sampling-based error bounds and extend them in two significant directions: first, by accommodating strategic bidding, thereby determining the utility loss for strategies beyond truthful bidding, and second, by offering guarantees for interdependent prior distributions as well.

\textbf{Verification via game abstraction}
Game abstraction is a key technique for solving large imperfect-information games~\citep{Shi00:Abstraction,Billings03:Approximating,Gilpin06:Competitive}, and has led to breakthroughs such as superhuman AIs for two-player limit Texas hold'em~\cite{Bowling15:Heads}, two-player no-limit Texas hold'em~\cite{Brown18:Superhuman}, and multiplayer no-limit Texas hold'em~\cite{Brown19:Superhuman}. The basic idea is that the game is automatically abstracted into a smaller game, then the smaller game is solved for (approximate) equilibrium, and then the strategies are mapped back into the original game. However, most game abstraction techniques do not yield guarantees for equilibrium approximation in the original game~\citep{waughAbstractionPathologiesExtensive2009}.
\citet{Gilpin07:Lossless} developed a lossless abstraction technique for games with finite actions and finite states that yields an exact equilibrium in the original game, but the abstracted game to be solved is only about two orders of magnitude smaller than the original game, so that does not scale to very large games. More recently, game abstraction techniques that can abstract more and still yield a provably approximate equilibrium in the original game have been developed~\citep{Sandholm12:Lossy, Kroer14:Extensive, Kroer16:Imperfect, Kroer18:Unified}. Some game abstraction work has focused on games with continuous actions~\citep{Kroer15:Discretization}. However, these models typically are not rich enough to model a Bayesian game with continuous types and actions, and have not yielded techniques for verifying approximate equilibrium in auctions.

\citet{bichlerComputingBayesNash2023} perform an abstraction by discretizing the valuation and bidding spaces to compute and verify equilibrium using distributional strategies, providing theoretical guarantees that the abstraction error can be controlled in the case of single-item auctions. Their verification results assume explicit access to both the joint and marginal prior density functions, allowing for querying at specific points and integrating over cells of their discretization. In contrast, our results only assume access to the prior distribution through sampling. Additionally, our findings are also applicable to multi-unit auctions and combinatorial auctions.

\citet{pierothEquilibriumComputationMultiStage2023} offer a verification method using a limiting argument applicable to sequential games with continuous observation and action spaces. However, their assumption of continuous utility functions means that their results do not directly apply to auctions. Instead, they rely on a game abstraction strategy that involves smoothing the allocation and price functions, drawing from the work of \citet{kohringEnablingFirstOrderGradientBased2023}. They demonstrate that the abstraction error can be controlled for single-unit auctions with independent prior distributions. In contrast to their work, we provide explicit bounds that can be computed for a specific setting and sample size. Additionally, our results apply to multi-unit and combinatorial auctions, and with interdependent priors.

\textbf{Verification methods in full auction games}
\citet{timbersApproximateExploitabilityLearning2020} propose a reinforcement learning-based method to estimate a \emph{lower} bound on the maximum utility loss. While this can provide valuable insight into potential gains from deviation, it does not verify whether the candidate strategy profile is an approximate equilibrium. On the other hand, the work by \citet{bosshardComputingBayesNashEquilibria2020} introduces a verification method for approximate equilibrium strategies in combinatorial auctions with independent prior distributions. Similar to our approach, they approximate the utility loss at a finite number of grid points. They employ a Monte-Carlo sampling method to estimate the expected utility and use a local search algorithm to estimate the best response for each grid point. Additionally, they exploit a convexity property of the best-response \textit{ex interim} utility to provide an upper bound of the utility loss for all valuations between the grid points. However, their theoretical analysis does not account for the approximation errors introduced by the sampling procedure and the best-response approximation. In contrast, our error bounds encompass all approximations performed. Furthermore, their analysis is restricted to auctions with independent prior distributions.

\textbf{\textit{Ex post} incentive-compatible mechanism design via deep learning}
\begin{comment}
	- the field of automated mechanism design often aims to compute an optimal mechanism that is incentive compatible for a given setting with the assumption that the auctioneer knows the prior distribution\\
	- earlier methods focused on integer programming techniques, so that a certain level of incentive compatibility could be verified~\citep{bibid}\\
	- however, this approach is inherently restricted in scalability.\\
\end{comment}
In recent years, deep learning approaches to design auction mechanisms have received significant attention~\citep{golowichDeepLearningMultifacility2018, feng2018deep, duettingOptimalAuctionsDeep2019}. These efforts aim to design mechanisms that are nearly incentive compatible by incorporating constraints into the deep learning optimization problem. These constraints enforce the mechanism to be \textit{ex post} incentive compatible over a set of buyer values sampled from the prior distribution. Essentially, they seek to identify a mechanism where a bidder has approximately no incentive to conceal its valuation, regardless of the reported valuations of the opponents—a property that does not hold for most mechanisms used in practice.
\citet{duettingOptimalAuctionsDeep2019} offer a concentration bound to empirically assess the violation of incentive compatibility. However, this bound presumes that the \textit{ex post} violation can be precisely determined, an assumption not met by their methodology.~\citep{curryCertifyingStrategyproofAuction2020} address this issue by linearizing the learned neural network, effectively reducing the problem to an integer program that allows for an accurate estimation of the error. \citet{curryDifferentiableEconomicsRandomized2023} use deep learning to learn auction mechanisms within randomized affine maximizer auctions, a class within which each mechanism is exactly incentive compatible.

The concept of \textit{ex post} incentive compatibility is a worst-case, distribution-independent notion focused on the utility gain from truthful bidding. This contrasts with our objectives, as we aim to provide \textit{ex interim} and \textit{ex ante} guarantees, where agents have no incentive to deviate from their current strategy---which might not be truthful---averaged over the opponents' type distribution.

\section{Preliminaries}
This section introduces the formal model and results from learning theory that are useful for our purposes.

\subsection{The model}

We model an auction as a Bayesian game $G=\left(n, \mathcal{A}, \Theta, \mathcal{O}, u, F \right)$. Here $n \in \mathbb{N}$ denotes the number of agents.
$F$ denotes an atomless prior distribution over the agents' observations $\mathcal{O} = \bigtimes_{i \in [n]}\mathcal{O}_i$ and valuations $\Theta = \bigtimes_{i \in [n]}\Theta_i$, and is assumed to be common knowledge.
We denote its marginals by $F_{\theta_i}, F_{o_i}$, etc.; its conditionals by $F_{\theta_i|o_i}$, etc. 
An agent $i$ receives its private observation $o_i \in \mathcal{O}_i$, and chooses an action or bid $b_i \in \mathcal{A}_i$ based on it. The joint bidding space is denoted by $\mathcal{A} = \bigtimes_{i \in [n]}\mathcal{A}_i$. The sets $\Theta_i$ denote an agent's ``true'', but possibly unobserved valuation. This formulation allows to model interdependencies and correlations beyond purely private or common values. The vector $u = (u_1, \dots, u_n)$ describes the individual (\textit{ex post}) utility functions $u_i: \Theta_i \times \mathcal{A} \rightarrow \mathbb{R}$ that map a valuation $\theta_i \in \Theta_i$ and bid profile $b \in \mathcal{A}$ to a game outcome for each agent. The game consists of three distinct stages. During the \emph{ex-ante} stage, that is, before the game, agents have only knowledge about $F$. In the \textit{ex interim} stage, each agent observes $o_i$ which provides information about its valuation $\theta_i$. After submitting a bid $b_i$, an agent receives the ex post information about the game outcome $u_i(\theta_i, b)$.
In the \textit{ex ante} stage, an agent needs to reason about a strategy $\beta_i:\mathcal{O}_i \rightarrow \mathcal{A}_i$ that maps observations to bids.
We denote agent $i$'s pure strategy space by $\Sigma_i = \left\{\giventhat{\beta_i}{\beta_i: \mathcal{O}_i \rightarrow \mathcal{A}_i}\right\}$. The joint strategy space is then denoted by $\Sigma = \prod_{i \in [n]}\Sigma_i$.

In this work, we are particularly concerned with the agents' bidding distributions. That is, the distribution of bids $\beta_i(o_i^\prime)$ for an agent $i$, where $o_i^\prime \sim F_{o_i}$ and $\beta_i \in \Sigma_i$. We denote the distribution with mapped bids under strategy profiles $\beta_i, \beta_{-i}, \beta$ by $F^{\beta_i}, F^{\beta_{-i}}$, and $F^\beta$, respectively.
We define the \textit{ex interim} utility for agent $i$ and observation $o_i$ as
\begin{align} \label{equ:interim-utility}
	\hat{u}_{i}(o_i, b_i, \beta_{-i}) := \mathbb{E}_{\theta_i, o_{-i}|o_i} \left[ u_{i}\left(\theta_i, b_i, \beta_{-i}(o_{-i}) \right) \right],
\end{align}
and the \textit{ex ante} utility as
\begin{align}\label{equ:ex-ante-utility}
	\tilde{u}_{i}(\beta_i, \beta_{-i}) := \mathbb{E}_{o_i} \left[\hat{u}_{i}(o_i, \beta_i(o_i), \beta_{-i}) \right].
\end{align}

We focus on sealed-bid auctions involving $m$ distinct items. In combinatorial auctions, this results in a set $\mathcal{K}$ representing all possible bundles, with valuation and action spaces of size $ |\mathcal{K}| = 2^m $. %Under private values, each bidder's observation \( o_i \) matches their true valuation \( v_i \). However, in common value scenarios, a hidden constant \( v_c \), equal for all \( n \) bidders, serves as a baseline for potentially imprecise observations \( o_i \). It's also feasible to have mixed valuation settings.
An auction's outcome, given a bid profile $b$, is determined by an auction \emph{mechanism} $\text{M}$ that decides on two things: the allocation $x = x(b) = (x_1, \ldots, x_n)$ with $x_i \in \{0, 1\}^{|K|}$, dividing the $m$ items among bidders, and the price vector $p(b) \in \mathbb{R}^n$, indicating the cost for each bidder to claim their items.
%Technically, allocations can be seen as one-hot-encoded vectors $x_i \in \{0, 1\}^{|K|}$. 
When considering a specific mechanism $\text{M}$, we denote the respective utility function for bidder $i$ by $u_{i, \text{M}}$.
In a typical risk-neutral model, bidders' utilities $u_{i, \text{M}}$ are captured by quasilinear payoff functions $u_{i, \text{M}}(\theta_i, b) = x_i(b) \cdot \theta_i - p_i(b)$. %, reflecting the value a bidder places on her items minus their cost.
%We assume all functions to be measurable over suitable $\sigma$-algebras. %and denote with $\Delta(X)$ the set of countably additive probability measures on the measurable subsets of $X$. For measurable spaces $X$ and $Y$, a mapping $f: Y \rightarrow \Delta(X)$ is called a transition probability if, for every measurable subset $A \subset X$, the function $f(A\,|\,\argdot): Y \rightarrow \mathbb{R}$ is measurable.
%\FP{There may be some subtle problems that may arise due to measurability. If everything is Borel-measurable, it should be fine. There may be some weird case, where we need to resort to another measure. But I do not think so.}
Furthermore, we assume\footnote{Our error bound increases by a multiplicative factor of $H$ if the range of utility functions is $[-H, H]$ instead of $[-1, 1]$.} the utility functions $(u_{1, \text{M}}, \dots, u_{n, \text{M}} )$ map to the bounded interval $[-1, 1]$.
%The boundedness of the utility functions arise naturally if one assumes that the valuation spaces $\Theta_i$ are bounded. Then, boundedness of the utility functions are guaranteed if the agents cannot overbid (their maximum valuation) and the mechanism is individually rational. That is a common assumption on rationality taken in most works on auctions (cite). Our error bound increases by a multiplicative factor of $H$ if the range of utility functions is $[-H, H]$ instead of $[-1, 1]$.
We are interested in the game-theoretic solution concept of an approximate Bayesian Nash Equilibrium  (BNE)~\citep{harsanyiGamesIncompleteInformation1967}.
\begin{definition}[$\varepsilon$-Bayesian Nash equilibrium] \label{def:eps-Nash-equilibrium}
	Let $\text{M}$ be a mechanism and $G=\left(n, \mathcal{A}, \Theta, \mathcal{O}, u_{i, \text{M}}, F \right)$ a corresponding Bayesian game. Let $\varepsilon \geq0$, then, a strategy profile $\left(\beta_i^*, \beta_{-i}^*\right)$ is an \textit{ex ante} $\varepsilon$-BNE if, for all $i \in [n]$,
	$\sup_{\beta_i^\prime \in \Sigma_i} \tilde{u}_{i, \text{M}}(\beta_i^\prime, \beta_{-i}^*) - \tilde{u}_{i, \text{M}}(\beta_i^*, \beta_{-i}^*) \leq \varepsilon$.
	A so-called \textit{ex interim} $\varepsilon$-BNE is given if, for all $i \in [n]$ and $o_i \in \mathcal{O}_i$,
	$\sup_{b_i \in \mathcal{A}_i} \hat{u}_{i, \text{M}}(o_i, b_i, \beta_{-i}^*) - \hat{u}_{i, \text{M}}(o_i, \beta_i^*(o_i), \beta_{-i}^*) \leq \varepsilon$.
\end{definition}
Clearly, if we have an \textit{ex interim} $\varepsilon$-BNE, then this also constitutes an \textit{ex ante} $\varepsilon$-BNE. This may not hold the other way around, as it is only guaranteed that the utility gained through deviating to another strategy is bounded by $\varepsilon$ in expectation. For some observations $o_i$, it may be strictly larger.

The \textit{ex ante} \emph{utility loss} is a metric to measure the loss of an agent by playing $\beta_{i}$ instead of a best-response to the opponents' strategies $\beta_{-i}$~\citep{srinivasanActorcriticPolicyOptimization2018, Brown19:Deep}. It expresses the distance to an approximate \textit{ex ante} equilibrium and is defined as $\tilde{\ell}_i(\beta_i, \beta_{-i}) := \sup_{\beta^{\prime}_i \in \Sigma_i} \tilde{u}_{i, \text{M}}(\beta^{\prime}_i, \beta_{-i}) - \tilde{u}_{i, \text{M}}(\beta_i, \beta_{-i})$.
Similarly, the \textit{ex interim} utility loss measures the loss of agent $i$ for observation $o_i \in \mathcal{O}_i$ of playing $b_i$ instead of an \textit{ex interim} best-response to $\beta_{-i}$. It is given by
$\hat{\ell}_i(o_{i}, b_i, \beta_{-i}) := \sup_{b_{i}^\prime \in \mathcal{A}_i} \hat{u}_{i, \text{M}}(o_i, b_{i}^\prime, \beta_{-i}) - \hat{u}_{i, \text{M}}(o_{i}, b_i, \beta_{-i})$.
%By Definition \ref{def:eps-Nash-equilibrium}, a strategy profile $\beta$ is an $\varepsilon$-BNE if and only if $\tilde{\ell}_i(\beta_i, \beta_{-i}) \leq \varepsilon$ for all agents $i$. Similarly, $\beta$ is an interim-$\varepsilon$-BNE if and only if $\sup_{o_i \in \mathcal{O}_i} \hat{\ell}_i(o_{i}, \beta_i(o_i), \beta_{-i})\leq \varepsilon$ for all agents $i$.

We assume access to a dataset of independent samples from the unknown prior distribution $F$ in the following form. We sample a dataset $\mathcal{D}$, comprising \textit{ex interim} and \textit{ex post} data for each agent. $\mathcal{D}$ consists of $N \in \mathbb{N}$ tuples of observations and valuations,

\begin{align*}
	\mathcal{D} = \left\{\giventhat{\left(o^{(j)}, \theta^{(j)} \right) = \left((o_1^{(j)}, \dots o_n^{(j)}), (\theta_1^{(j)}, \dots, \theta_n^{(j)})\right)}{o^{(j)} \in \mathcal{O}, \theta^{(j)} \in \Theta \text{ for } 1\leq j \leq N} \right\}.
\end{align*}

Either by accessing the strategy profile $\beta$ or by observing the bids for each data point in $\mathcal{D}$, we obtain the full dataset of observations, valuations, and bids. This dataset is denoted by
\begin{align*}
	\mathcal{D}^\beta := \left\{\giventhat{\left(o^{(j)}, \beta(o^{(j)}), \theta^{(j)} \right)}{\left(o^{(j)}, \theta^{(j)} \right) \in \mathcal{D}} \right\}.
\end{align*}

\subsection{Concentration bounds from learning theory}
We introduce a distribution-independent concentration bound that allows us to approximate the expected utilities by an empirical mean over $\mathcal{D}^{\beta}$. It is grounded in the learning theoretic concept of the pseudo-dimension. This concept captures the inherent complexity of a function class, essentially reflecting how challenging it is to learn. The pseudo-dimension is defined as follows:
\begin{definition}[\citet{DBLP:books/daglib/0034861}] \label{def:pseudo-dimension}
	Let $\mathcal{F} \subset \{\giventhat{f: \mathcal{X} \rightarrow [-1, 1]}{f \text{ measurable.}} \}$ be an abstract class of functions. Further, let $\mathcal{S} = \{x^{(1)}, \dots, x^{(N)} \} \subset \mathcal{X}$ and $\{z^{(1)}, \dots, z^{(N)} \} \subset [-1, 1]$ be a set of \emph{targets}. We say that $\{z^{(1)}, \dots, z^{(N)} \}$ witness the shattering of $\mathcal{S}$ by $\mathcal{F}$ if for all subsets $T \subset \mathcal{S}$, there exists some function $f_T \in \mathcal{F}$ such that for all $x^{(j)} \in T$, $f_T(x^{(j)}) \leq z^{(j)}$ and for all $x^{(j)} \notin T$, $f_T(x^{(j)}) > z^{(j)}$. If there exists some vector $\mathbf{z} \in [-1, 1]^N$ that witnesses the shattering of $\mathcal{S}$ by $\mathcal{F}$, then we say that $\mathcal{S}$ is \emph{shatterable} by $\mathcal{F}$. Finally, the pseudo-dimension of $\mathcal{F}$, denoted by $\text{Pdim}(\mathcal{F})$, is the size of the largest set that is shatterable by $\mathcal{F}$. 
\end{definition}
A standard result for an abstract generalization bound is provided by the next theorem.
\begin{theorem}[{\citet[Theorem~10.6]{DBLP:books/daglib/0034861}}] \label{thm:pollard-pac-bound-general-distribution}
	Let $\Phi$ be a distribution over $\mathcal{X}$ and $\mathcal{F} \subset \{\giventhat{f: \mathcal{X} \rightarrow [-1, 1]}{f \text{ measurable.}} \}$. Set $d=\text{Pdim}(\mathcal{F})$, then, with probability $1-\delta$ over a draw $x^{(1)}, \dots, x^{(N)} \sim \Phi$, for all $f \in \mathcal{F}$, it holds that
	\begin{align*}
		\abs{\frac{1}{N} \sum_{j=1}^{N}f(x^{(j)}) - \mathbb{E}_{x\sim \Phi}\left[f(x)\right] } \leq 2\sqrt{\frac{2d}{N} \log\left( \frac{e N}{d}\right)} + \sqrt{\frac{2}{N} \log\left(\frac{1}{\delta} \right)}.
	\end{align*}
\end{theorem}

\section{Challenges for sampling-based equilibrium verification} \label{sec:a-data-drive-approach-for-equilibrium-verification-bounds}
In this section, we outline our general approach and discuss some of the key challenges that must be overcome to articulate a statement of the following nature:
For a strategy profile $\beta=(\beta_i, \beta_{-i})$, with probability $1 - \delta$ over the draw of the dataset $\mathcal{D}^\beta$, for any agent $i \in [n]$, one can guarantee, (1) for all observations $o_i \in \mathcal{O}_i$, $\hat{\ell}_i(o_i, \beta_i(o_i), \beta_{-i}) \leq \hat{\varepsilon}$ or (2) $\tilde{\ell}_i(\beta_i, \beta_{-i}) \leq \tilde{\varepsilon}$.

We want to give upper bounds $\hat{\varepsilon}$ or $\tilde{\varepsilon}$ that are as tight as possible to the true values of the utility losses.
However, computing the utility losses $\hat{\ell}_i$ and $\tilde{\ell}_i$ is intractable in general. This difficulty arises from two major challenges. First, one cannot evaluate the expected utilities for even a single instance due to the potentially intractable nature of the integrals involved. Second, computing the best-response utilities requires a search over an infinite space.

To overcome these obstacles, our approach is twofold. We approximate the integrals by calculating the empirical mean of the \textit{ex post} utility $u_{i, \text{M}}$ over suitable subsets of the dataset $\mathcal{D}^\beta$. This process is referred to as the \emph{simulation step}. 
To address the issue of searching over an infinite space for the best-response utilities, we constrain this search to a finite set, a method we denote as the \emph{discretization step}.
To this end, let $w>0$. We then consider a so-called $w$-grid $\mathcal{G}_w \subset \mathcal{A}_i$.
That is, $\mathcal{G}_w$ is a finite set such that for every $p \in A$ there exists a $p^\prime \in \mathcal{G}_w$ such that $\norm{p - p^\prime}_1 \leq w$.

We illustrate our approach by detailing the steps involved in estimating agent $i$'s \textit{ex interim} utility loss $\hat{\ell}_i(o_i, \beta_i(o_i), \beta{-i})$ for an observation $o_i$ and strategy profile $\beta$ through an explicit example.

\begin{example} \label{exa:two-agents-fpsb}
	Consider a first-price single-item auction with two bidders and independent priors. The two bidders receive their true valuations as observations, that is, $o_i = \theta_i$. We set $\Theta_i = \mathcal{A}_i=[0, 1]$.
	The utility function for agent $1$ is then given by $u_{1, \text{M}}(\theta_1, b_1, b_{2}) = \mathds{1}_{\left\{b_1 > b_{2} \right\}}(\theta_1 - b_1)$.
	The problem of determining the \textit{ex interim} utility loss for a valuation $\theta_1$ for agent $1$ is the optimization problem
	\begin{align} \label{equ:example-fpsb-interim-br-problem}
		\hat{\ell}_1(\theta_1, \beta_1(\theta_1), \beta_2) = \sup_{b_1 \in [0, 1]} \hat{u}_{1, \text{M}}(\theta_1, b_1, \beta_2) - \hat{u}_{1, \text{M}}(\theta_1, \beta_1(\theta_1), \beta_2)
	\end{align}
	where agent $2$ bids according to $\beta_2$.
	Further, consider the dataset of samples of bid queries $\mathcal{D}^{\beta_2} = \left\{\beta_2(\theta_2^{(1)}), \dots, \beta_2(\theta_2^{(N)}) \right\}$ that can be extracted from $\mathcal{D}^{\beta}$.
	Then, the estimator that searches a best-response over a finite set $\mathcal{G}_w \subset [0, 1]$ of the empirical mean as described above is given by
	\begin{align} \label{equ:example-fpsb-empirical-mean-approximation}
		\sup_{b_1 \in \mathcal{G}_w} \frac{1}{N} \sum_{j=1}^N u_{1, \text{M}}(\theta_1, b_1, \beta_2(\theta_2^{(j)})) - u_{1, \text{M}}(\theta_1, \beta_1(\theta_1), \beta_2(\theta_2^{(j)})).
	\end{align}
\end{example}

The challenge presented in the aforementioned example centers on mitigating the estimation error between Equation~\ref{equ:example-fpsb-interim-br-problem} and its approximation through Equation~\ref{equ:example-fpsb-empirical-mean-approximation}. To address this, we propose to limit the approximation error associated with employing the empirical mean (simulation step) by applying a classic learning theoretic concentration bounds.

However, controlling the error introduced during the discretization step--namely, restricting the search to a finite set--proves more challenging due to the discontinuous nature of the utility functions. A minor variation in the bid $b_1$ can affect the allocation outcome, thereby causing abrupt shifts in agent $1$'s utility. This discontinuity poses a particular problem when working with finite precision $w$, as it might prevent agent $1$ from achieving substantial improvements in utility due to the granularity of the bid increments.
We control this by considering the concept of dispersion.
\begin{definition}[\citet{balcanDispersionDataDrivenAlgorithm2018}] \label{def:dispersion}
	Let $f_1, \ldots, f_N: \mathbb{R}^d \rightarrow \mathbb{R}$ be a set of functions where each $f_i$ is piecewise Lipschitz with respect to the $\ell_1$-norm over a partition $\mathcal{P}_i$ of $\mathbb{R}^d$. We say that $\mathcal{P}_i$ splits a set $A \subseteq \mathbb{R}^d$ if $A$ intersects with at least two sets in $\mathcal{P}_i$. The set of functions is $(w, v)$-dispersed if for every point ${p} \in \mathbb{R}^d$, the ball $\left\{{p}^{\prime} \in \mathbb{R}^d:\left\|{p}-{p}^{\prime}\right\|_1 \leq w\right\}$ is split by at most $v$ of the partitions $\mathcal{P}_1, \ldots, \mathcal{P}_N$
\end{definition}

Dispersion quantifies the number of discontinuities present within any given ball of width $w$. The larger the value of $w$ and the smaller the value of $v$, the more ``dispersed'' the discontinuities of the functions are. For a $(w, v)$-dispersed set of $N$ functions, at most $v$ jump discontinuities occur within a ball of radius $w$. Thus, within any ball of radius $w$, at least $N-v$ functions exhibit $L$-Lipschitz continuity, while at most $v$ do not.

Considering Example~\ref{exa:two-agents-fpsb}, assume the functions $u_{1, \text{M}}(\theta_i, \cdot, \beta_2(\theta_2^{(1)})), \dots, u_{1, \text{M}}(\theta_i, \cdot, \beta_2(\theta_2^{(N)}))$ are $(w, v)$-dispersed. Then, for any $b_1, b_1^\prime \in [0, 1]$ with $\norm{b_1 - b_1^\prime}_1 \leq w$, we can bound the difference by
\begin{align*}
	\abs{\frac{1}{N} \sum_{j=1}^N u_{1, \text{M}}(\theta_1, b_1, \beta_2(\theta_2^{(j)})) - u_{1, \text{M}}(\theta_1, b_1^\prime, \beta_2(\theta_2^{(j)}))} \leq \frac{N - v}{N}L_i w_i + \frac{2v}{N}\norm{u_{1, \text{M}}}_\infty.
\end{align*}
For sufficiently small $v$, the error is small. Therefore, if we can ensure that the discontinuities are sufficiently dispersed with high probability, the error from searching over $\mathcal{G}_w$ can be controlled.

The approach we have outlined aligns with the work of \citet{balcanEstimatingApproximateIncentive2019}. However, in our scenario, an agent must reason about the opponents' bid distribution rather than the prior distribution directly. Next, we discuss the additional considerations necessary for this.

\subsection{Sufficient properties for strategies to be verifiable} \label{sec:classifying-space-of-verifiable-strategies}
We address the question of identifying the kinds of bidding strategies that can be effectively verified using the approach outlined above. To achieve meaningful dispersion guarantees, we discuss specific sufficient conditions of regularity for strategies to be verifiable.

We observed that the concept of dispersion hinges on a sufficient spread of discontinuities, implying in our context that the opponents' bidding distribution $F^{\beta_{-i}}$ should not be too concentrated. Initially, we assume the prior distribution $F$ to be $\kappa$-bounded, meaning it possesses a $\kappa$-bounded density function $\phi$, that is, $\sup_x \phi(x) \leq \kappa$. However, the bidding distribution $F^{\beta_{-i}}$ may still exhibit concentration even if the prior distributions do not. We demonstrate that a bounded prior distribution remains bounded under a bidding strategy if the bidding strategy is sufficiently smooth and changes with a minimal rate over the received observation.
%This extends the work of \citet{balcanEstimatingApproximateIncentive2019}, who showed that strong dispersion guarantees can be provided if the prior distribution has a $\kappa$-bounded density and the agents play truthfully.
More specifically, we demand the bidding strategies to be bi-Lipschitz continuous.
\begin{definition}[Bi-Lipschitz function, \citet{verineExpressivityBiLipschitzNormalizing2023}] \label{def:bi-lipschitz-function}
	Let $\mathcal{X}, \mathcal{Y} \subset \mathbb{R}^m$. A bijective function $g: \mathcal{X} \rightarrow \mathcal{Y}$ is said to be $(L_g, L_{g^{-1}})$-bi-Lipschitz if $g$ is $L_g$-Lipschitz and its inverse $g^{-1}$ is $L_{g^{-1}}$-Lipschitz, that is, for all $x_1, x_2 \in \mathcal{X}$ and $y_1, y_2 \in \mathcal{Y}$
	\begin{align*}
		\norm{g(x_1) - g(x_2)} \leq L_g \norm{x_1 - x_2} \text{ and } \norm{g^{-1}(y_1) - g^{-1}(y_2)} \leq L_{g^{-1}} \norm{x_1 - x_2}.
	\end{align*}
	%This is equivalent to the following condition for all $x_1, x_2 \in \mathcal{X}$
	%\begin{align*}
	%	\frac{1}{L_{g^{-1}}} \norm{x_1 - x_2} \leq \norm{g(x_1) - g(x_2)} \leq L_g \norm{x_1- x_2}.
	%	\end{align*}
\end{definition}

The bi-Lipschitz continuity ensures that neither the function nor its inverse can change arbitrarily fast. More precisely, our approximation results in the following sections are valid for the set of bidding strategies for an agent $i$, defined as
\begin{align*}
\tilde{\Sigma}_{i} := \left\{\giventhat{\beta_{i} \in \Sigma_i}{\beta_{i} \text{ is continuously differentiable and bi-Lipschitz continuous}}\right\}.
\end{align*}
%As usual, we define $\tilde{\Sigma}_{-i} = \bigtimes_{j \in [n]\setminus \{i\}} \tilde{\Sigma}_{j}$ and $\tilde{\Sigma} = \bigtimes_{i \in [n]} \tilde{\Sigma}_{i}$. 
For a strategy $\beta_i \in \tilde{\Sigma}_i$, we denote the bi-Lipschitz constants by $L_{\beta_i}, L_{\beta_i^{-1}}$. Further, define $L_{\beta^{-1}_{\text{max}}} := \max_{t \in [n]} L_{\beta^{-1}_{t}}$.
We leverage the properties of bi-Lipschitz functions to bound the density function of the bidding distribution.

\begin{theorem}\label{thm:bounded-prior-density-remains-bounded-under-bi-Lipschitz}
Let $\mathcal{O}_i \subset \mathbb{R}^m$ for all $i \in [n]$. Denote with $\phi_{F_{o_i}}$ and $\phi_{F_{o_i, o_j}}$ the density functions for the marginal prior distributions $F_{o_i}$ and $F_{o_i, o_j}$ for any $i, j \in [n]$.
Further assume that $\phi_{F_{o_i}}$ and $\phi_{F_{o_i, o_j}}$ are $\kappa$-bounded density functions for some $\kappa > 0$. Further, let $(\beta_i, \beta_{-i}) \in \tilde{\Sigma}$ be a strategy profile of bi-Lipschitz continuous bidding strategies. Then, the probability density functions of the bidding distributions $F^{\beta_i}_{o_i}$ and $F^{\beta_i, \beta_j}_{o_i, o_j}$ satisfy
\begin{align*}
	\sup_{b_{i} \in \beta_{i}(\mathcal{O}_{i})} \phi_{F^{\beta_{i}}_{o_{i}}}(b_{i}) &\leq \kappa \cdot L_{\beta_i^{-1}}^m\\
	\sup_{(b_{i}, b_j) \in \beta_{i}(\mathcal{O}_{i}) \times \beta_{j}(\mathcal{O}_{j})} \phi_{F^{\beta_{i}, \beta_j}_{o_{i}, o_j}}(b_{i}, b_j) &\leq \kappa \cdot L_{\beta_i^{-1}}^m \cdot L_{\beta_j^{-1}}^m.
\end{align*}
\end{theorem}

\begin{sproof}
By the definition of bi-Lipschitz continuity, the function $\beta_i : \mathcal{O}_i \rightarrow \beta_i \left(\mathcal{O}_i\right)$ is invertible for any $i \in [n]$. We perform a change of variables and get for $b_i \in \beta_i\left(\mathcal{O}_i \right)$
\begin{align*}
	\phi_{F^{\beta_{i}}_{o_{i}}}(b_{i}) = \phi_{F_{o_i}} \left(\beta_i^{-1}(b_i) \right) \cdot \abs{ \det(\mathcal{J} \beta_i^{-1}(b_i))} \leq \kappa \cdot L_{\beta_i^{-1}}^m,
\end{align*}
where $m$ denotes the dimension of $\mathcal{O}_i$. We used a well-known bound on a bi-Lipschitz mapping's Jacobian determinant in the second step. The case of two agents $i, j \in [n]$ is similar and additionally leverages a property of the determinant for block matrices. The full proof is in Appendix~\ref{sec:proofs-appendix-sampling-based-approach-for-verification}.
\end{sproof}

%- Theorem \ref{thm:bounded-prior-density-remains-bounded-under-bi-Lipschitz} states that a $\kappa$ bounded prior density function translates to a $K\cdot \kappa$ bidding density function, where constant $K$ is given by the Lipschitz functions of the inverse bidding strategies.\\
%- Leveraging this theorem allows us extend the dispersion guarantees by \citet{balcanEstimatingApproximateIncentive2019} to settings with strategic bidders with only small adaptations.\\

To illustrate the effect, Figure \ref{fig:exemplary-density-function-transformations} shows how the density function of a beta-distribution $\text{Beta}(2, 5)$ is transformed under different strategies. Linear transformations such as $\theta_i \mapsto \frac{1}{2} \theta_i$ restrict the bidding space to $[0, \frac{1}{2}]$, which compresses the density and leads to a higher maximum value.
%Conversely, the mapping $\theta_i \mapsto \sqrt{\theta_i}$ results in a more evenly distributed density, decreasing the maximum value. 
The mapping $\theta_i \mapsto \theta_i^2$ leads to an unbounded density, which can occur because its inverse is \emph{not} Lipschitz continuous. However, the bidding strategy under the mapping $\theta_i \mapsto \theta_i^{3/2}$ remains bounded, even though the mapping itself is not bi-Lipschitz continuous. The prior density assigns a high mass to valuations close to zero, but the strategy increases rapidly enough to redistribute a significant amount of mass away from zero.
These examples underscore that while our assumptions provide a sufficient condition for the verification of bidding strategies, our findings may extend to a broader class of bidding strategies and prior distributions.
\begin{figure}[htbp]
\centering
\includegraphics[width=.7\textwidth]{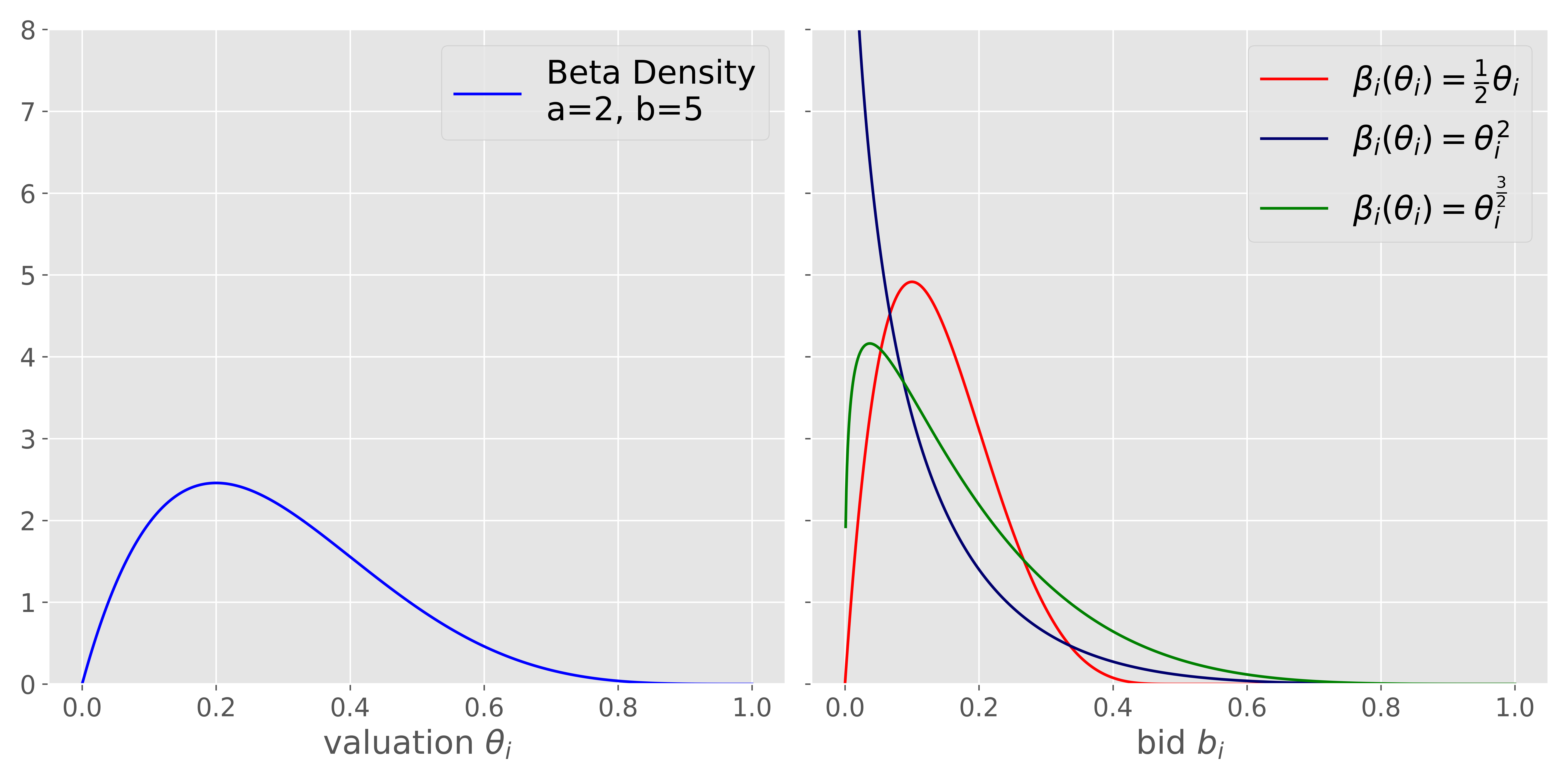}
\caption{A beta-distributed density function of agent $i$'s valuations (left) and the corresponding bidding density functions under different strategies $\beta_i$ (right).}
\label{fig:exemplary-density-function-transformations}
\end{figure}

%- first, due to simplicity, we assume that the bidding strategies are continuously differentiable, otherwise, one would need to consider various additional measure-theoretic technicalities.\\
This raises the question of how restrictive it is to consider only strategies from $\tilde{\Sigma}$ for verification. Continuously differentiable is a common assumption met by most function approximation techniques, such as neural networks, making this restriction relatively mild. The restriction to bi-Lipschitz continuous strategies is more stringent.
In the one-dimensional case, this translates to the bid strategy being strictly monotonic. Monotonicity is a very common assumption for strategies in auctions, where bids are assumed not to decrease with a rising valuation~\citep{renyNashEquilibriumDiscontinuous2020}. In the standard model that we consider here, strict monotonicity remains a reasonable assumption. Furthermore, known equilibrium strategy profiles commonly fall within this set $\tilde{\Sigma}$~\citep{krishna2009auction}.
%- \FP{Add discussion about training neural networks to be bi-Lipschitz.}\\
Nevertheless, under model assumptions such as reserve prices or budget constraints, bidders may resort to constant bidding for a range of observations, so that their strategies do not lie within $\tilde{\Sigma}$.

\section{Verifying approximate equilibrium under independent priors} \label{sec:independent-prior-statements}
In this section, we give guarantees for the maximum \textit{ex interim} utility loss for a given strategy profile $\beta$ under the common assumption of independent prior valuations~\citep{lubinApproximateStrategyproofness2012, azevedoStrategyproofnessLarge2019, hartApproximateRevenueMaximization2012, chawlaAuctionsUniqueEquilibria2013}. Specifically, we consider an auction $G=\left(n, \mathcal{A}, \Theta, \mathcal{O}, u, F \right)$, simplifying several aspects. 
For all agents $i \in [n]$, we assume $\mathcal{A}_i = \mathcal{O}_i = \Theta_i$, and a drawn observation $o_i$ equals the true valuation $\theta_i$, allowing us to omit the observation space entirely.
Furthermore, the prior distribution simplifies to a product distribution over the valuation spaces, that means, $F = \prod_{i \in [n]} F_{\theta_i}$.
A bidding strategy $\beta_i$ is then a mapping from agent $i$'s valuation space onto itself, $\beta_i: \Theta_i \rightarrow \Theta_i$.

In the \textit{ex interim} stage, agent $i$ must reason about the opponents' bid distribution $F^{\beta_{-i}}$ given its valuation $\theta_i$. We organize the dataset $\mathcal{D}^{\beta}$ as follows: denote with $\mathcal{D}^{\beta_{-i}} := \left\{\beta_{-i}(\theta_{-i}^{(1)}), \dots, \beta_{-i}(\theta_{-i}^{(N)}) \right\}$ the dataset of agent $i$'s opponents' bids. Then $\mathcal{D}^{\beta_{-i}}$ consists of i.i.d. samples from $F^{\beta_{-i}}$.

\subsection{A sampling-based bound on the \textit{ex interim} utility loss via grid search}
We start with the \emph{sampling step} of our approach, presenting a result to estimate the \textit{ex interim} utility loss by considering the empirical mean instead of the expectation. We use a classical PAC-learning result (Theorem \ref{thm:pollard-pac-bound-general-distribution}) to bound the error incurred by taking the empirical mean compared to evaluating the integral, demonstrating that this error converges towards zero as the number of samples $N$ increases.
For mechanism $\text{M}$ and agent $i$, define the class of functions that map opponent bids to utility by $\hat{\mathcal{F}}_{i, \text{M}} := \left\{\giventhat{u_{i, \text{M}}(\theta_i, \hat{\theta}_i, \argdot): \Theta_{-i} \rightarrow [-1, 1]}{\theta_i, \hat{\theta_i} \in \Theta_i} \right\}$.

\begin{theorem} \label{thm:pseudo-dim-guarantee-independent-priors-interim-utility}
Let $\delta > 0$, $\text{M}$ be a mechanism, and $\beta \in \Sigma$. Then, it holds with probability $1-\delta$ for all agents $i \in [n]$ over the draw of datasets $\mathcal{D}^{\beta_{-1}}, \dots, \mathcal{D}^{\beta_{-n}}$ of valuation-bid queries,
\begin{align*}
	&\sup_{\theta_i \in \Theta_i} \hat{\ell}_i(\theta_i, \beta_i(\theta_i), \beta_{-i}) = \sup_{\theta_i, \hat{\theta}_i \in \Theta_i} \hat{u}_{i, \text{M}}(\theta_i, \hat{\theta}_i, \beta_{-i}) - \hat{u}_{i, \text{M}}(\theta_i, \beta_i(\theta_i), \beta_{-i}) \\
	&\leq \sup_{\theta_i, \hat{\theta}_i \in \Theta_i} \frac{1}{N} \sum_{j=1}^N u_{i, \text{M}}(\theta_i, \hat{\theta}_i, \beta_{-i}(\theta_{-i}^{(j)})) - u_{i, \text{M}}(\theta_i, \beta_i(\theta_i), \beta_{-i}(\theta_{-i}^{(j)})) + \hat{\varepsilon}_{i, \text{Pdim}}(N, \delta),\\
	&\mbox{ where } \hat{\varepsilon}_{i, \text{Pdim}}(N, \delta) := 4\sqrt{\frac{2d_i}{N} \log\left( \frac{e N}{d_i}\right)} + 2\sqrt{\frac{2}{N} \log\left(\frac{2n}{\delta} \right)}, \mbox{ and } d_i=\text{Pdim}(\hat{\mathcal{F}}_{i, \text{M}}).
\end{align*}
\end{theorem}

\begin{sproof}
Fix an arbitrary agent $i \in [n]$. Then we have with $\mathcal{D}^{\beta_{-i}} := \left\{\beta_{-i}(\theta_{-i}^{(1)}), \dots, \beta_{-i}(\theta_{-i}^{(N)}) \right\}$ that $\beta_{-i}(\theta_{-i}^{(j)}) \sim F^{\beta_{-i}}_{\theta_{-i}}$ is i.i.d. for $1 \leq j \leq N$. 
Therefore, by applying Theorem \ref{thm:pollard-pac-bound-general-distribution}, we have with probability at least $1 - \frac{\delta}{2}$ for all $u_{i, \text{M}}(\theta_{i}, \hat{\theta}_i, \argdot) \in \hat{\mathcal{F}}_{i, \text{M}}$ that
\begin{align*}
	\abs{\frac{1}{N} \sum_{j=1}^{N}u_{i, \text{M}}(\theta_{i}, \hat{\theta}_i, \beta_{-i}(\theta_{-i}^{(j)})) - \mathbb{E}_{\beta_{-i}(\theta_{-i})\sim F^{\beta_{-i}}}\left[u_{i, \text{M}}(\theta_{i}, \hat{\theta}_i, \beta_{-i}(\theta_{-i}))\right] } \leq \frac{1}{2}\hat{\varepsilon}_{i, \text{Pdim}}(N, n\delta).
\end{align*}
We apply this to the pairs $(\theta_i, \hat{\theta_i})$ and $(\theta_i, \beta_i(\theta_i))$. A union bound over the agents finishes the statement. The full proof is in Appendix~\ref{sec:appendix-independent-prior-proofs}.
\end{sproof}

The statement is similar to Theorem 3.2 from \citet{balcanEstimatingApproximateIncentive2019}. The key difference lies in the observation that one can average over the opponents' bidding distribution $F^{\beta_{-i}}$ instead of the opponents' prior distribution $F_{\theta_{-i}}$.

We proceed with the \emph{discretization step} of our procedure. For this purpose, we assume $\Theta_i = [0, 1]^m$ for some suitable $m \in \mathbb{N}$. Let $\mathcal{G}_w \subset \Theta_i$ be a $w$-grid for $w>0$, where the largest distance between any point is bounded by $w$. To bound the error incurred by restricting the search to a finite grid, we assume a certain degree of dispersion, as discussed in Section~\ref{sec:a-data-drive-approach-for-equilibrium-verification-bounds}.% More specifically, we make the following assumption.

\begin{assumption} \label{ass:general-dispersion-holds-independent-prior}
Suppose that for mechanism $\text{M}$ and each agent $i \in[n]$, there exist $L_i, w_i \in \mathbb{R}$ and a function $v_{i}: \mathbb{R} \rightarrow \mathbb{R}$, such that with probability $1-\delta$ over the draw of the $n$ sets $\mathcal{D}^{\beta_{-i}} := \left\{\beta_{-i}(\theta_{-i}^{(1)}), \dots, \beta_{-i}(\theta_{-i}^{(N)}) \right\}$, the following conditions hold:
\begin{enumerate}
	\item For any valuation $\theta_i \in[0,1]^m$, the functions $u_{i, \text{M}}\left(\theta_i, \argdot, \beta_{-i}(\theta_{-i}^{(1)})\right), \dots, u_{i, \text{M}}\left(\theta_i, \argdot, \beta_{-i}(\theta_{-i}^{(N)})\right)$ are piecewise $L_i$-Lipschitz and $\left(w_i, v_i\left(w_i \right)\right)$-dispersed.
	\item For any reported $\hat{\theta}_i \in[0,1]^m$, the functions $u_{i, \text{M}}\left(\argdot, \hat{\theta}_i, \beta_{-i}(\theta_{-i}^{(1)})\right), \dots, u_{i, \text{M}}\left(\argdot, \hat{\theta}_i, \beta_{-i}(\theta_{-i}^{(N)})\right)$ are piecewise $L_i$-Lipschitz and $\left(w_i, v_i\left(w_i \right)\right)$-dispersed.
\end{enumerate}
\end{assumption}

The constants $w_i$ and $v_i\left(w_i \right)$ will be properties resulting from the interplay of the utilized mechanism $\text{M}$, the prior distribution $F$, the opponents' strategy profile $\beta_{-i}$, and the number of drawn samples. Under the assumption that the dispersion guarantees hold, we can provide the following guarantee on the \textit{ex interim} utility loss. The full proof is in Appendix~\ref{sec:appendix-independent-prior-proofs}.

\begin{theorem}\label{thm:independent-prior-guarantee-by-dispersion-over-finite-grid}
Let $\delta > 0$ and $\text{M}$ be a mechanism. Furthermore, let $\beta \in \tilde{\Sigma}$ be a strategy profile.
Given that Assumption \ref{ass:general-dispersion-holds-independent-prior} holds for $w_i >0$, $v_i(w_i)$, and $v_i(L_{\beta_i}w_i)$, we have with probability at least $1 -  3\delta$ over the draw of the datasets $\mathcal{D}^{\beta_{-1}}, \dots, \mathcal{D}^{\beta_{-n}}$ for every agent $i \in [n]$
\begin{align*}
	& \sup_{\theta_i \in \Theta_i} \hat{\ell}_i(\theta_i, \beta_i(\theta_i), \beta_{-i}) = \sup_{\theta_i, \hat{\theta}_i \in \Theta_i} \hat{u}_{i, \text{M}}(\theta_i, \hat{\theta}_i, \beta_{-i}) - \hat{u}_{i, \text{M}}(\theta_i, \beta_i(\theta_i), \beta_{-i}) \\
	&\leq \sup_{\theta_i, \hat{\theta}_i \in \mathcal{G}_{w_i}} \frac{1}{N} \sum_{j=1}^N u_{i, \text{M}}(\theta_i, \hat{\theta}_i, \beta_{-i}(\theta_{-i}^{(j)})) - u_{i, \text{M}}(\theta_i, \beta_i(\theta_i), \beta_{-i}(\theta_{-i}^{(j)}))+\hat{\varepsilon}_i,\\
	&\mbox{ where } \ \hat{\varepsilon}_i:=4\sqrt{\frac{2d_i}{N} \log\left( \frac{e N}{d_i}\right)} + 2\sqrt{\frac{2}{N} \log\left(\frac{2n}{\delta} \right)} + 3\hat{\varepsilon}_{i, \text{disp}}(w_i) + \hat{\varepsilon}_{i, \text{disp}}(L_{\beta_i}w_i), \\
	&\mbox{ with } \ \hat{\varepsilon}_{i, \text{disp}}(x) := \frac{N - v_i\left(x \right)}{N} L_i x + \frac{2 v_i\left(x \right)}{N}, \text{ and } d_i=\operatorname{Pdim}\left(\hat{\mathcal{F}}_{i, \text{M}}\right).
\end{align*}
\end{theorem}

\begin{sproof}
Fix an agent $i \in [n]$. By the definition of dispersion, we have with probability at least $1 - \delta$,
for all $i \in [n]$, $\theta_i \in \Theta_i$, and reported valuations $\hat{\theta}_i, \hat{\theta}_i^\prime \in \Theta_i$ with $\norm{\hat{\theta}_i - \hat{\theta}_i^\prime}_1 \leq x$, we have
\begin{align} \label{thm:sproof-independent-prior-main-bidding-relation}
	\begin{split}
		\abs{ \frac{1}{N} \sum_{j=1}^N u_{i, \text{M}}(\theta_i, \hat{\theta}_i, \beta_{-i}(\theta_{-i}^{(j)})) - u_{i, \text{M}}(\theta_i, \hat{\theta}_i^\prime, \beta_{-i}(\theta_{-i}^{(j)})) }
		\leq \hat{\varepsilon}_{i, \text{disp}}(x);
	\end{split}
\end{align}
and for all $i \in [n]$, reported valuations $\hat{\theta}_i \in \Theta_i$, and $\theta_i, \theta_i^\prime \in \Theta_i$ with $\norm{\theta_i - \theta_i^\prime}_1 \leq x$, we have
\begin{align} \label{thm:sproof-independent-prior-main-valuation-relation}
	\begin{split}
		\abs{ \frac{1}{N} \sum_{j=1}^N u_{i, \text{M}}(\theta_i, \hat{\theta}_i, \beta_{-i}(\theta_{-i}^{(j)})) - u_{i, \text{M}}(\theta_i^\prime, \hat{\theta}_i, \beta_{-i}(\theta_{-i}^{(j)})) }
		\leq \hat{\varepsilon}_{i, \text{disp}}(x).
	\end{split}
\end{align}
For any $\theta_i \in \Theta_i$, there exists a grid point $p \in \mathcal{G}_{w_i}$ such that $\norm{\theta_i - p}_1 \leq w_i$ and $\norm{\beta_i(\theta_i) - \beta_i(p)}_1 \leq L_{\beta_i}w_i$. We apply Equations~\ref{thm:sproof-independent-prior-main-bidding-relation} and \ref{thm:sproof-independent-prior-main-valuation-relation} for the grid width $w_i$ and the stretched grid width $L_{\beta_i}w_i$. The statement follows with an application of Theorem~\ref{thm:pseudo-dim-guarantee-independent-priors-interim-utility} and a suitable union bound.
\end{sproof}

The above result is similar to Theorem 3.5 of \citet{balcanEstimatingApproximateIncentive2019}---which assumed truthful bidding---with the distinction that we need to ensure the dispersion of the utility functions under the opponents' bidding distribution. Additionally, it is necessary to consider the potential distortion of the grid $\mathcal{G}_{w_i}$ under the bidding strategies. %We address this by assuming a limited amount of change for a bidding strategy, as explained in Section~\ref{sec:classifying-space-of-verifiable-strategies}.

\section{Verifying approximate equilibrium under interdependent priors} \label{sec:verifying-equilibria-with-interdependent-priors}
We present the first, to our knowledge, sampling-based results to verify approximate equilibrium with interdependent prior distributions. We limit our focus to \textit{ex ante} guarantees.

In this setting, from agent $i$'s perspective, for two distinct received observations $o_{i}$ and $o_{i}^\prime$, he must consider two different conditional prior distributions $F_{\giventhat{\theta, o_{-i}}{o_i}}^{\beta_{-i}}$ and $F_{\giventhat{\theta, o_{-i}}{o_i^\prime}}^{\beta_{-i}}$. For $j \in [N]$, a sample $\left(o^{(j)}, \theta^{(j)}, \beta(o^{(j)}) \right)$ from $\mathcal{D}^\beta$ can be interpreted as a draw $\left(o_{-i}^{(j)}, \theta_i^{(j)}, \beta_{-i}(o_{-i}^{(j)}) \right) \sim F_{\giventhat{\theta, o_{-i}}{o_i^{(j)}}}^{\beta_{-i}}$. However, the probability that there is another $l \neq j$ such that $o_i^{(l)} = o_i^{(j)}$ is zero. Therefore, we cannot implement the sampling step in the same manner as we did in Section~\ref{sec:independent-prior-statements}.

We address this challenge by considering a partition $\mathcal{B}_i=\left\{B_1, \dots, B_{N_{\mathcal{B}_i}}\right\}$ of $\mathcal{O}_i$ for each agent $i \in [n]$. Denote the maximum number of elements in anz partition by $N_{\mathcal{B}_{\text{max}}} := \max_{i \in [n]} N_{\mathcal{B}_{i}}$. We demonstrate that it is sufficient to assume a constant best-response for each $B_k \in \mathcal{B}_i$ if the conditional distribution $F_{\giventhat{\theta, o_{-i}}{o_i}}$ does not vary too strongly for $o_i \in B_k$, according to an appropriate distance measure over the space of probability distributions. With this premise, we establish that one can group the samples based on ${o_i \in B_k}$. Subsequently, we present our upper bound $\tilde{\varepsilon}$ for the \textit{ex ante} guarantee by conducting the sampling and discretization step for each $B_k \in \mathcal{B}_i$.
\subsection{Bounding best-response utility differences with constant best-responses} \label{sec:restricting-br-search-to-constants}
For a $B_k$ from partition $\mathcal{B}_i$, we want to bound the error incurred when limiting bidding to a constant best response for all $o_i \in B_k$. To achieve this, it is necessary to limit distance between conditional prior distributions $F_{\giventhat{\theta_i, o{-i}}{o_i}}$ and $F_{\giventhat{\theta_i, o{-i}}{o_i^\prime}}$ according to some distance for $o_i, o_i^\prime \in B_k$. In contrast to finite-dimensional Euclidean spaces, different distance functions can induce vastly different topologies on the space of probability distributions over continuous spaces~\citep{gibbsChoosingBoundingProbability2002}. Therefore, the selection of an appropriate distance measure for this purpose is crucial.

Common choices in the machine learning literature for measuring distances between probability distributions include the Wasserstein metric $d_{\text{W}}$ (also known as the earth mover's distance or Kantorovich metric), the total variation metric $d_{\text{TV}}$, and the Kullback-Leibler divergence $d_{\text{KL}}$ (also referred to as relative entropy).
Let $\mu$ and $\nu$ denote two probability measures over agent $i$'s observation space $\mathcal{O}_i$. We have the following relationship between these distance measures:
\begin{align*}
	d_{\text{W}}(\mu, \nu) \leq \text{diam}(\mathcal{O}_i) \cdot d_{\text{TV}}(\mu, \nu) \leq \text{diam}(\mathcal{O}_i) \cdot \sqrt{\nicefrac{1}{2} \cdot d_{\text{KL}}(\mu, \nu)},
\end{align*}
where $\text{diam}(\mathcal{O}_i)$ denotes $\mathcal{O}_i$'s diameter. The above inequalities can be strict, and there are no constants so that they may hold in the other direction in general~\citep{gibbsChoosingBoundingProbability2002}.

The objective is to furnish guarantees using the weakest possible distance measure. Unfortunately, the Wasserstein metric, seems too weak to provide sufficient guarantees for discontinuous utility functions~\citep{villaniOptimalTransport2009}.
The Kullback-Leibler divergence, despite its appealing properties, can be unbounded, which poses a limitation for establishing practical guarantees. On the other hand, the total variation distance has the advantage of being upper bounded by one, making it a more suitable choice for our purposes.
Therefore, we opt for the total variation distance as the measure to base our guarantees upon.

\begin{definition}[Total variation, \citet{gibbsChoosingBoundingProbability2002}] \label{def:total-variation-distance}
	Let $\mu$ and $\nu$ be two probability measures over $\mathbb{R}^m$ and $\Lambda$ be the Borel-$\sigma$-algebra. Then the total variation distance between $\mu$ and $\nu$ is given by
	$d_{\text{TV}}(\mu, \nu) := \sup_{A \in \Lambda} \abs{\mu(A) - \nu(A)}$.
\end{definition}

\begin{comment}
	- the total variation distance is in general hard to determine\\
	- however, in our case, we assume that the prior distribution has a density function, so that we have the following well-known equality\\
	
	\begin{lemma}[{\citet[Lemma 2.1]{tsybakovIntroductionNonparametricEstimation2009}}] \label{thm:equality-total-variation-and-L1-distance}
		Let $\mu$ and $\nu$ be two probability measures over $\mathbb{R}^m$ that are absolutely continuous with regard to the Borel-measure $\lambda$. Then, the total variation distance satisfies
		\begin{align*}
			d_{\text{TV}}(\mu, \nu) = \frac{1}{2} \int_{\mathbb{R}^m} \abs{\phi_\mu(x) - \phi_\nu(x)} d \lambda(x) =  \frac{1}{2} \norm{\phi_\mu - \phi_\nu}_1,
		\end{align*}
		where $\phi_\mu$ and $\phi_\nu$ denote the probability density functions of $\mu$ and $\nu$, respectively.
		% The proof of lemma can be found here: https://stephentu.github.io/blog/probability-theory/2017/10/09/total-variation-distance-identity.html
	\end{lemma}
\end{comment}

We leverage a well-known fact that the distance between two integrals over different probability measures can be bounded by the total variation of these measures~\citep{villaniOptimalTransport2009}. This principle enables us to bound differences in the \textit{ex interim} utility function for different observations. For the sake of completeness, we provide a proof for this statement.

\begin{theorem} \label{thm:bound-distance-of-integrals-by-total-variation}
	Let $A \subset \mathbb{R}^m$ and $g: A \rightarrow \mathbb{R}$ be a bounded function. Furthermore, let $\mu$ and $\nu$ be probability measures over $A$ with density functions $\phi_{\mu}$ and $\phi_{\nu}$. Then, we have
	\begin{align*}
		\abs{\int_A g(x) d \mu(x) - \int_A g(x) d \nu(x)} \leq 2 \norm{g}_\infty \cdot d_{\text{TV}}(\mu, \nu).
	\end{align*}
\end{theorem}

\begin{proof}
	The total variation distance is equal to one half of the $L^1$-distance between the density functions \citep[Lemma 2.1]{tsybakovIntroductionNonparametricEstimation2009}, that is, 
	$d_{\text{TV}}(\mu, \nu) =  \frac{1}{2} \norm{\phi_\mu - \phi_\nu}_1$.
	Therefore, we have
	\begin{align*}
		\abs{\int_A g(x) d \mu(x) - \int_A g(x) d \nu(x)} \leq \norm{g}_\infty \int_A \abs{\phi_\mu(x) - \phi_\nu(x) } d \lambda(x) = 2 \norm{g}_\infty \cdot d_{\text{TV}}(\mu, \nu).
	\end{align*}
\end{proof}

We show next that the error incurred by assuming a constant best-response for all observations from $B_k$ can be controlled, provided the distance between conditional prior distributions $F_{\giventhat{\theta_i, o_{-i}}{o_i}}$ and $F_{\giventhat{\theta_i, o_{-i}}{o_i^\prime}}$ is small enough in terms of the total variation distance for $o_i, o_i^\prime \in B_k$. The full proof is in Appendix~\ref{sec:appendix-interdependent-proof-tv-distance-over-segments}.
\begin{theorem} \label{thm:constant-best-response-difference-bounded-over-partition}
	Let $\mathcal{B}_i = \left\{B_1, \dots, B_{N_{\mathcal{B}_i}} \right\}$ be a partition of $\mathcal{O}_i$.
	The difference between a best-response utility over function space to best-responses that are constant for every $B_k$ satisfies
	\begin{align*}
		&\sup_{\beta_i^\prime \in \Sigma_{i}} \tilde{u}_{i, \text{M}}(\beta_i^\prime, \beta_{-i}) - \sup_{b \in \mathcal{A}_i^{N_{\mathcal{B}_i}}} \tilde{u}_{i, \text{M}} \left(\sum_{k=1}^{N_{\mathcal{B}_i}} b_k \mathds{1}_{B_k}, \beta_{-i} \right) \leq 2 \sum_{k=1}^{N_{\mathcal{B}_i}} P(o_i \in B_k) \tau_{i, B_k},
	\end{align*}
	with $\tau_{i, B_k} := \sup_{\hat{o}_i, \hat{o}_i^\prime \in B_k} d_{\text{TV}}\left(F_{\giventhat{\theta_i, o_{-i}}{\hat{o}_i}}, F_{\giventhat{\theta_i, o_{-i}}{\hat{o}_i^\prime}} \right)$.
	If there exists a constant $L_{B_k} > 0$ such that $d_{\text{TV}}\left(F_{\giventhat{\theta_i, o_{-i}}{\hat{o}_i}}, F_{\giventhat{\theta_i, o_{-i}}{\hat{o}_i^\prime}} \right) \leq L_{B_k} \norm{o_i - o_i^\prime}$ for $o_i, o_i^\prime \in B_k$, then $\tau_{i, B_k} \leq L_{B_k} \text{diam}(B_k)$, where $\text{diam}(B_k)$ denotes $B_k$'s diameter.
\end{theorem}
\begin{sproof}
	Fix an agent $i \in [n]$ and $B_k \in \mathcal{B}_i$. We leverage Teorem~\ref{thm:bound-distance-of-integrals-by-total-variation} to establish a bound of the interim utilities for any $o_i, o_i^\prime \in B_k$
	\begin{align*}
		\abs{ \sup_{b_i \in \mathcal{A}_i}\hat{u}_{i, \text{M}}(o_i, b_i, \beta_{-i}) - \sup_{b_i^\prime \in \mathcal{A}_i} \hat{u}_{i, \text{M}}(o_i, b_i^\prime, \beta_{-i})} \leq 2 \norm{u_{i, \text{M}}}_\infty d_{\text{TV}}\left(F_{\giventhat{\theta_i, o_{-i}}{\hat{o}_i}}, F_{\giventhat{\theta_i, o_{-i}}{\hat{o}_i^\prime}} \right).
	\end{align*}
	We extend this relation to constant best-responses for all $o_i^\prime \in B_k$ for one of these terms, establishing the bound for a single $B_k \in \mathcal{B}_i$. We apply the low of total expectation to formulate this relation for step functions of the form $\sum_{k=1}^{N_{\mathcal{B}_i}}b_k \mathds{1}_{B_k}$.
\end{sproof}

For each $B_k$, a meaningful upper bound can be established if $\tau_{i, B_k}$ is sufficiently small. A weaker, but potentially easier to determine, bound can be given if there exists an $L_{B_k} > 0$ such that $d_{\text{TV}}\left(g_{B_k}(o_i), g_{B_k}(o_i^\prime) \right) \leq L_{B_k} \norm{o_i - o_i^\prime}$ for $o_i, o_i^\prime \in B_k$. This term is directly related to the diameter of $B_k$, and thus, to the number of elements in the partition $\mathcal{B}_i$.
However, importantly, even if for some $B_l \in \mathcal{B}_i$ the value $\tau_{i, B_l}$ does not have a bound below one, a non-trivial \textit{ex ante} upper bound may still be achievable if it does hold for sufficiently many $B_k \in \mathcal{B}_i$. This makes our results applicable to a wide variety of settings. 
However, while there are several closed-form solutions available for calculating the total variation distance between continuous probability distributions, determining this distance remains hard in general, marking a limitation of our approach. Nevertheless, the growing interest in the total variation distance for applications within machine learning has spurred recent research efforts. For instance, \citet{nielsenGuaranteedDeterministicBounds2018} proposes methods for upper bounding the total variation distance, offering potential pathways to overcome this challenge.

\begin{comment}
	\FP{Discussion about sampling from the conditional prior distribution directly.}
	- one could argue that instead of computing a best-response over a whole partition $B_k$, one chooses a grid of points $o_i^\prime$ over $\mathcal{O}_i$ and collects a dataset from $F|o_i^\prime$ for each of those grid points.\\
	- However, even if we assume to have access to the prior distribution $F$, it may be intractable / computationally hard to sample from $F|o_i^\prime$ directly. \FP{Cite and more detailed discussion.}\\
	- therefore, this approach is infeasible in most practical settings.\\
\end{comment}

\subsection{A sampling-based bound on the \textit{ex ante} utility loss via finite precision step functions}
\begin{comment}
	Our goal is to determine an $\tilde{\varepsilon}_i > 0$ for a strategy profile $\beta=(\beta_{i}, \beta_{-i})$ with high confidence such that it gives an upper bound to the \textit{ex ante} utility loss $\tilde{\ell}_i(\beta)$. Guaranteeing this for every agent $i$, gives a high confidence certificate for $(\beta_{i}, \beta_{-i})$ being an $\tilde{\varepsilon}$-BNE, with $\tilde{\varepsilon}:=\max_{i \in [n]} \tilde{\varepsilon}_i$.
	
	- This comes from two sources. The first is, that one cannot directly evaluate the expected utility for a given strategy profile. \\
	- the second, is that the search for a best-response is over the infinite set of bidding strategies $\Sigma_i$\\
	- in full generality, there is no finite time procedure to determine the best-responses (cite).\\
	
	We tackle these issues with two simplifications to estimate the utility loss. First, instead of evaluating the expectation directly, we approximate the distance to the utility via a Monte-Carlo sampling procedure. Second, instead of searching for a maximum utility loss over the infinite set $\Sigma_i$, we search over a finite set of step-functions, that is specified below.
	
	- we assume access to the unknown prior distribution $F$ in the following form.\\
\end{comment}
In this section, we derive sampling-based estimation bounds $\tilde{\varepsilon}$ for the \textit{ex ante} utility loss. Theorem~\ref{thm:constant-best-response-difference-bounded-over-partition} established that finding a constant best-response for all observations from each element $B_k \in \mathcal{B}_i$ is sufficient. Therefore, we execute the sampling and discretization step for each $B_k \in \mathcal{B}_i$.

Starting with the sampling step, we categorize the dataset $\mathcal{D}^\beta$ according to the partition $\mathcal{B}_i$, for each agent $i$. For each $1 \leq k \leq N_{\mathcal{B}_i}$, we define the conditional samples by
\begin{align*}
	\mathcal{D}^{\beta}\left(B_k\right) := \left\{\giventhat{\left(o^{(j)}, \beta(o^{(j)}), \theta^{(j)} \right) \in \mathcal{D}^\beta}{o^{(j)} \in B_k} \right\}.
\end{align*}
Then, $\mathcal{D}^{\beta}\left(B_k\right)$ constitutes a dataset of draws from $F^\beta|\left\{o_i \in B_k \right\}$. Denote the complete separation of $\mathcal{D}^\beta$ according to partition $\mathcal{B}_i$ by $\mathcal{D}^\beta\left({\mathcal{B}_i}\right) := \left\{\giventhat{\mathcal{D}^{\beta}\left(B_k\right) }{1 \leq k \leq N_{\mathcal{B}_i}}\right\}$.

\begin{comment}
	- that is, given a $B_k \in \mathcal{B}_i$, we want to find the best bid $b_i \in \mathcal{A}_i$ over the \textit{ex ante} utility conditioned on the event $\left\{o_i \in B_k\right\}$.\\
	- However, evaluating the \textit{ex ante} utility loss is, in general, intractable\\
	- instead, we use the dataset $\mathcal{D}^\beta$ of $N$ independent samples of observations, valuations, and bids.\\
	\begin{align}
		\mathcal{D}^\beta := \left\{\giventhat{\left(o^{(j)}, \beta(o^{(j)}), \theta^{(j)} \right)}{\left(o^{(j)}, \theta^{(j)} \right) \in \mathcal{D}} \right\}.
	\end{align}
	- for a datasample $j$, in the \textit{ex interim} stage, an agent $i$ has access to its observation $o_i^{(j)}$, and needs to reason about which bid performs best considering the conditional prior distribution $F^{\beta_{-i}}| o_i^{(j)}$. The tuple $\left(\beta_{-i}(o_{-i}^{(j)}), \theta_{i}^{(j)} \right)$ constitutes a random draw from $F^{\beta_{-i}}| o_i^{(j)}$.\\
\end{comment}
One advantage of providing \textit{ex ante} guarantees, as opposed to \textit{ex interim} guarantees, is the ability to separate the estimation of the best-response utility $\sup_{\beta^{\prime}_i \in \Sigma_i} \tilde{u}_{i, \text{M}}(\beta_i^\prime, \beta_{-i})$ from the estimation of the \textit{ex ante} utility $\tilde{u}_{i, \text{M}}(\beta_i, \beta_{-i})$. Therefore, conveniently, we can estimate the \textit{ex ante} utility using the distribution-independent Hoeffding inequality, eliminating the need to rely on complex concepts such as the pseudo-dimension or partitioning the dataset $\mathcal{D}^{\beta}$. The full proof is in Appendix~\ref{sec:appendix-interdependent-proof-complete-bound}.
\begin{theorem} \label{thm:interdependent-prior-case-estimate-ex-ante-utility-hoeffding}
	Let $\beta \in \Sigma$ be a strategy profile. With probability $1 - \delta$ over the draw of the dataset $\mathcal{D}^{\beta}$, we have for every agent $i \in [n]$
	\begin{align*}
		\abs{\tilde{u}_{i, \text{M}}(\beta_i, \beta_{-i}) - \frac{1}{N} \sum_{j=1}^N u_{i, \text{M}}\left(\theta_{i}^{(j)}, \beta_{i}(o_i^{(j)}), \beta_{-i}(o_{-i}^{(j)})\right) } \leq \sqrt{\frac{2}{N}\log\left(\frac{2 n}{\delta}\right)}.
	\end{align*}
\end{theorem}

\begin{sproof}
	We fix an agent $i \in [n]$ and apply Theorem~\ref{thm:hoeffding-inequality} to $u_{i, \text{M}}(\theta_i, \beta_i(o_i), \beta_{-i}(o_{-i}))$ with $\left(\theta_i, \beta_i(o_i), \beta_{-i}(o_{-i}) \right) \sim F^\beta$. A union bound over the $n$ agents finishes the statement.
\end{sproof}

It remains to estimate the best-response utility. For this purpose, we continue with the sampling step of our approach.
For mechanism $\text{M}$ and agent $i$, define the class of functions that map valuations and opponent bids to utility by $\tilde{\mathcal{F}}_{i, \text{M}}:= \left\{\giventhat{u_{i, \text{M}}(\argdot, b_i, \argdot): \Theta_i \times \mathcal{A}_{-i} \rightarrow \mathbb{R}}{ b_i \in \mathcal{A}_i}  \right\}$. The proof of the following theorem is in Appendix~\ref{sec:appendix-interdependent-proof-complete-bound}.
\begin{theorem} \label{thm:finite-sample-approximation-of-constant-best-response-utility-pseudo-dimension}
	With probability $1 - \delta$ over the draw of the $n$ sets $\mathcal{D}^{\beta}(\mathcal{B}_1), \dots, \mathcal{D}^{\beta}(\mathcal{B}_n)$, for partitions $\mathcal{B}_i = \left\{B_1, \dots, B_{N_{\mathcal{B}_i}} \right\}$ of $\mathcal{O}_i$ for every agent $i \in [n]$, we have
	\begin{align*}
		&\abs{\sup_{b_i \in \mathcal{A}_i} \mathbb{E}_{\giventhat{o_i, o_{-i}, \theta_i}{\left\{o_i \in B_k \right\}}} \left[u_{i, \text{M}}\left(\theta_i, b_i, \beta_{-i}(o_{-i}) \right) \right] - \sup_{b_i \in \mathcal{A}_i} \frac{1}{N_{B_k}} \sum_{j=1}^{N_{B_k}} u_{i, \text{M}}\left(\theta_{i}^{(j)}, b_i, \beta_{-i}(o_{-i}^{(j)})\right) } \\
		& \leq  \tilde{\varepsilon}_{i, \text{Pdim}}(N_{B_k}),\\
		&\mbox{ with } \tilde{\varepsilon}_{i, \text{Pdim}}\left(N_{B_k}\right) := 2\sqrt{\frac{2 d_i}{N_{B_k}} \log\left(\frac{e N_{B_k}}{d_i} \right) } + \sqrt{\frac{2}{N_{B_k}} \log \left(\frac{n N_{\mathcal{B}_{\text{max}}}}{\delta} \right)}, \mbox{ and } d_i := \text{Pdim} \left(\tilde{\mathcal{F}}_{i, \text{M}} \right).
	\end{align*}
\end{theorem}
We proceed with the \emph{discretization step} to identify a constant best-response for each $B_k \in \mathcal{B}_i$ over the bidding space $\mathcal{A}_i$. To this end, we assume $\mathcal{A}_i = [0, 1]^m$ for a suitable $m \in \mathbb{N}$. For a $w > 0$, denote with $\mathcal{G}_{w} \subset [0, 1]^m$ a finite $w$-grid. We make the following assumption.
\begin{assumption} \label{ass:bidding-dispersion-holds-interdependent-prior}
	Suppose that for mechanism $\text{M}$, each agent $i \in[n]$, and segment $B_k \in \mathcal{B}_i$, there exist $L_{i}, w_{i} \in \mathbb{R}$ and a function $v_{i, B_k}: \mathbb{R} \rightarrow \mathbb{R}$, such that with probability $1-\delta$ over the draw of the sets $\left\{\giventhat{\mathcal{D}^{\beta}(B_k)}{B_k \in \mathcal{B}_i, i \in [n]} \right\}$, the functions $u_{i, \text{M}}\left(\theta_i^{(1)}, \argdot, \beta_{-i}(\theta_{-i}^{(1)})\right), \dots, u_{i, \text{M}}\left(\theta_i^{(N_{B_k})}, \argdot, \beta_{-i}(\theta_{-i}^{(N_{B_k})})\right)$ are piecewise $L_i$-Lipschitz and $\left( w_{i}, v_{i, B_k}\left(w_{i} \right)\right)$-dispersed.
\end{assumption}

Under this assumption, we can provide the following approximation bounds by approximating a best-response over a finite subset of the action space. The proof of the following lemma is conceptually similar to the one of Theorem~\ref{thm:independent-prior-guarantee-by-dispersion-over-finite-grid} and can be found in Appendix~\ref{sec:appendix-interdependent-proof-complete-bound}.
\begin{lemma} \label{thm:dispersion-bound-single-segment-approximation-uniform-grid}
	Let $\delta > 0$, $\beta \in \tilde{\Sigma}$ be a strategy profile, and $\text{M}$ be a mechanism. Suppose that for each agent $i \in [n]$ and segment $B_k \in \mathcal{B}_i$, Assumption \ref{ass:bidding-dispersion-holds-interdependent-prior} holds for $w_i >0$ and $v_i(w_i)$. 
	Then, with probability $1- \delta$ over the draw of the sets $\left\{\giventhat{\mathcal{D}^{\beta}(\mathcal{B}_i)}{i \in [n]} \right\}$, agents $i \in [n]$, and segments $B_k \in \mathcal{B}_i$,
	\begin{align*}
		&\abs{\sup_{b_i \in \mathcal{A}_i} \frac{1}{N_{B_k}} \sum_{j=1}^{N_{B_k}} u_{i, \text{M}}\left(\theta_{i}^{(j)}, b_i, \beta_{-i}(o_{-i}^{(j)})\right) -  \max_{b_i \in \mathcal{G}_w} \frac{1}{N_{B_k}} \sum_{j=1}^{N_{B_k}} u_{i, \text{M}}\left(\theta_{i}^{(j)}, b_i, \beta_{-i}(o_{-i}^{(j)})\right)} \\
		&\leq \frac{N_{B_k} - v_{i, B_k}\left(w_{i} \right)}{N_{B_k}} L_{i} w_{i} + \frac{2v_{i, B_k}\left(w_{i} \right)}{N_{B_k}} =: \tilde{\varepsilon}_{i, \text{disp}}(N_{B_k}).
	\end{align*}
\end{lemma}

With this foundation, we can present our main theorem, which combines this section's results to establish an approximation bound on the \textit{ex ante} utility loss. The proof combines Theorems~\ref{thm:constant-best-response-difference-bounded-over-partition}, \ref{thm:finite-sample-approximation-of-constant-best-response-utility-pseudo-dimension}, and Lemma~\ref{thm:dispersion-bound-single-segment-approximation-uniform-grid} and can be found in Appendix~\ref{sec:appendix-interdependent-proof-complete-bound}.
\begin{theorem} \label{thm:complete-approximation-bound-interdependent-prior}
	Let $\delta > 0$ and $\beta \in \tilde{\Sigma}$ be a strategy profile. Suppose that for each agent $i \in [n]$ and segment $B_k \in \mathcal{B}_i$, Assumption \ref{ass:bidding-dispersion-holds-interdependent-prior} holds.
	Then, with probability $1- 4\delta$ over the draw of the sets $\left\{\giventhat{\mathcal{D}^{\beta}(\mathcal{B}_i)}{i \in [n]} \right\}$, agents $i \in [n]$, and segments $B_k \in \mathcal{B}_i$,
	\begin{align*}
		&\tilde{\ell}_i(\beta_i, \beta_{-i}) = \sup_{\beta^{\prime}_i \in \Sigma_i} \tilde{u}_{i, \text{M}}(\beta^{\prime}_i, \beta_{-i}) - \tilde{u}_{i, \text{M}}(\beta_i, \beta_{-i})\\
		&\leq \sum_{k=1}^{N_{\mathcal{B}_i}} \frac{N_{B_k}}{N} \max_{b_i \in\mathcal{G}_{w_i}} \frac{1}{N_{B_k}} \sum_{j=1}^{N_{B_k}} u_{i, \text{M}} (\theta_i^{(j)}, b_i, \beta_{-i}(o_{-i}^{(j)})) - \frac{1}{N} \sum_{l=1}^{N} u_{i, \text{M}} (\theta_i^{(l)}, \beta_i(o_i^{(l)}), \beta_{-i}(o_{-i}^{(l)})) \\
		&+ 2 \sqrt{\frac{2}{N} \log\left(\frac{2n}{\delta} \right)} + \sum_{k=1}^{N_{\mathcal{B}_i}} \frac{N_{B_k}}{N} \min\left\{1, \left(\tau_{i, B_k} + \tilde{\varepsilon}_{i, \text{Pdim}}(N_{B_k}) + \tilde{\varepsilon}_{i, \text{disp}}(N_{B_k}) \right)\right\},
	\end{align*}
	where $\tau_{i, B_k}$, $\tilde{\varepsilon}_{i, \text{Pdim}}(N_{B_k})$, and $\tilde{\varepsilon}_{i, \text{disp}}(N_{B_k})$ are the constants defined in Theorems~\ref{thm:constant-best-response-difference-bounded-over-partition}, \ref{thm:finite-sample-approximation-of-constant-best-response-utility-pseudo-dimension}, and Lemma~\ref{thm:dispersion-bound-single-segment-approximation-uniform-grid}.	
\end{theorem}

\section{Guarantees on dispersion and pseudo-dimension for four mechanisms}
\label{sec:dispersion-and-pseudo-dimension-guarantees}
In this section, we report dispersion and pseudo-dimension guarantees for various mechanisms. These enable us to instantiate the bounds from the previous two sections, thus allowing us to assess the degree to which the empirical utility loss estimates correspond to the true utility losses.

We build upon the work of \citet{balcanEstimatingApproximateIncentive2019}, which offers dispersion and pseudo-dimension guarantees for a range of mechanisms. We demonstrate how to adapt their guarantees to our context---namely strategic bidding and not just for independent priors but also for interdependent priors. We study some of our settings in the body and the rest in the appendix; a detailed summary of all our guarantees can be found in Table~\ref{tab:dispersion-and-pseudo-dim-guarantees}.
\begin{table}[ht]
	\caption{Dispersion and pseudo-dimension guarantees for different auction mechanisms. Interchanging $\kappa_{i, B_k}$ and $N_{B_k}$ in the right column with $\kappa$ and $N$ gives the dispersion results for the independent prior case $v_{i}$.}
	\label{tab:dispersion-and-pseudo-dim-guarantees}
	\begin{tabular}{lll}
		\hline
		Mechanism &
		\begin{tabular}[c]{@{}l@{}}Pseudo-dimension guarantees\\ for $\hat{\mathcal{F}}_{i, \text{M}}$ and $\tilde{\mathcal{F}}_{i, \text{M}}$\end{tabular} &
		Dispersion guarantees \\ \hline
		First-price single-item auction &
		$\tilde{O}(1)$ &
		\begin{tabular}[c]{@{}l@{}}$w_i = O\left(1 / \left(\kappa_{i, B_k}L_{\beta^{-1}_{t_{\text{max}}}}\sqrt{N_{B_k}}\right)\right)$\\ $v_{i, B_k}(w_i) = \tilde{O}\left(n \sqrt{N_{B_k}} \right)$\end{tabular} \\ \hline
		\begin{tabular}[c]{@{}l@{}}First-price combinatorial\\ auction over $l$ items\end{tabular} &
		$O\left(l 2^l \log(n) \right)$ &
		\begin{tabular}[c]{@{}l@{}}$w_i = O\left(1 / \left(\kappa_{i, B_k}L^{2^{l+1}}_{\beta^{-1}_{t_{\text{max}}}}\sqrt{N_{B_k}}\right)\right)$\\ $v_{i, B_k}(w_i)=\tilde{O}\left((n+1)^{2l} \sqrt{N_{B_k}l} \right)$\end{tabular} \\ \hline
		\begin{tabular}[c]{@{}l@{}}Discriminatory auction\\ over $m$ units of a single good\end{tabular} &
		$O\left(m \log(n m) \right)$ &
		\begin{tabular}[c]{@{}l@{}}$w_i = O\left(1 / \left(\kappa_{i, B_k}L_{\beta^{-1}_{t_{\text{max}}}}\sqrt{N_{B_k}}\right)\right)$\\ $v_{i, B_k}(w_i)=\tilde{O}\left(n m^2 \sqrt{N_{B_k}} \right)$\end{tabular} \\ \hline
		\begin{tabular}[c]{@{}l@{}}Uniform-price auction\\ over $m$ units of a single good\end{tabular} &
		$O\left(m \log(n m) \right)$ &
		\begin{tabular}[c]{@{}l@{}}$w_i = O\left(1 / \left(\kappa_{i, B_k}L_{\beta^{-1}_{t_{\text{max}}}}\sqrt{N_{B_k}}\right)\right)$\\ $v_{i, B_k}(w_i)=\tilde{O}\left(n m^2 \sqrt{N_{B_k}} \right)$\end{tabular}
	\end{tabular}
\end{table}

\subsection{Dispersion guarantees under strategic bidding}
To adapt the dispersion guarantees from \citet{balcanEstimatingApproximateIncentive2019} to our context, two significant modifications are required. First, in situations involving interdependent priors, it is necessary to focus on the conditional prior distribution. Second, one needs to reason about the (conditional) bidding distribution $F^\beta$ instead of the prior distribution $F$.
To address the first challenge, we extend the assumption of $\kappa$-bounded distributions to the conditional prior distribution. We then apply Theorem~\ref{thm:bounded-prior-density-remains-bounded-under-bi-Lipschitz} to tackle the second challenge. As a result, a $\kappa$-bounded prior distribution transforms into a $\kappa L_{\beta^{-1}_{\text{max}}}$-bounded bidding distribution.
By making these adjustments, we can apply the dispersion guarantees from \citet{balcanEstimatingApproximateIncentive2019} to our specific situation with small modifications required in the original proofs.
We illustrate how to formulate and extend the dispersion guarantees for the first-price single-item auction. For the detailed statements on other mechanisms, see Appendix~\ref{sec:dispersion-and-pseudo-dim-appendix}.

\subsection*{First-price single-item auction}
In the first-price single-item auction, the item is awarded to the highest bidder, who then pays the amount of its bid. Each agent $i$ has a valuation $\theta_i \in [0, 1]$ for the item and submits a bid $b_i \in [0, 1]$. The utility function for agent $i$ is given by $u_{i, \text{M}}(\theta_i, b_i, b_{-i}) = \mathds{1}_{\{b_i > \norm{b_{-i}}_{\infty} \}}\left(\theta_i - b_i\right)$. We limit ourselves to present the statement for the interdependent prior case, as it incorporates both changes described above. For the statement with independent prior distributions, see Appendix~\ref{sec:disp-and-pdim-first-price-single-item-auction}. 

The following theorem asserts Assumption~\ref{ass:bidding-dispersion-holds-interdependent-prior} is valid for the first-price auction with interdependent prior distributions (Section~\ref{sec:verifying-equilibria-with-interdependent-priors}). The full proof is in Appendix~\ref{sec:disp-and-pdim-first-price-single-item-auction}.
\begin{theorem} \label{thm:dispersion-guarantee-for-FPSB-single-item}
	Let $(\beta_i, \beta_{-i}) \in \tilde{\Sigma}$. Assume that for each agent $i \in [n]$ and segment $B_k \in \mathcal{B}_i$, there exists $\kappa_{i, B_k}>0$, such that the conditional marginal distributions $F_{\giventhat{o_j}{\left\{o_i \in B_k \right\}}}$ for $j \in [n] \setminus \{i\}$ are $\kappa_{i, B_k}$ bounded. Then, for $w_i>0$, with probability at least $1 - \delta$ over the draw of the sets $\left\{\giventhat{\mathcal{D}^{\beta_{-i}}(\mathcal{B}_i)}{i \in [n]} \right\}$ for every $i \in [n]$ and $B_k \in \mathcal{B}_i$, the functions 
	$u_{i, \text{M}}\left(\theta_i^{(1)}, \cdot, \beta_{-i}(o_{-i}^{(1)}) \right), \dots, u_{i, \text{M}}\left(\theta_i^{(N_{B_k})}, \cdot, \beta_{-i}(o_{-i}^{(N_{B_k})}) \right)$ are piecewise 1-Lipschitz and $\left(w_i, v_{i, B_k}\left(w_{i} \right) \right)$-dispersed, with
	$v_{i, B_k}\left(w_{i} \right) := (n-1)w_i N_{B_k} \kappa_{i, B_k} L_{\beta^{-1}_{\text{max}}} 
	+ (n-1) \sqrt{2 N_{B_k} \log \left(\frac{2n(n-1) N_{\mathcal{B}_{\text{max}}}}{\delta} \right)} + 4(n-1) \sqrt{N_{B_k} \log \left(\frac{e N_{B_k}}{2} \right)}$.
\end{theorem}
\begin{sproof}
	For agent $i \in [n]$, apply Theorem~\ref{thm:bounded-prior-density-remains-bounded-under-bi-Lipschitz} to the marginal bidding distribution $F_{\giventhat{o_j}{\left\{o_i \in B_k \right\}}}^{\beta_{i}}$. Then, the $\kappa_{i, B_k}$-bounded density function for every agent $j \in [n]\setminus \{i \}$ and $B_k \in \mathcal{B}_i$ transforms into a $\kappa_{i, B_k} L_{\beta^{-1}_j}$-bounded bidding distribution. 
	Next, we determine that for a sample $j$, the discontinuity in the utility functions is located at the point $\norm{\beta_{-i}(o_{-i}^{(j)})}_{\infty}$. Following this, we apply standard dispersion results (as detailed in Appendix~\ref{sec:generic-dispersion-statements}) to restrict the number of points $\left\{\beta_l(o_l^{(j)}) \right\}_{j \in [N_{B_k}], l \in [n] \setminus \{i\}}$ within any interval of width $w_i$ with high probability.
	A suitable union bound finishes the statement.
\end{sproof}

We provide dispersion guarantees for three other mechanisms in Appendix~\ref{sec:dispersion-and-pseudo-dim-appendix}. A detailed summary of all our guarantees can be found in Table~\ref{tab:dispersion-and-pseudo-dim-guarantees}.

\subsection{Pseudo-dimension guarantees via delineability}
\citet{balcanEstimatingApproximateIncentive2019} build their pseudo-dimension guarantees on the concept of $(m, t)$-delineability~\citep{balcanGeneralTheorySample2018}. If one can show that a function class is $(m, t)$-delineable, one can bound its pseudo-dimension. \citet{balcanEstimatingApproximateIncentive2019} show $\hat{\mathcal{F}}_{i, \text{M}}$ is $(2m, t)$-delineable for several auction mechanisms to derive their bounds. We extend their statements to $\tilde{\mathcal{F}}_{i, \text{M}}$ by showing if $\hat{\mathcal{F}}_{i, \text{M}}$ is $(2m, t)$-delineable, then $\tilde{\mathcal{F}}_{i, \text{M}}$ is $(m, t)$-delineable. This way, we can readily extend their pseudo-dimension guarantees.
The concept of $(m, t)$-delineability is defined as follows.
\begin{definition}[$(m, t)$-delineable, \citet{balcanGeneralTheorySample2018}] \label{def:mt-delineable-functions}
	Let $\mathcal{P} \subset \mathbb{R}^m$ and $\mathcal{X}$ a vector-space. A class of functions $\mathcal{F} = \left\{\giventhat{f(\argdot, p): \mathcal{X} \rightarrow \mathbb{R} }{p\in \mathcal{P} }\right\}$ is $(m, t)$-dealineable if
	for any $v\in \mathcal{X}$, there is a set $\mathcal{H}$ of $t$ hyperplanes such that for any connected component $\mathcal{P}^\prime$ of $\mathcal{P} \setminus \mathcal{H}$, $f(v, p)$ is linear over $\mathcal{P}^\prime$. 
\end{definition}
The following theorem is similar to \citet{balcanGeneralTheorySample2018}'s main statement to bound the pseudo-dimension of an $(m, t)$-delineable function class. We slightly reformulated it to our setting.
\begin{theorem}[\citet{balcanGeneralTheorySample2018}]\label{thm:balcans-general-pseudo-dim-bound-for-mt-delineable-functions}
	If a function class $\mathcal{F}$ is $(m, t)$-dealineable, then $\text{Pdim}\left(\mathcal{F} \right) = O\left(m \log(mt) \right)$.
\end{theorem}
We now give our statement that extends the pseudo-dimension results from $\hat{\mathcal{F}}_{i, \text{M}}$ to $\tilde{\mathcal{F}}_{i, \text{M}}$.
\begin{theorem}\label{thm:if-hatF-delieanble-then-tildeF-delineable}
	Let $\text{M}$ be a mechanism and $i \in [n]$. Suppose the function class $\hat{\mathcal{F}}_{i, \text{M}}$ is $(2m, t)$-delineable, then $\tilde{\mathcal{F}}_{i, \text{M}}$ is $(m, t)$-delineable.
\end{theorem}
\begin{proof}
	For a $b_{-i} \in \mathcal{A}_{-i}$, let $C^{\Theta_i} \times C^{\mathcal{A}_i} \subset \Theta_i \times \mathcal{A}_i = [0, 1]^m \times [0, 1]^m$ be an open subset such that $u_{i, \text{M}}(\theta_i, b_i, b_{-i}) = x_i(b_i, b_{-i}) \cdot \theta_i - p_i(b_i, b_{-i})$ is linear in $(\theta_i, b_i)$ over $C^{\Theta_i} \times C^{\mathcal{A}_i}$. As the allocation $x_i(b_i, b_{-i}) \in \{0, 1\}^m$ and the price $p_i(b_i, b_{-i}$ are independent of $\theta_i$, the allocation $x_i$ has to be constant for all $b_i \in C^{\mathcal{A}_i}$, otherwise there would be a jump for a changing $\theta_i$. Therefore, $u_{i, \text{M}}(\theta_i, b_i, b_{-i})$ is linear in $\theta_i \in \Theta_i$ for $b_i \in C^{\mathcal{A}_i}$.
	
	Let $(\theta_i, b_{-i}) \in \Theta_i \times \mathcal{A}_{-i}$. As $\hat{\mathcal{F}}_{i, \text{M}}$ is $(2m, t)$-delineable, for $b_{-i}$, there exists a set $\hat{\mathcal{H}}$ of $t$ hyperplanes such that for any connected component $C^{\Theta_i}_l \times C^{\mathcal{A}_i}_l$ of $\Theta_i \times \mathcal{A}_i \setminus \hat{\mathcal{H}}$ the utility $u_{i, \text{M}}(\theta_i^\prime, b_i, b_{-i})$ is linear for $\theta_i^\prime \in C^{\Theta_i}_l$ and $b_i \in C^{\mathcal{A}_i}_l$. Denote with $\left\{C^{\Theta_i}_l \times C^{\mathcal{A}_i}_l \right\}_{l \in [N_t]}$ the set of connected components of $\Theta_i \times \mathcal{A}_i \setminus \hat{\mathcal{H}}$, where $N_t$ is the number of connected components. For $b_{-i}$, we need at most $t$ hyperplanes $\tilde{\mathcal{H}}$ so that $\mathcal{A}_i \setminus \tilde{\mathcal{H}} = \bigcup_{l \in [N_t]} C^{\mathcal{A}_i}_l$.    
	By the argument above, the allocation is fixed for $b_i \in C^\mathcal{A_i}_l$ for every $l \in [N_t]$ and $u_{i, \text{M}}(\theta_i^\prime, b_i, b_{-i})$ is linear in $\theta_i^\prime \in \Theta_i$. Therefore, $u_{i, \text{M}}(\theta_i, b_i, b_{-i})$ is linear in $b_i \in C^{\mathcal{A}_i}_l$. Therefore, $\tilde{\mathcal{F}}_{i, \text{M}}$ is $(m, t)$-delineable.
\end{proof}
\begin{comment}
	As $\hat{\mathcal{F}}_{i, \text{M}}$ is $(2m, t)$-delineable, for a fixed $b_{-i}$, there exists a set $\hat{\mathcal{H}}$ of $t$ hyperplanes such that $u_{i, \text{M}}(\theta_i^\prime, b_i, b_{-i})$ is linear over the connected components of $\Theta_i \times \mathcal{A}_i \setminus \hat{\mathcal{H}}$. As $u_{i, \text{M}}(\theta_i, b_i, b_{-i}) = x_i(b_i, b_{-i}) \cdot \theta_i - p_i(b_i, b_{-i})$ is always linear in $\theta_i$, there exists, for a fixed $(\theta_i^\prime, b_{-i})$, a set of hyperplanes $\tilde{\mathcal{H}}$ with $\abs{\tilde{\mathcal{H}}}\leq t$ and $u_{i, \text{M}}(\theta_i, b_i, b_{-i})$ being linear in $b_i$ over the connected components of $\mathcal{A}_i \setminus \tilde{\mathcal{H}}$.
\end{comment}

Due to space restrictions, we direct readers to Appendix~\ref{sec:dispersion-and-pseudo-dim-appendix} for thorough descriptions of the mechanisms and the detailed guarantees derived by the approach described above.

\section{Conclusions and future research}
We introduced sampling-based methods for estimating the distance of a strategy profile from an \emph{ex interim} or \textit{ex ante} Bayesian Nash equilibrium. 
Our approach significantly broadens the scope of approximate equilibrium verification compared to prior methods, which rely on narrow assumptions like truthful bidding, single-item auctions, and/or complete knowledge of the prior. 
Notably, we enhance the sampling method proposed by \citet{balcanEstimatingApproximateIncentive2019} by extending it to allow strategic bidding, and correcting their prior assertion regarding its applicability to interdependent priors in the \textit{ex ante} scenario.

Our key contribution is the development of an empirical estimator for the utility loss, which---intuitively speaking---measures the maximum utility an agent can gain by deviating from its current strategy. We have effectively bounded the error between this empirical estimate and the true utility loss by employing a mixture of learning theory tools such as dispersion and pseudo-dimension. We established sufficient conditions for strategy profiles and a closeness criterion for conditional distributions that ensure that utility gains estimated through our finite subset of the strategy space closely approximate the maximum gains. We thus derived strong guarantees for a broad class of auctions with independent or interdependent priors, including the first-price single-item and combinatorial auction, discriminatory auction, and uniform-price auction.

In related research, we discussed several promising techniques to computationally determine equilibrium candidates in complex auctions. To better understand the implications of our results, a natural next step is to combine equilibrium computation with our method of verification to analyze practically relevant settings.

However, it is important to note that our current bounds on the utility loss scale exponentially with the complexity inherent in general combinatorial auctions. Recognizing this limitation, a valuable avenue for future research involves exploiting the unique structural characteristics of certain combinatorial auctions, such as those involving items that are substitutes or complements. By doing so, there is potential to derive bounds that scale polynomially rather than exponentially with the number of items. This could significantly enhance the efficiency and feasibility of applying our methods to a broader range of auction formats, thereby extending their practical applicability.

\bibliographystyle{unsrtnat}
\bibliography{My_Library,dairefs}  %%% Uncomment this line and comment out the ``thebibliography'' section below to use the external .bib file (using bibtex) .





% \end{thebibliography}

\appendix
\section{Auxiliary lemmas and results}
In this section, we introduce some helpful concepts to proof our results.
\subsection{Bi-Lipschitz continuous functions}
We revisit some well-established results from existing literature. 

Formally, the restrictions on the rate of change for a bi-Lipschitz mapping are captured by the following bounds on the determinant of its Jacobian matrix. This is presented in the following lemma.
\begin{lemma}[\citet{verineExpressivityBiLipschitzNormalizing2023, federerGeometricMeasureTheory1996}] \label{thm:bound-determinante-of-jacobian-of-bi-lipschitz-map}
	Let $g: \mathcal{X} \subset \mathbb{R}^m \rightarrow \mathcal{Y}$ be a $(L_g, L_{g^{-1}})$-bi-Lipschitz function. Then, for all $x \in \mathcal{X}$ it holds that
	\begin{align*}
		\frac{1}{L_{g^{-1}}^m} \leq \abs{\det(\mathcal{J} g (x))} \leq L_g^m \text{ and } \frac{1}{L_g^m} \leq \abs{\det(\mathcal{J} g^{-1} (x))} \leq L_{g^{-1}}^m.
	\end{align*}
\end{lemma}

The change of variables formula serves as a foundational tool in our analysis, permitting the expression of the density function of a probability measure under a mapping that exhibits sufficient regularity, such as bi-Lipschitz maps. We consider the following version of the well-known change of variables formula. 
\begin{theorem}[Change of Variables, {\citet[p.12]{villaniOptimalTransport2009}}] \label{thm:change-of-variables}
	Let $\mathcal{X}, \mathcal{Y} \subset \mathbb{R}^m$ be open, bounded, and connected subsets. Let $\mu_0, \mu_1$ be two probability measures on $\mathcal{X}$ and $\mathcal{Y}$, respectively, that are absolutely continuous with respect to the Borel-measure $\lambda$. Let $T: \mathcal{X} \rightarrow \mathcal{Y}$ be an injective, locally Lipschitz function such that $\mu_1$ is the pushforward measure of $\mu_0$ under $T$, that is, $T_{\#}\mu_0 = \mu_1$. Then, it holds that
	\begin{align*}
		\phi_{\mu_0}(x) = \phi_{\mu_1}(T(x)) \abs{\det{\mathcal{J} T(x)}},
	\end{align*}
	where $\phi_{\mu_0}, \phi_{\mu_1}$ denote the density functions of $\mu_0$ and $\mu_1$, respectively, and $\mathcal{J} T$ denotes the Jacobian matrix of $T$.
\end{theorem}
% This is not the exact same formula as stated in Villani. I need to be able to map between different subsets of R^m. However, Villani states the change of variables only for the same Riemann manifold. At the same time, other versions directly need a C1 Diffeomorphism. This is inconvenient, as I want to simply state the mappings to be bi-Lipschitz. I formulate this over open subsets first, as every open subset of a Riemann manifold is again a Riemann manifold: https://math.stackexchange.com/questions/122108/an-open-subset-of-a-manifold-is-a-manifold. Furthermore, this is then closer to standard formulations with Diffeomorphisms: https://www.math.toronto.edu/courses/mat237y1/20199/notes/Chapter4/S4.4.html

The following well-known statement directly follows from Theorem \ref{thm:change-of-variables} and Lemma \ref{thm:bound-determinante-of-jacobian-of-bi-lipschitz-map}.
\begin{lemma}\label{thm:boundedness-of-distribution-under-bi-lipschitz-map}
	Let $\mathcal{X}, \mathcal{Y} \subset \mathbb{R}^m$ and $g: \mathcal{X} \rightarrow \mathcal{Y}$ be a $(L_g, L_{g^{-1}})$-bi-Lipschitz function. Furthermore, let $\mu$ be a probability measure over $\mathcal{X}$ with a $\kappa$-bounded density function $\phi_{\mu}$, i.e., $\sup_{x \in \mathcal{X}} \phi_{\mu}(x) \leq \kappa$. Then, the push-forward probability measure $g_{\#}\mu$ has a $\kappa \cdot L_{g^{-1}}^m$-bounded density function 
	\begin{align*}
		\sup_{y \in \mathcal{Y}}\phi_{g_{\#}\mu}(y) \leq \kappa \cdot L_{g^{-1}}^m.
	\end{align*}
\end{lemma}

\begin{proof}
	We start by using the change of variables formula from Theorem \ref{thm:change-of-variables}. Let $\mu_0 := g_{\#}\mu$ and $\mu_1 := g^{-1}_{\#} \left(g_{\#}\mu\right) = \mu$. Then, we get for $T=g^{-1}$ and $y \in \mathcal{Y}$
	\begin{align*}
		\phi_{\mu_0}(y) = \phi_{g_{\#}\mu}(y) \overset{\text{Thm} \ref{thm:change-of-variables}}{=} \phi_{\mu}(g^{-1}(y)) \cdot \abs{ \det(\mathcal{J} g^{-1}(y))} \leq \kappa \cdot \abs{ \det(\mathcal{J} g^{-1}(y))} \overset{\text{Lemma } \ref{thm:bound-determinante-of-jacobian-of-bi-lipschitz-map}}{\leq} \kappa \cdot L_{g^{-1}}^m.
	\end{align*}
\end{proof}

\subsection{Generic dispersion statements} \label{sec:generic-dispersion-statements}
We present several generic dispersion lemmas based on the work by \citet{balcanDispersionDataDrivenAlgorithm2018}, refining some of their statements to provide explicit guarantees rather than presenting results in big $O$ notation. This refinement requires minor adjustments to their proofs.
However, first we introduce the Hoeffding inequality, which is another well-known concentration bound. It provides a distribution-independent concentration bound, enabling an accurate sampling-based estimation of the expectation of a single random variable.

\begin{theorem}[{\citet{hoeffdingProbabilityInequalitiesSums1963}}] \label{thm:hoeffding-inequality}
	Let $X=X^{(1)}, \dots, X^{(N)}$ be i.i.d. random variables over $[-1, 1]$. Then, with probability at least $1 - \delta$,
	\begin{align*}
		\abs{\frac{1}{N}\sum_{j=1}^N X^{(j)} - \mathbb{E}\left[X\right] } \leq \sqrt{\frac{2}{N}\log\left(\frac{2}{\delta}\right)}.
	\end{align*}
\end{theorem}

We restate a well-known folklore lemma next, providing explicit bounds for uniform convergence for non-identical random variables. This is supported by well-established results regarding Rademacher complexity and the VC-dimension~\citep{DBLP:books/daglib/0034861}.
\begin{lemma}[{\citet[Lemma 2, p.23]{balcanDispersionDataDrivenAlgorithm2018}}] \label{thm:uniform-convergence-aux-lemma-dispersion}
	Let $S = \{z_1, \ldots, z_r\} \subset \mathbb{R} $ be a set of random variables where \( z_i \sim p_i \). For any \( \delta > 0 \), with probability at least \( 1 - \delta \) over the draw of the set \( S \),
	
	\[
	\sup_{a,b \in \mathbb{R}, a < b} \left( \left| \sum_{i=1}^{r} \mathds{1}_{z_i \in (a,b)} - \mathbb{E}_{S'} \left[ \sum_{i=1}^{r} \mathds{1}_{z'_i \in (a,b)} \right] \right| \right) \leq \sqrt{2r \log\left(\frac{2}{\delta}\right)} + 4\sqrt{r \log\left(\frac{e r}{2} \right)},
	\]
	
	where \( S' = \{z'_1, \ldots, z'_r\} \) is another sample drawn from \( p_1, \ldots, p_r \).
\end{lemma}

\begin{proof}
	Let $\sigma$ be an $r$-dimensional vector of Rademacher random variables. The empirical Rademacher complexity is given by
	\begin{align*}
		\hat{R}_S(G) := \mathbb{E}_{\sigma} \left[\sup_{a,b \in \mathbb{R}, a < b} \frac{1}{r} \sum_{i=1}^{r} \sigma_i \mathds{1}_{z_i \in (a, b)} \right],
	\end{align*}
	where $G$ denotes the set of indicator functions over intervals. The empirical Rademacher complexity can be bounded via the VC-dimension $d = VCdim(G)$ by
	\begin{align*}
		\hat{R}_S(G) \leq \sqrt{\frac{2 \log\left(\frac{e r}{d}\right)}{r}},
	\end{align*}
	which uses Corollary 3.1 and 3.3 by \citet{DBLP:books/daglib/0034861}, and that we can bound the empirical Rademacher complexity by the Rademacher complexity for distribution-independent bounds.
	Therefore,
	\begin{align} \label{equ:uniform-convergence-aux-lemma-dispersion-1}
		r \hat{R}_S(G) \leq 2 \sqrt{r \log\left(\frac{e r}{d}\right)}.
	\end{align}
	Following the proof by \citet[Lemma 2, p.23]{balcanDispersionDataDrivenAlgorithm2018}, we derive
	\begin{align} \label{equ:uniform-convergence-aux-lemma-dispersion-2}
		\sup_{a,b \in \mathbb{R}, a < b} \left( \sum_{i=1}^{r} \mathds{1}_{z_i \in (a, b)} - \mathbb{E}_{S^\prime} \left[ \sum_{i=1}^{r} \mathds{1}_{z^\prime_i \in (a, b)} \right] \right) \leq 2 \mathbb{E}_{\sigma, S} \left[\sup_{a,b \in \mathbb{R}, a < b}  \sum_{i=1}^{r} \sigma_i \mathds{1}_{z_i \in (a, b)} \right],
	\end{align}
	and
	\begin{align} \label{equ:uniform-convergence-aux-lemma-dispersion-3}
		\abs{\mathbb{E}_{\sigma} \left[\sup_{a,b \in \mathbb{R}, a < b} \sum_{i=1}^{r} \sigma_i \mathds{1}_{z_i \in (a, b)} \right] - \mathbb{E}_{\sigma, S} \left[\sup_{a,b \in \mathbb{R}, a < b}  \sum_{i=1}^{r} \sigma_i \mathds{1}_{z_i \in (a, b)} \right]} \leq \sqrt{\frac{r}{2} \log\left(\frac{2}{\delta} \right) }.
	\end{align}
	Combining the results up until now results in
	\begin{align*}
		&\sup_{a,b \in \mathbb{R}, a < b} \left( \left| \sum_{i=1}^{r} 1_{z_i \in (a,b)} - \mathbb{E}_{S'} \left[ \sum_{i=1}^{r} 1_{z'_i \in (a,b)} \right] \right| \right)\\
		& \overset{\text{Equ}. \ref{equ:uniform-convergence-aux-lemma-dispersion-2}}{\leq} 2 \mathbb{E}_{\sigma, S} \left[\sup_{a,b \in \mathbb{R}, a < b}  \sum_{i=1}^{r} \sigma_i \mathds{1}_{z_i \in (a, b)} \right]\\
		&\leq 2 \abs{\mathbb{E}_{\sigma} \left[\sup_{a,b \in \mathbb{R}, a < b} \sum_{i=1}^{r} \sigma_i \mathds{1}_{z_i \in (a, b)} \right] - \mathbb{E}_{\sigma, S} \left[\sup_{a,b \in \mathbb{R}, a < b}  \sum_{i=1}^{r} \sigma_i \mathds{1}_{z_i \in (a, b)} \right]} \\
		&+ 2\abs{\mathbb{E}_{\sigma} \left[\sup_{a,b \in \mathbb{R}, a < b}  \sum_{i=1}^{r} \sigma_i \mathds{1}_{z_i \in (a, b)} \right] }\\
		& \overset{\text{Equ}. \ref{equ:uniform-convergence-aux-lemma-dispersion-1} \text{ and } \ref{equ:uniform-convergence-aux-lemma-dispersion-3}}{\leq} \sqrt{2r \log\left(\frac{2}{\delta}\right)} + 4\sqrt{r \log\left(\frac{e r}{2} \right)}.
	\end{align*}
	
\end{proof}

To prove dispersion we will use the following probabilistic lemma, showing that samples from $\kappa$-bounded distributions do not tightly concentrate.\\

\begin{lemma}[{\citet[Lemma 1, p.23]{balcanDispersionDataDrivenAlgorithm2018}}] \label{thm:aux-dispersion-lemma-indepdent-and-buckets}
	Let \( S = \{z_1, \ldots, z_r\} \subset \mathbb{R} \) be a collection of samples where each \( z_i \) is drawn from a \( \kappa \)-bounded distribution with density function \( p_i \). For any \( \delta \geq 0 \), the following statements hold with probability at least \( 1 - \delta \):
	
	\begin{enumerate}
		\item If the \( z_i \) are independent, then every interval of width \( w \) contains at most \( k = rw\kappa + \sqrt{2r \log\left(\frac{2}{\delta}\right)} + 4\sqrt{r \log\left(\frac{e r}{2} \right)} \) samples.
		\item If the samples can be partitioned into \( P \) buckets \( S_1, \ldots, S_P \) such that each \( S_i \) contains independent samples and \( |S_i| \leq M \), then every interval of width \( w \) contains at most \( k = Pw \kappa M + P\sqrt{2M \log\left(\frac{2P}{\delta}\right)} + 4P\sqrt{M \log\left(\frac{e M}{2} \right)} \) samples.
	\end{enumerate}
\end{lemma}

\begin{proof}
	We consider Part 1 first. The expected number of samples that land in an interval $(a, b)$ of width $w$ is at most $w \kappa r$, since the probability that $z_i \in (a,b)$ is bounded by $w \kappa$. By Lemma \ref{thm:uniform-convergence-aux-lemma-dispersion}, we have that with probability at least $1 - \delta$ over the draw of the set $S$,
	\begin{align} \label{equ:aux-dispersion-lemma-indepdent-and-buckets-1}
		\sup_{a,b \in \mathbb{R}, a < b} \left( \left| \sum_{i=1}^{r} 1_{z_i \in (a,b)} - \mathbb{E}_{S'} \left[ \sum_{i=1}^{r} 1_{z'_i \in (a,b)} \right] \right| \right) \leq \sqrt{2r \log\left(\frac{2}{\delta}\right)} + 4\sqrt{r \log\left(\frac{e r}{2} \right)},
	\end{align}
	where $S^\prime = {z^\prime_1, \dots, z_r^\prime}$ is another sample from $p_1, \dots, p_r$. The number of elements in an interval $(a, b)$ satisfies
	\begin{align} \label{equ:aux-dispersion-lemma-indepdent-and-buckets-2}
		\sum_{i=1}^r \mathds{1}_{z_i \in (a, b)} \leq \mathbb{E}_{S^\prime} \left[\sum_{i=1}^r \mathds{1}_{z^\prime_i \in (a, b)} \right] + \abs{\sum_{i=1}^r \mathds{1}_{z_i \in (a, b)} - \mathbb{E}_{S^\prime} \left[\sum_{i=1}^r \mathds{1}_{z^\prime_i \in (a, b)} \right]}.
	\end{align}
	Combining Equations \ref{equ:aux-dispersion-lemma-indepdent-and-buckets-1} and \ref{equ:aux-dispersion-lemma-indepdent-and-buckets-2} implies that with probability of at least $1 - \delta$, every interval $(a, b)$ of width $w$ satisfies $\abs{S \cap (a, b)} \leq rw\kappa + \sqrt{2r \log\left(\frac{2}{\delta}\right)} + 4\sqrt{r \log\left(\frac{e r}{2} \right)}$.
	
	Part 2 follows by applying Part 1 to each bucket $S_i$ and taking a union bound over the buckets.
\end{proof}

\section{Dispersion and pseudo-dimension guarantees} \label{sec:dispersion-and-pseudo-dim-appendix}
We provide detailed statements regarding the dispersion and pseudo-dimension guarantees for several mechanisms. The descriptions of these market mechanisms are adapted from \citet{balcanEstimatingApproximateIncentive2019}.

\subsection{First-price single-item auction} \label{sec:disp-and-pdim-first-price-single-item-auction}
In the first-price auction, the item is awarded to the highest bidder, who then pays the amount of its bid. Each agent $i$ has a valuation $\theta_i \in [0, 1]$ for the item and submits a bid $b_i \in [0, 1]$. The utility function for agent $i$ is given by $u_{i, \text{M}}(\theta_i, b_i, b_{-i}) = \mathds{1}_{\{b_i > \norm{b_{-i}}_{\infty} \}}\left(\theta_i - b_i\right)$, where $\norm{b{-i}}_{\infty}$ denotes the highest bid among the other bidders.

In the context of independent prior distributions (Section~\ref{sec:independent-prior-statements}), we show Assumption~\ref{ass:general-dispersion-holds-independent-prior} is satisfied with the following statement.
\begin{theorem}
	Assume every agent $i \in [n]$ has a $\kappa$-bounded marginal prior distribution $F_{\theta_i}$. Let $\beta \in \tilde{\Sigma}$ be a strategy profile of bi-Lipschitz bidding strategies.
	With probability $1 - \delta$ for all agents $i \in [n]$ over the draw of the $n$ datasets $\mathcal{D}^{\beta_{-i}} := \left\{\beta_{-i}(\theta_{-i}^{(1)}), \dots, \beta_{-i}(\theta_{-i}^{(N)}) \right\}$,
	\begin{enumerate}
		\item For any $\theta_i\in [0, 1]$, the functions $u_{i, \text{M}}(\theta_i, \argdot, \beta_{-i}(\theta_{-i}^{(1)})), \dots, u_{i, \text{M}}(\theta_i, \argdot, \beta_{-i}(\theta_{-i}^{(N)}))$ are piecewise $1$-Lipschitz and $\left(w_i, v_i\left(w_i \right)\right)$-dispersed with $v_i\left(w_i \right) := (n-1)w_i N \kappa L_{\beta^{-1}_{\text{max}}} + (n-1) \sqrt{2N \log\left(\frac{2n(n-1)}{\delta}\right)} + 4(n-1) \sqrt{N \log\left(\frac{eN}{2}\right)}$.
		\item For any $b_i \in [0, 1]$ and $b_{-i} \in [0, 1]^{n-1}$, the function $u_{i, \text{M}}(\argdot, b_i, b_{-i})$ is $1$-Lipschitz continuous.
	\end{enumerate}
\end{theorem}
\begin{proof}
	We start with the first part of the statement. Consider $i \in [n]$ and $\beta_{-i}(\theta_{-i}^{(j)}) \in \mathcal{D}^{\beta_{-i}}$ arbitrary. For any $\theta_i \in \Theta_i$ and bid $b_i \in \Theta_i$, we have $u_{i, \text{M}}(\theta_i, b_i, \beta_{-i}(\theta_{-i}^{(j)})) = \mathds{1}_{\left\{b_i > \norm{\beta_{-i}(\theta_{-i}^{(j)})}_{\infty} \right\}}\left(\theta_i - b_i\right)$. Therefore, if $b_i \leq \norm{\beta_{-i}(\theta_{-i}^{(j)})}_{\infty}$, then $u_{i, \text{M}}(\theta_i, b_i, \beta_{-i}(\theta_{-i}^{(j)}))$ is a constant function in $b_i$. On the other hand, if $b_i > \norm{\beta_{-i}(\theta_{-i}^{(j)})}_{\infty}$, then $u_{i, \text{M}}(\theta_i, b_i, \beta_{-i}(\theta_{-i}^{(j)}))$ is linear in $b_i$ with slope of $-1$. Consequently, we have for all $\theta_i \in \Theta_i$ and $\beta_{-i}(\theta_{-i}^{(j)}) \in \mathcal{D}^{\beta_{-i}}$, the function $u_{i, \text{M}}(\theta_i, \argdot, \beta_{-i}(\theta_{-i}^{(j)}))$ is piecewise $1$-Lipschitz continuous with a discontinuity at $\norm{\beta_{-i}(\theta_{-i}^{(j)})}_{\infty}$.
	
	We proceed with the dispersion constants $\left(w_i, v_i\left(w_i \right)\right)$. As discussed previously, the function $u_{i, \text{M}}(\theta_i, \argdot, \beta_{-i}(\theta_{-i}^{(j)}))$ can only have a discontinuity at a point in the set $\left\{\beta_l(\theta_l^{(j)}) \right\}_{\left\{l \in [n]\setminus\{i\} \right\}}$. Therefore, it is sufficient to guarantee with probability $1 - \frac{\delta}{n}$, at most $v_i\left(w_i \right)$ points in the set $\mathcal{C} := \bigcup_{j=1}^N \left\{\beta_l(\theta_l^{(j)}) \right\}_{\left\{l \in [n]\setminus\{i\}\right\}}$ fall within an interval of width $w_i$. The statement then follows over a union bound over the $n$ bidders. We apply Lemma \ref{thm:aux-dispersion-lemma-indepdent-and-buckets} in Appendix~\ref{sec:generic-dispersion-statements} to show this statement. For $l \in [n]\setminus \{i\}$, define $\mathcal{C}_l := \left\{\beta_l(\theta_l^{(j)}) \right\}_{\left\{j \in [N]\right\}}$. Then, within each $\mathcal{C}_l$, the samples are independently drawn from the bidding distribution $F^{\beta_l}_{\theta_l}$. Per assumption, the marginal prior $F_{\theta_l}$ is a $\kappa$-bounded distribution. By Theorem~\ref{thm:bounded-prior-density-remains-bounded-under-bi-Lipschitz}, the bidding distribution's density function satisfies $\norm{\phi_{F^{\beta_l}_{\theta_l}}}_\infty \leq L_{\beta_l^{-1}} \cdot \kappa \leq L_{\beta^{-1}_{\text{max}}} \cdot \kappa$. Therefore, the samples $\beta_l(\theta_l^{(j)})$ are drawn from a $\kappa L_{\beta^{-1}_{\text{max}}}$-bounded distribution. Therefore, with probability at most $1 - \frac{\delta}{n}$ every interval of width $w_i$ contains at most
	\begin{align*}
		v_i\left(w_i \right) = &(n-1) w_i \kappa L_{\beta^{-1}_{\text{max}}} N + (n-1) \sqrt{2N \log\left(\frac{2n(n-1)}{\delta} \right)} \\
		&+ 4(n-1) \sqrt{N \log\left( \frac{eN}{2}\right)}.
	\end{align*}
	The second statement can be seen as follows. For any given bids $b_i, b_{-i}$, the allocation is fixed. Therefore, $u_{i, \text{M}}(\argdot, b_i, b_{-i})$ is either constant if $b_i \leq \norm{b_{-i}}_\infty$ or linear with slope $1$ if $b_i > \norm{b_{-i}}_\infty$.
\end{proof}

\begin{customthm}{7.1}
	Let $(\beta_i, \beta_{-i}) \in \tilde{\Sigma}$. Assume that for each agent $i \in [n]$ and segment $B_k \in \mathcal{B}_i$, there exists $\kappa_{i, B_k}>0$, such that the conditional marginal distributions $F_{\giventhat{o_j}{\left\{o_i \in B_k \right\}}}$ for $j \in [n] \setminus \{i\}$ are $\kappa_{i, B_k}$ bounded. Then, for $w_i>0$, with probability at least $1 - \delta$ over the draw of the sets $\left\{\giventhat{\mathcal{D}^{\beta_{-i}}(\mathcal{B}_i)}{i \in [n]} \right\}$ for every $i \in [n]$ and $B_k \in \mathcal{B}_i$, the functions
	$u_{i, \text{M}}\left(\theta_i^{(1)}, \cdot, \beta_{-i}(o_{-i}^{(1)}) \right), \dots, u_{i, \text{M}}\left(\theta_i^{(N_{B_k})}, \cdot, \beta_{-i}(o_{-i}^{(N_{B_k})}) \right)$ are piecewise 1-Lipschitz and $\left(w_i, v_{i, B_k}\left(w_{i} \right) \right)$-dispersed, with
	$v_{i, B_k}\left(w_{i} \right) := (n-1)w_i N_{B_k} \kappa_{i, B_k} L_{\beta^{-1}_{\text{max}}} 
	+ (n-1) \sqrt{2 N_{B_k} \log \left(\frac{2n(n-1) N_{\mathcal{B}_{\text{max}}}}{\delta} \right)} + 4(n-1) \sqrt{N_{B_k} \log \left(\frac{e N_{B_k}}{2} \right)}$.
\end{customthm}
\begin{proof}
	We start with the first part of the statement. Consider $i \in [n]$ and $\beta_{-i}(o_{-i}^{(j)}) \in \mathcal{D}^{\beta_{-i}}(B_k)$ arbitrary. For any $\theta_i \in \Theta_i$ and bid $b_i \in \Theta_i$, we have $u_{i, \text{M}}(\theta_i, b_i, \beta_{-i}(o_{-i}^{(j)})) = \mathds{1}_{\left\{b_i > \norm{\beta_{-i}(o_{-i}^{(j)})}_{\infty} \right\}}\left(\theta_i - b_i\right)$. Therefore, if $b_i \leq \norm{\beta_{-i}(o_{-i}^{(j)})}_{\infty}$, then $u_{i, \text{M}}(\theta_i, b_i, \beta_{-i}(o_{-i}^{(j)}))$ is a constant function in $b_i$. On the other hand, if $b_i > \norm{\beta_{-i}(o_{-i}^{(j)})}_{\infty}$, then $u_{i, \text{M}}(\theta_i, b_i, \beta_{-i}(o_{-i}^{(j)}))$ is linear in $b_i$ with slope of $-1$. Consequently, we have for all $\left(\theta_i^{(j)}, \beta_{-i}(o_{-i}^{(j)}) \right) \in \mathcal{D}^{\beta_{-i}}$, the function $u_{i, \text{M}}(\theta_i^{(j)}, \argdot, \beta_{-i}(o_{-i}^{(j)}))$ is piecewise $1$-Lipschitz continuous with a discontinuity at $\norm{\beta_{-i}(o_{-i}^{(j)})}_{\infty}$.
	
	Fix agent $i \in [n]$ and $B_k \in \mathcal{B}_i$. For any $\theta_i \in \Theta_i$ and $b_{-i} \in \mathcal{A}_{-i}$, the function $u_{i, \text{M}}(\theta_i, \argdot, b_{-i})$ can only have a discontinuity at a point in the set $\left\{b_l \right\}_{l \in [n] \setminus \{i\}}$. Therefore, it is sufficient to guarantee with probability at least $1 - \frac{\delta}{n N_{\mathcal{B}_{\text{max}}}}$, at most $v_{i, B_k}\left(w_{i} \right)$ points in the set $C := \bigcup_{j=1}^{N_{B_k}} \left\{\beta_l(o_l^{(j)}) \right\}_{l \in [n] \setminus \{i\}}$ fall within any interval of width $w_i$. The statement then follows over a union bound over the $n$-bidders and up to $N_{\mathcal{B}_{\text{max}}}$ segments.
	
	We apply Lemma~\ref{thm:aux-dispersion-lemma-indepdent-and-buckets}. For $l \in [n] \setminus \{i\}$ define $C_l := \left\{\beta_l(o_l^{(j)}) \right\}_{j \in [N_{B_k}]}$. Within each $C_l$, the samples are independently drawn from the marginal conditional bidding distribution $F_{\giventhat{o_l}{\{o_i \in B_k\}}}^{\beta_l}$. Per assumption, $F_{\giventhat{o_l}{\{o_i \in B_k\}}}$ is a $\kappa_{i, B_k}$-bounded distribution. By Theorem~\ref{thm:bounded-prior-density-remains-bounded-under-bi-Lipschitz}, the conditional bidding distribution's density function satisfies $\norm{\phi_{F_{\giventhat{o_l}{\{o_i \in B_k\}}}^{\beta_l}}}_\infty \leq L_{\beta^{-1}_l} \cdot \kappa_{i, B_k} \leq L_{\beta^{-1}_{\text{max}}} \cdot \kappa_{i, B_k}$.
	The samples $\beta_l(o_l^{(j)})$ for $1 \leq j \leq N_{B_k}$ are drawn from a $L_{\beta^{-1}_{\text{max}}} \cdot \kappa_{i, B_k}$-bounded distribution. Therefore, with probability at most $1- \frac{\delta}{n N_{\mathcal{B}_{\text{max}}}}$ any interval of width $w_i$ contains at most
	\begin{align*}
		v_{i, B_k}\left(w_{i} \right) :=& (n-1)w_i \kappa_{i, B_k} L_{\beta^{-1}_{\text{max}}} \cdot N_{B_k} \\
		&+ (n-1) \sqrt{2 N_{B_k} \log \left(\frac{2n(n-1) N_{\mathcal{B}_{\text{max}}}}{\delta} \right)} + 4(n-1) \sqrt{N_{B_k} \log \left(\frac{e N_{B_k}}{2} \right)} \text{ samples}.
	\end{align*}
\end{proof}

\begin{theorem}[{\citet[Theorem~3.9]{balcanEstimatingApproximateIncentive2019}}]
	$\text{Pdim}\left(\hat{\mathcal{F}}_{i, \text{M}} \right) = 2$ for all $i \in [n]$.
\end{theorem}

\subsection{First-price combinatorial auction}

There are $l$ items for sale. An agent's valuation space is represented by $\Theta_i = [0, 1]^{2^l}$, indicating its value for each possible bundle $a \subset [l]$. The valuation and bid for a bundle $a$ are denoted by $\theta_i[a]$ and $b_i[a]$, respectively. The allocation $x_i(b_i, b_{-i}) \in {0, 1}^{2^l}$ is determined as the solution to the \emph{winner determination} problem:
\begin{align*}
	& \text{maximize } \sum_{i \in [n]} x_i \cdot b_i \\
	& \text{subject to } x_i \cdot x_j = 0 \text{ for all } i, j \in [n], i \neq j.
\end{align*}
The price for agent $i$ is then given by $p_i(b_i, b_{-i}) = b_i \cdot x_i(b_i, b_{-i})$. 

We start with the dispersion guarantees. 
\begin{theorem}
	Let $(\beta_i, \beta_{-i}) \in \tilde{\Sigma}$. Assume that for each pair of agents $i, j \in [n]$ and each pair of bundles $a, a^\prime \subset [l]$, the joint marginal prior distribution $F_{\theta_i[a], \theta_j[a^\prime]}$ is $\kappa$-bounded. With probability $1 - \delta$ for all agents $i \in [n]$ over the draw of the $n$ datasets $\mathcal{D}^{\beta_{-i}} := \left\{\beta_{-i}(\theta_{-i}^{(1)}), \dots, \beta_{-i}(\theta_{-i}^{(N)}) \right\}$,
	\begin{enumerate}
		\item For any $\theta_i\in [0, 1]^{2^l}$, the functions $u_{i, \text{M}}(\theta_i, \argdot, \beta_{-i}(\theta_{-i}^{(1)})), \dots, u_{i, \text{M}}(\theta_i, \argdot, \beta_{-i}(\theta_{-i}^{(N)}))$ are piecewise $1$-Lipschitz and $\left(O\left(1 / \left(\kappa L^{2^{l+1}}_{\beta^{-1}_{\text{max}}} \sqrt{N}\right) \right), \tilde{O}\left((n+1)^{2l} \sqrt{N \cdot l} \right)\right)$-dispersed.
		\item For any $b_i \in [0, 1]^{2^l}$ and $b_{-i} \in [0, 1]^{(n-1)2^l}$, the function $u_{i, \text{M}}(\argdot, b_i, b_{-i})$ is $1$-Lipschitz continuous.
	\end{enumerate}
\end{theorem}
\begin{proof}
	For the first statement, apply Theorem~\ref{thm:bounded-prior-density-remains-bounded-under-bi-Lipschitz} to the joint marginal bidding distribution $F_{\theta_i[a], \theta_j[a^\prime]}^{\beta_{i, j}}$. Then, the $\kappa$-bounded density function for every pair of agents $i, j \in [n]$ and for all bundles $a, a^\prime \subset [l]$ transforms into a $\kappa L_{\beta^{-1}_i}^{2^l} L_{\beta^{-1}_j}^{2^l}$-bounded bidding distribution. Form this point onward, the proof for the first statement follows analogously to the proof of Theorem~3.10 of \citet{balcanEstimatingApproximateIncentive2019}.
	The second statement is a direct consequence of Theorem~3.11 from \citet{balcanEstimatingApproximateIncentive2019}.
\end{proof}

\begin{theorem}
	Let $(\beta_i, \beta_{-i}) \in \tilde{\Sigma}$. Assume that for each agent $i \in [n]$ and each pair of agents $j, j^\prime \in [n]\setminus \{i\}$, each pair of bundles $a, a^\prime \subset [l]$, and segment $B_k \in \mathcal{B}_i$, the joint marginal prior distribution $F_{\giventhat{o_j(a), o_{j^\prime}(a^\prime)}{\left\{o_i \in B_k \right\}}}$ is $\kappa_{i, B_k}$-bounded.
	Then, with probability at least $1 - \delta$ over the draw of the sets $\left\{\giventhat{\mathcal{D}^{\beta_{-i}}(B_k)}{B_k \in \mathcal{B}_i, i \in [n]} \right\}$ for every $i \in [n]$ and $B_k \in \mathcal{B}_i$, the functions $u_{i, \text{M}}\left(\theta_i^{(1)}, \cdot, \beta_{-i}(o_{-i}^{(1)}) \right), \dots, u_{i, \text{M}}\left(\theta_i^{(N_{B_k})}, \cdot, \beta_{-i}(o_{-i}^{(N_{B_k})}) \right)$ are piecewise 1-Lipschitz and \\
	$\left(O\left(1 / \left(\kappa_{i, B_k} L^{2^{l+1}}_{\beta^{-1}_{\text{max}}} \sqrt{N_{B_k}}\right) \right), \tilde{O}\left((n+1)^{2l} \sqrt{N_{B_k}l} \right)\right)$-dispersed.
\end{theorem}

\begin{proof}
	For agent $i \in [n]$, we apply Theorem~\ref{thm:bounded-prior-density-remains-bounded-under-bi-Lipschitz} to the joint marginal bidding distribution $F_{\giventhat{o_j[a], o_{j^\prime}[a^\prime]}{\left\{o_i \in B_k \right\}}}^{\beta_{i, j}}$. Then, the $\kappa_{i, B_k}$-bounded density function for every pair of agents $j, j^\prime \in [n] \setminus \{i \}$, for all bundles $a, a^\prime \subset [l]$, and $B_k \in \mathcal{B}_i$, transforms into a $\kappa_{i, B_k} L_{\beta^{-1}_j}^{2^l} L_{\beta^{-1}_{j^\prime}}^{2^l}$-bounded bidding distribution. From this point onward, the proof follows analogously to the proof of Theorem~3.10 of \citet{balcanEstimatingApproximateIncentive2019}.
\end{proof}

\begin{theorem}[{\citet[Theorem~3.12]{balcanEstimatingApproximateIncentive2019}}] \label{thm:first-price-combinatorial-pseudo-dim-balcan-hatF}
	For any agent $i \in [n]$, the pseudo-dimension of the function class $\hat{\mathcal{F}}_{i, \text{M}}$ is $O(l 2^l \log(n))$.
\end{theorem}

\begin{theorem}
	For any agent $i \in [n]$, the pseudo-dimension of the function class $\tilde{\mathcal{F}}_{i, \text{M}}$ is $O(l 2^l \log(n))$.
\end{theorem}
\begin{proof}
	\citet{balcanEstimatingApproximateIncentive2019} established in the proof of Theorem~\ref{thm:first-price-combinatorial-pseudo-dim-balcan-hatF} that for every $i \in [n]$, the function class $\hat{\mathcal{F}}_{i, \text{M}}$ is $(2^{l+1}, (n+1)^{2l}))$-delineable. By applying Theorem~\ref{thm:if-hatF-delieanble-then-tildeF-delineable}, we have $\tilde{\mathcal{F}}_{i, \text{M}}$ is $(2^{l}, (n+1)^{2l}))$-delineable. Subsequently, with an application of Theorem~\ref{thm:balcans-general-pseudo-dim-bound-for-mt-delineable-functions}, we find that the pseudo-dimension of $\tilde{\mathcal{F}}_{i, \text{M}}$ is $O(l 2^l \log(n)$.
\end{proof}

\subsection{Discriminatory auction} \label{sec:disp-and-pdim-discriminatory-auction}
In the discriminatory auction model, $m$ identical units of an item are for sale, with each agent $i \in [n]$ having a valuation vector $\theta_i \in [0, 1]^m$, indicating its willingness to pay for each additional unit. The valuation decreases with each additional unit, implying $\theta_i[1] \geq \theta_i[2] \geq \cdots \geq \theta_i[m]$. In total $nm$ bids $b_i[\mu]$ for $i \in [n]$ and $\mu \in [m]$ are submitted to the auctioneer. If $m_i$ of agent $i$'s bids are among the $m$ highest, it receives the units at its bid price, paying a cumulative amount based on the quantity awarded, i.e., $p_i = \sum_{\mu = 1}^{m_i} b_i[\mu]$.

\begin{theorem}
	Let $(\beta_i, \beta_{-i}) \in \tilde{\Sigma}$. Assume that for each agent $i\in [n]$ and unit $l \in [m]$, the marginal prior distribution $F_{\theta_i[l]}$ is $\kappa$-bounded. With probability $1 - \delta$ for all agents $i \in [n]$ over the draw of the $n$ datasets $\mathcal{D}^{\beta_{-i}} := \left\{\beta_{-i}(\theta_{-i}^{(1)}), \dots, \beta_{-i}(\theta_{-i}^{(N)}) \right\}$,
	\begin{enumerate}
		\item For any $\theta_i\in [0, 1]^m$, the functions $u_{i, \text{M}}(\theta_i, \argdot, \beta_{-i}(\theta_{-i}^{(1)})), \dots, u_{i, \text{M}}(\theta_i, \argdot, \beta_{-i}(\theta_{-i}^{(N)}))$ are piecewise $1$-Lipschitz and $\left(O\left(1 / \left(\kappa L_{\beta^{-1}_{\text{max}}} \sqrt{N}\right) \right), \tilde{O}\left(n m^2 \sqrt{N} \right)\right)$-dispersed.
		\item For any $b_i \in [0, 1]^m$ and $b_{-i} \in [0, 1]^{(n-1)m}$, the function $u_{i, \text{M}}(\argdot, b_i, b_{-i})$ is $1$-Lipschitz continuous.
	\end{enumerate}
\end{theorem}
\begin{proof}
	For the first statement, we apply Theorem~\ref{thm:bounded-prior-density-remains-bounded-under-bi-Lipschitz} to the marginal bidding distribution $F_{\theta_i[l]}^{\beta_{i}}$. Then, the $\kappa$-bounded density function for agent $i \in [n]$ and for units $l \in [l]$ transforms into a $\kappa L_{\beta^{-1}_i}$-bounded bidding distribution. Form this point onward, the proof for the first statement follows analogously to the proof of Theorem~3.16 of \citet{balcanEstimatingApproximateIncentive2019}.
	The second statement is a direct consequence of Theorem~3.17 from \citet{balcanEstimatingApproximateIncentive2019}.
\end{proof}

\begin{theorem}
	Let $(\beta_i, \beta_{-i}) \in \tilde{\Sigma}$. Assume that for each agent $i \in [n]$, agent $j \in [n]\setminus \{i\}$, unit $l \in [m]$, and segment $B_k \in \mathcal{B}_i$, the marginal prior distribution $F_{\giventhat{o_j[l]}{\left\{o_i \in B_k \right\}}}$ is $\kappa_{i, B_k}$-bounded.
	Then, with probability at least $1 - \delta$ over the draw of the sets $\left\{\giventhat{\mathcal{D}^{\beta_{-i}}(B_k)}{B_k \in \mathcal{B}_i, i \in [n]} \right\}$ for every $i \in [n]$ and $B_k \in \mathcal{B}_i$, the functions $u_{i, \text{M}}\left(\theta_i^{(1)}, \cdot, \beta_{-i}(o_{-i}^{(1)}) \right), \dots, u_{i, \text{M}}\left(\theta_i^{(N_{B_k})}, \cdot, \beta_{-i}(o_{-i}^{(N_{B_k})}) \right)$ are piecewise 1-Lipschitz and
	$\left(O\left(1 / \left(\kappa_{i, B_k} L_{\beta^{-1}_{\text{max}}} \sqrt{N_{B_k}}\right) \right), \tilde{O}\left(n m^2 \sqrt{N_{B_k}l} \right)\right)$-dispersed.
\end{theorem}

\begin{proof}
	For agent $i \in [n]$, apply Theorem~\ref{thm:bounded-prior-density-remains-bounded-under-bi-Lipschitz} to the marginal bidding distribution $F_{\giventhat{o_j[l]}{\left\{o_i \in B_k \right\}}}^{\beta_{i}}$. Then, the $\kappa_{i, B_k}$-bounded density function for every agent $j \in [n]\setminus \{i \}$, every unit $l \in [m]$, and $B_k \in \mathcal{B}_i$ transforms into a $\kappa_{i, B_k} L_{\beta^{-1}_j}$-bounded bidding distribution. Form this point onward, the proof for the first statement follows analogously to the proof of Theorem~3.16 of \citet{balcanEstimatingApproximateIncentive2019}.
\end{proof}

\begin{theorem}[{\citet[Theorem~3.18]{balcanEstimatingApproximateIncentive2019}}] \label{thm:discriminatory-auction-pseudo-dim-balcan-hatF}
	For any agent $i \in [n]$, we have $\text{Pdim}(\hat{\mathcal{F}}_{i, \text{M}})$ is $O(m \log(nm))$.
\end{theorem}

\begin{theorem}
	For any agent $i \in [n]$, the pseudo-dimension of the function class $\tilde{\mathcal{F}}_{i, \text{M}}$ is $O(m \log(nm))$.
\end{theorem}
\begin{proof}
	\citet{balcanEstimatingApproximateIncentive2019} established in the proof of Theorem~\ref{thm:discriminatory-auction-pseudo-dim-balcan-hatF} that for every $i \in [n]$, the function class $\hat{\mathcal{F}}_{i, \text{M}}$ is $(2m, m^2(n-1))$-delineable. By applying Theorem~\ref{thm:if-hatF-delieanble-then-tildeF-delineable}, we have $\tilde{\mathcal{F}}_{i, \text{M}}$ is $(m, m^2(n-1))$-delineable. Subsequently, with an application of Theorem~\ref{thm:balcans-general-pseudo-dim-bound-for-mt-delineable-functions}, we find that the pseudo-dimension of $\tilde{\mathcal{F}}_{i, \text{M}}$ is $O(m \log(nm))$.
\end{proof}

\subsection{Uniform-price auction}
In the uniform-price auction model, the allocation mechanism parallels that of the discriminatory auction (Section~\ref{sec:disp-and-pdim-discriminatory-auction}). 
The uniform-price auction sells all $m$ units at a market-clearing price, with demand meeting supply. Following the principle that the market-clearing price is the highest bid not resulting in a sale~\citep{krishna2009auction}, we define $c_{-i} \in \mathbb{R}^m$ as the array of the top $m$ competing bids $b_{-i}$ against agent $i$, ordered in descending value. This means $c_{-i}[1] = \left\| b_{-i} \right\|_{\infty}$ is the highest of the opponents' bids, $c_{-i}[2]$ is the second-highest, and so on.
Agent $i$ secures exactly one unit if and only if its highest bid surpasses the lowest winning bid and its second-highest bid does not exceed the second-lowest winning bid, i.e., $b_i[1] > c_{-i}[m]$ and $b_i[2] < c_{-i}[m - 1]$. This condition extends to multiple units where agent $i$ wins exactly $m_i\geq 0$ units if its $m_i$th bid exceeds the corresponding winning bid and the next highest bid does not. The market-clearing price is set to $p = \max \left\{ b_i[m_i + 1], c_{-i}[m - m_i + 1] \right\}$, which is the maximum of the lowest winning bid and the highest losing bid. The final payment by agent $i$ is $m_i \cdot p$.

\begin{theorem}
	Let $(\beta_i, \beta_{-i}) \in \tilde{\Sigma}$. Assume that for each agent $i\in [n]$ and unit $l \in [m]$, the marginal prior distribution $F_{\theta_i[l]}$ is $\kappa$-bounded. With probability $1 - \delta$ for all agents $i \in [n]$ over the draw of the $n$ datasets $\mathcal{D}^{\beta_{-i}} := \left\{\beta_{-i}(\theta_{-i}^{(1)}), \dots, \beta_{-i}(\theta_{-i}^{(N)}) \right\}$,
	\begin{enumerate}
		\item For any $\theta_i\in [0, 1]^m$, the functions $u_{i, \text{M}}(\theta_i, \argdot, \beta_{-i}(\theta_{-i}^{(1)})), \dots, u_{i, \text{M}}(\theta_i, \argdot, \beta_{-i}(\theta_{-i}^{(N)}))$ are piecewise $1$-Lipschitz and $\left(O\left(1 / \left(\kappa L_{\beta^{-1}_{\text{max}}} \sqrt{N}\right) \right), \tilde{O}\left(n m^2 \sqrt{N} \right)\right)$-dispersed.
		\item For any $b_i \in [0, 1]^m$ and $b_{-i} \in [0, 1]^{(n-1)m}$, the function $u_{i, \text{M}}(\argdot, b_i, b_{-i})$ is $1$-Lipschitz continuous.
	\end{enumerate}
\end{theorem}
\begin{proof}
	For the first statement, apply Theorem~\ref{thm:bounded-prior-density-remains-bounded-under-bi-Lipschitz} to the marginal bidding distribution $F_{\theta_i[l]}^{\beta_{i}}$. Then, the $\kappa$-bounded density function for agent $i \in [n]$ and for units $l \in [l]$ transforms into a $\kappa L_{\beta^{-1}_i}$-bounded bidding distribution. Form this point onward, the proof for the first statement follows analogously to the remaining proof of Theorem~D.5 of \citet{balcanEstimatingApproximateIncentive2019}.
	The second statement is a direct consequence of Theorem~D.6 from \citet{balcanEstimatingApproximateIncentive2019}.
\end{proof}

\begin{theorem}
	Let $(\beta_i, \beta_{-i}) \in \tilde{\Sigma}$. Assume that for each agent $i \in [n]$, agent $j \in [n]\setminus \{i\}$, unit $l \in [m]$, and segment $B_k \in \mathcal{B}_i$, the marginal prior distribution $F_{\giventhat{o_j[l]}{\left\{o_i \in B_k \right\}}}$ is $\kappa_{i, B_k}$-bounded.
	Then, with probability at least $1 - \delta$ over the draw of the sets $\left\{\giventhat{\mathcal{D}^{\beta_{-i}}(B_k)}{B_k \in \mathcal{B}_i, i \in [n]} \right\}$ for every $i \in [n]$ and $B_k \in \mathcal{B}_i$, the functions $u_{i, \text{M}}\left(\theta_i^{(1)}, \cdot, \beta_{-i}(o_{-i}^{(1)}) \right), \dots, u_{i, \text{M}}\left(\theta_i^{(N_{B_k})}, \cdot, \beta_{-i}(o_{-i}^{(N_{B_k})}) \right)$ are piecewise 1-Lipschitz and
	$\left(O\left(1 / \left(\kappa_{i, B_k} L_{\beta^{-1}_{\text{max}}} \sqrt{N_{B_k}}\right) \right), \tilde{O}\left(n m^2 \sqrt{N_{B_k}l} \right)\right)$-dispersed.
\end{theorem}

\begin{proof}
	For agent $i \in [n]$, apply Theorem~\ref{thm:bounded-prior-density-remains-bounded-under-bi-Lipschitz} to the marginal bidding distribution $F_{\giventhat{o_j[l]}{\left\{o_i \in B_k \right\}}}^{\beta_{i}}$. Then, the $\kappa_{i, B_k}$-bounded density function for every agent $j \in [n]\setminus \{i \}$, every unit $l \in [m]$, and $B_k \in \mathcal{B}_i$ transforms into a $\kappa_{i, B_k} L_{\beta^{-1}_j}$-bounded bidding distribution. Form this point onward, the proof for the first statement follows analogously to the proof of Theorem~D.5 of \citet{balcanEstimatingApproximateIncentive2019}.
\end{proof}

\begin{theorem}[{\citet[Theorem~D.7]{balcanEstimatingApproximateIncentive2019}}] \label{thm:uniform-price-auction-pseudo-dim-balcan-hatF}
	For any agent $i \in [n]$, $\text{Pdim}(\hat{\mathcal{F}}_{i, \text{M}})$ is $O(m \log(nm))$.
\end{theorem}

\begin{theorem}
	For any agent $i \in [n]$, the pseudo-dimension of the function class $\tilde{\mathcal{F}}_{i, \text{M}}$ is $O(m \log(nm))$.
\end{theorem}
\begin{proof}
	\citet{balcanEstimatingApproximateIncentive2019} established in the proof of Theorem~\ref{thm:uniform-price-auction-pseudo-dim-balcan-hatF} that for every $i \in [n]$, the function class $\hat{\mathcal{F}}_{i, \text{M}}$ is $(2m, m^2(n-1))$-delineable. By applying Theorem~\ref{thm:if-hatF-delieanble-then-tildeF-delineable}, we have $\tilde{\mathcal{F}}_{i, \text{M}}$ is $(m, m^2(n-1))$-delineable. Subsequently, with an application of Theorem~\ref{thm:balcans-general-pseudo-dim-bound-for-mt-delineable-functions}, we find that the pseudo-dimension of $\tilde{\mathcal{F}}_{i, \text{M}}$ is $O(m \log(nm))$.
\end{proof}

\begin{comment}
	\subsection{Second-price auction with spiteful bidders}
	We study a scenario with \emph{spiteful agents}~\citep{Brandt et al, 2007; Morgan et al, 2003; Sharma and Sandholm, 2010; Tang and Sandholm, 2012} next.
	Each bidder's utility increases when his surplus increases, but also decreases when the other bidders' surpluses increase. Formally, given a spite parameter $\alpha_i \in [0, 1] $, agent $i$'s utility under the second-price mechanism $\text{M}$ is $u_{i,M} (\theta_i, b_i, b_{-i}) = \alpha_i \mathds{1}_{\left\{ b_i - \norm{b_{-i}}_{\infty} \right\} } (\theta_i - \norm{b_{-i}}_{\infty}) - (1 - \alpha_i) \sum_{i' \neq i} \mathds{1}_{\{b_{i'} > \norm{b_{-i}}_{\infty}\}} (\theta_{i'} - \norm{b_{-i}}_{\infty})$. The lower $ \alpha_i $ is, the more spiteful bidder $i$.
\end{comment}

\section{Proofs to limit concentration of bidding distributions Section~\ref{sec:a-data-drive-approach-for-equilibrium-verification-bounds}} \label{sec:proofs-appendix-sampling-based-approach-for-verification}

\begin{customthm}{4.4}
	Denote with $\phi_{F_{o_i}}$ and $\phi_{F_{o_i, o_j}}$ the density functions for the marginal prior distributions $F_{o_i}$ and $F_{o_i, o_j}$ for any $i, j \in [n]$.
	Further assume that $\phi_{F_{o_i}}$ and $\phi_{F_{o_i, o_j}}$ are $\kappa$-bounded density functions for some $\kappa > 0$. Further, let $(\beta_i, \beta_{-i}) \in \tilde{\Sigma}$ be a strategy profile of bi-Lipschitz continuous bidding strategies. Then, the probability density functions of the bidding distributions $F^{\beta_i}_{o_i}$ and $F^{\beta_i, \beta_j}_{o_i, o_j}$ satisfy
	\begin{align*}
		\sup_{b_{i} \in \beta_{i}(\mathcal{O}_{i})} \phi_{F^{\beta_{i}}_{o_{i}}}(b_{i}) &\leq \kappa \cdot L_{\beta_i^{-1}}^m\\
		\sup_{(b_{i}, b_j) \in \beta_{i}(\mathcal{O}_{i}) \times \beta_{j}(\mathcal{O}_{j})} \phi_{F^{\beta_{i}, \beta_j}_{o_{i}, o_j}}(b_{i}, b_j) &\leq \kappa \cdot L_{\beta_i^{-1}}^m \cdot L_{\beta_j^{-1}}^m
	\end{align*}
	where $m$ denotes the dimension of $\mathcal{O}_i$.
\end{customthm}

\begin{proof}
	By the definition of bi-Lipschitz continuity, the function $\beta_i : \mathcal{O}_i \rightarrow \beta_i \left(\mathcal{O}_i\right)$ is invertible for any $i \in [n]$. We perform a change of variables (Theorem~\ref{thm:change-of-variables}) with $\mu_0 := \left(\beta_{i}\right)_{\#}F_{o_i}$ and $\mu_1 := \left(\beta_i^{-1} \right)_{\#} \left( \left(\beta_{i}\right)_{\#}F_{o_i} \right) = F_{o_i}$. Then, we have for $b_i \in \beta_i\left(\mathcal{O}_i \right)$
	\begin{align*}
		\phi_{F^{\beta_{i}}_{o_{i}}}(b_{i}) \overset{\text{Theorem } \ref{thm:change-of-variables}}{=} \phi_{F_{o_i}} \left(\beta_i^{-1}(b_i) \right) \cdot \abs{ \det(\mathcal{J} \beta_i^{-1}(b_i))} \overset{\text{Lemma } \ref{thm:bound-determinante-of-jacobian-of-bi-lipschitz-map}}{\leq} \kappa \cdot L_{\beta_i^{-1}}^m,
	\end{align*}
	where $m$ denotes the dimension of $\mathcal{O}_i$. We used a well-known bound on a bi-Lipschitz mapping's Jacobian determinant in the last step.
	
	For $i, j \in [n]$ with $i \neq j$, the functions $\beta_i$ and $\beta_j$ are independent from one another. That is, $\beta_i: \mathcal{O}_i \rightarrow \mathcal{A}_i$ and $\beta_j: \mathcal{O}_j \rightarrow \mathcal{A}_j$. The same holds for their inverses, so that the Jacobian matrix of $\beta_{i, j}^{-1} = \left(\beta_i^{-1}, \beta_j^{-1} \right)$ is a block matrix. That is, for $b_i \in \beta_i \left(\mathcal{O}_i \right)$ and $b_j \in \beta_j \left(\mathcal{O}_j \right)$
	\begin{align*}
		\left(\mathcal{J} \beta_{i, j}^{-1} \right) \left(b_i, b_j \right) = \left(\begin{array}{cc}
			\mathcal{J} \beta_i^{-1}(b_i) & 0 \\
			0 & \mathcal{J} \beta_j^{-1}(b_j)
		\end{array} \right)
	\end{align*}
	A well-known fact about the determinant of a block-matrix is that it equals the product of the blocks' determinants. By another application of the change of variables formula, we have
	\begin{align*}
		&\phi_{F^{\beta_{i}, \beta_j}_{o_{i}, o_j}}(b_{i}, b_j) \overset{\text{Theorem } \ref{thm:change-of-variables}}{=} \phi_{F_{o_i, o_j}} \left( \left(\beta_{i, j}^{-1} \right)^{-1} \left(b_i, b_j \right) \right) \cdot \abs{ \det\left( \mathcal{J} \left(\beta_{i, j}^{-1} \right)^{-1} \left(b_i, b_j \right) \right) } \\
		&\leq \kappa \abs{ \det(\mathcal{J} \beta_i^{-1}(b_i)) \cdot \det(\mathcal{J} \beta_j^{-1}(b_j))} \overset{\text{Lemma } \ref{thm:bound-determinante-of-jacobian-of-bi-lipschitz-map}}{\leq} \kappa \cdot L_{\beta_i^{-1}}^m \cdot L_{\beta_j^{-1}}^m.
	\end{align*}
\end{proof}

\section{Proofs for independent prior distributions Section~\ref{sec:independent-prior-statements}} \label{sec:appendix-independent-prior-proofs}

\begin{customthm}{5.1}
	Let $\delta > 0$, $\text{M}$ be a mechanism, and $\beta \in \Sigma$ a strategy profile. Then, it holds with probability $1-\delta$ for all agents $i \in [n]$ over the draw of datasets $\mathcal{D}^{\beta_{-1}}, \dots, \mathcal{D}^{\beta_{-n}}$ of valuation-bid queries,
	\begin{align*}
		&\sup_{\theta_i \in \Theta_i} \hat{\ell}_i(\theta_i, \beta_i(\theta_i), \beta_{-i}) = \sup_{\theta_i, \hat{\theta}_i \in \Theta_i} \hat{u}_{i, \text{M}}(\theta_i, \hat{\theta}_i, \beta_{-i}) - \hat{u}_{i, \text{M}}(\theta_i, \beta_i(\theta_i), \beta_{-i}) \\
		&\leq \sup_{\theta_i, \hat{\theta}_i \in \Theta_i} \frac{1}{N} \sum_{j=1}^N u_{i, \text{M}}(\theta_i, \hat{\theta}_i, \beta_{-i}(\theta_{-i}^{(j)})) - u_{i, \text{M}}(\theta_i, \beta_i(\theta_i), \beta_{-i}(\theta_{-i}^{(j)})) + \hat{\varepsilon}_{i, \text{Pdim}}(N, \delta),\\
		&\mbox{ where } \hat{\varepsilon}_{i, \text{Pdim}}(N, \delta) := 4\sqrt{\frac{2d_i}{N} \log\left( \frac{e N}{d_i}\right)} + 2\sqrt{\frac{2}{N} \log\left(\frac{2n}{\delta} \right)}, \mbox{ and } d_i=\text{Pdim}(\hat{\mathcal{F}}_{i, \text{M}}).
	\end{align*}
\end{customthm}

\begin{proof}
	Fix an arbitrary agent $i \in [n]$. Then we have with $\mathcal{D}^{\beta_{-i}} := \left\{\beta_{-i}(\theta_{-i}^{(1)}), \dots, \beta_{-i}(\theta_{-i}^{(N)}) \right\}$ that $\beta_{-i}(\theta_{-i}^{(j)}) \sim F^{\beta_{-i}}_{\theta_{-i}}$ is i.i.d. for $1 \leq j \leq N$. 
	Therefore, by applying Theorem \ref{thm:pollard-pac-bound-general-distribution}, we have with probability at least $1 - \frac{\delta}{2}$ for all $u_{i, \text{M}}(\theta_{i}, \hat{\theta}_i, \argdot) \in \hat{\mathcal{F}}_{i, \text{M}}$ that
	\begin{align*}
		\abs{\frac{1}{N} \sum_{j=1}^{N}u_{i, \text{M}}(\theta_{i}, \hat{\theta}_i, \beta_{-i}(\theta_{-i}^{(j)})) - \mathbb{E}_{\beta_{-i}(\theta_{-i})\sim F^{\beta_{-i}}}\left[u_{i, \text{M}}(\theta_{i}, \hat{\theta}_i, \beta_{-i}(\theta_{-i}))\right] } \\
		\leq 2\sqrt{\frac{2d_i}{N} \log\left( \frac{e N}{d_i}\right)} + \sqrt{\frac{2}{N} \log\left(\frac{2}{\delta} \right)} = \frac{1}{2}\hat{\varepsilon}_{i, \text{Pdim}}(N, n\delta).
	\end{align*}
	
	As this holds for all $\theta_{i}, \hat{\theta}_i \in \Theta_i$, we have with probability $1 - \delta$
	\begin{align*}
		&\sup_{\theta_i, \hat{\theta}_i \in \Theta_i} \mathbb{E}_{\beta_{-i}(\theta_{-i})\sim F^{\beta_{-i}}}\left[u_{i, \text{M}}(\theta_{i}, \hat{\theta}_i, \beta_{-i}(\theta_{-i}))\right] - \mathbb{E}_{\beta_{-i}(\theta_{-i})\sim F^{\beta_{-i}}}\left[u_{i, \text{M}}(\theta_{i}, \beta_{i}(\theta_{i}), \beta_{-i}(\theta_{-i}))\right] \\
		&\leq \sup_{\theta_i, \hat{\theta}_i \in \Theta_i} \mathbb{E}_{\beta_{-i}(\theta_{-i})\sim F^{\beta_{-i}}}\left[u_{i, \text{M}}(\theta_{i}, \hat{\theta}_i, \beta_{-i}(\theta_{-i}))\right] - \frac{1}{N} \sum_{j=1}^{N}u_{i, \text{M}}(\theta_{i}, \hat{\theta}_i, \beta_{-i}(\theta_{-i}^{(j)})) \\
		& + \frac{1}{N} \sum_{j=1}^{N}u_{i, \text{M}}(\theta_{i}, \hat{\theta}_i, \beta_{-i}(\theta_{-i}^{(j)})) - u_{i, \text{M}}(\theta_{i}, \beta_{i}(\theta_{i}), \beta_{-i}(\theta_{-i}^{(j)}))\\
		& + \frac{1}{N} \sum_{j=1}^{N}u_{i, \text{M}}(\theta_{i}, \beta_{i}(\theta_{i}), \beta_{-i}(\theta_{-i}^{(j)})) - \mathbb{E}_{\beta_{-i}(\theta_{-i})\sim F^{\beta_{-i}}}\left[u_{i, \text{M}}(\theta_{i}, \beta_{i}(\theta_{i}), \beta_{-i}(\theta_{-i}))\right] \\
		&\leq \sup_{\theta_i, \hat{\theta}_i \in \Theta_i} \frac{1}{N} \sum_{j=1}^{N}u_{i, \text{M}}(\theta_{i}, \hat{\theta}_i, \beta_{-i}(\theta_{-i}^{(j)})) - u_{i, \text{M}}(\theta_{i}, \beta_{i}(\theta_{i}), \beta_{-i}(\theta_{-i}^{(j)})) + \hat{\varepsilon}_{i, \text{Pdim}}(N, n\delta).
	\end{align*}
	Denote the event that the previous inequalities hold for agent $i$ by $A_i(\delta)$. Then, we have shown $P(A_i(\delta)) \geq 1 - \delta$ so far. It remains to show the bounds hold for all agents. We apply a union bound to the events $A_i\left(\frac{\delta}{n}\right)$, which gives
	\begin{align*}
		P\left(\bigcap_{i=1}^n A_i\left(\frac{\delta}{n}\right) \right) &= P\left( \left(\bigcup_{i=1}^n A_i\left(\frac{\delta}{n}\right)^\complement \right)^\complement \right) = 1 - P\left( \bigcup_{i=1}^n A_i\left(\frac{\delta}{n}\right)^\complement  \right) \\
		&\geq 1 - \sum_{i=1}^{n} P\left(A_i\left(\frac{\delta}{n}\right)^\complement \right) \geq 1 - n \frac{\delta}{n} = 1 - \delta.
	\end{align*}
	
\end{proof}

\begin{customthm}{5.3}
	Let $\delta > 0$ and $\text{M}$ be a mechanism. Furthermore, let $\beta \in \tilde{\Sigma}$ be a strategy profile.
	Given that Assumption \ref{ass:general-dispersion-holds-independent-prior} holds for $w_i >0$, $v_i(w_i)$, and $v_i(L_{\beta_i}w_i)$, we have with probability at least $1 -  3\delta$ over the draw of the datasets $\mathcal{D}^{\beta_{-1}}, \dots, \mathcal{D}^{\beta_{-n}}$ for every agent $i \in [n]$
	\begin{align*}
		& \sup_{\theta_i \in \Theta_i} \hat{\ell}_i(\theta_i, \beta_i(\theta_i), \beta_{-i}) = \sup_{\theta_i, \hat{\theta}_i \in \Theta_i} \hat{u}_{i, \text{M}}(\theta_i, \hat{\theta}_i, \beta_{-i}) - \hat{u}_{i, \text{M}}(\theta_i, \beta_i(\theta_i), \beta_{-i}) \\
		&\leq \sup_{\theta_i, \hat{\theta}_i \in \mathcal{G}_{w_i}} \frac{1}{N} \sum_{j=1}^N u_{i, \text{M}}(\theta_i, \hat{\theta}_i, \beta_{-i}(\theta_{-i}^{(j)})) - u_{i, \text{M}}(\theta_i, \beta_i(\theta_i), \beta_{-i}(\theta_{-i}^{(j)}))+\hat{\varepsilon}_i,\\
		&\mbox{ where } \ \hat{\varepsilon}_i:=4\sqrt{\frac{2d_i}{N} \log\left( \frac{e N}{d_i}\right)} + 2\sqrt{\frac{2}{N} \log\left(\frac{2n}{\delta} \right)} + 3\hat{\varepsilon}_{i, \text{disp}}(w_i) + \hat{\varepsilon}_{i, \text{disp}}(L_{\beta_i}w_i), \\
		&\mbox{ with } \ \hat{\varepsilon}_{i, \text{disp}}(x) := \frac{N - v_i\left(x \right)}{N} L_i x + \frac{2 v_i\left(x \right)}{N}, \text{ and } d_i=\operatorname{Pdim}\left(\hat{\mathcal{F}}_{i, \text{M}}\right).
	\end{align*}
\end{customthm}
\begin{proof}
	For a $(w, v)$-dispersed set of $N$ functions, with probability $1 - \delta$, at most $v$ jump discontinuities fall within a ball of radius $w$. Therefore, within any ball of radius $w$, at least $N-v$ functions are Lipschitz continuous, and at most $v$ are not. 
	Let $w_i > 0$ and $v_i\left(w_i \right)$ be the function from the dispersion guarantees from Assumption \ref{ass:general-dispersion-holds-independent-prior}. Define $\hat{\varepsilon}_{i, \text{disp}}(w_i) := \frac{N - v_i\left(w_i \right)}{N} L_i w_i + \frac{2 v_i\left(w_i \right)}{N}$. Then, with probability at least $1 - \delta$, the following conditions hold:
	\begin{enumerate}
		\item For all $i \in [n]$, valuations $\theta_i \in \Theta_i$, and reported valuations $\hat{\theta}_i, \hat{\theta}_i^\prime \in \Theta_i$ with $\norm{\hat{\theta}_i - \hat{\theta}_i^\prime}_1 \leq w_i$, we have
		\begin{align} \label{equ:proof-main-independent-priors-dispersion-for-bids}
			\begin{split}
				\abs{ \frac{1}{N} \sum_{j=1}^N u_{i, \text{M}}(\theta_i, \hat{\theta}_i, \beta_{-i}(\theta_{-i}^{(j)})) - u_{i, \text{M}}(\theta_i, \hat{\theta}_i^\prime, \beta_{-i}(\theta_{-i}^{(j)})) }
				\leq \hat{\varepsilon}_{i, \text{disp}}(w_i)
			\end{split}
		\end{align}
		\item For all $i \in [n]$, reported valuations $\hat{\theta}_i \in \Theta_i$, and valuations $\theta_i, \theta_i^\prime \in \Theta_i$ with $\norm{\theta_i - \theta_i^\prime}_1 \leq w_i$, we have
		\begin{align} \label{equ:proof-main-independent-priors-dispersion-for-valuations}
			\begin{split}
				\abs{ \frac{1}{N} \sum_{j=1}^N u_{i, \text{M}}(\theta_i, \hat{\theta}_i, \beta_{-i}(\theta_{-i}^{(j)})) - u_{i, \text{M}}(\theta_i^\prime, \hat{\theta}_i, \beta_{-i}(\theta_{-i}^{(j)})) }
				\leq \hat{\varepsilon}_{i, \text{disp}}(w_i).
			\end{split}
		\end{align}
	\end{enumerate}
	Let $\theta_i, \hat{\theta}_i \in \Theta_i$. By the definition of $\mathcal{G}_{w_i}$, there exist points $p, \hat{p} \in \mathcal{G}_{w_i}$ such that $\norm{\theta_i - p}_1 \leq w_i$ and $\norm{\hat{\theta}_i - \hat{p}}_1 \leq w_i$. 
	Equation \ref{equ:proof-main-independent-priors-dispersion-for-valuations} results in
	\begin{align*}
		\abs{ \frac{1}{N} \sum_{j=1}^N u_{i, \text{M}}(\theta_i, \hat{p}, \beta_{-i}(\theta_{-i}^{(j)})) - u_{i, \text{M}}(p, \hat{p}, \beta_{-i}(\theta_{-i}^{(j)})) }
		\leq \hat{\varepsilon}_{i, \text{disp}}(w_i),
	\end{align*}
	and
	\begin{align*}
		\abs{ \frac{1}{N} \sum_{j=1}^N u_{i, \text{M}}(p, \beta_i(p), \beta_{-i}(\theta_{-i}^{(j)})) - u_{i, \text{M}}(\theta_i, \beta_i(p), \beta_{-i}(\theta_{-i}^{(j)})) }
		\leq \hat{\varepsilon}_{i, \text{disp}}(w_i).
	\end{align*}
	Equation \ref{equ:proof-main-independent-priors-dispersion-for-bids} gives
	\begin{align*}
		\abs{ \frac{1}{N} \sum_{j=1}^N u_{i, \text{M}}(\theta_i, \hat{\theta}_i, \beta_{-i}(\theta_{-i}^{(j)})) - u_{i, \text{M}}(\theta_i, \hat{p}_i, \beta_{-i}(\theta_{-i}^{(j)})) }
		\leq \hat{\varepsilon}_{i, \text{disp}}(w_i).
	\end{align*}
	Due to the Lipschitz continuity of $\beta_i$, we have $\norm{\beta_i(\theta_i) - \beta_i(p)}_1 \leq L_{\beta_i} w_i$. An additional application of Assumption \ref{ass:general-dispersion-holds-independent-prior} and Equation \ref{equ:proof-main-independent-priors-dispersion-for-bids} gives with probability at least $1 - \delta$
	\begin{align*}
		\abs{ \frac{1}{N} \sum_{j=1}^N u_{i, \text{M}}(\theta_i, \beta_i(p), \beta_{-i}(\theta_{-i}^{(j)})) - u_{i, \text{M}}(\theta_i, \beta_i, \beta_{-i}(\theta_{-i}^{(j)})) }
		\leq \hat{\varepsilon}_{i, \text{disp}}(L_{\beta_i} w_i).
	\end{align*}
	Therefore, combining these statements, we have with probability at least $1 - 2\delta$
	\begin{align*}
		&\abs{ \frac{1}{N} \sum_{j=1}^N u_{i, \text{M}}(\theta_i, \hat{\theta}_i, \beta_{-i}(\theta_{-i}^{(j)})) - u_{i, \text{M}}(\theta_i, \beta_i(\theta_i), \beta_{-i}(\theta_{-i}^{(j)})) } \\
		&\leq \abs{ \frac{1}{N} \sum_{j=1}^N u_{i, \text{M}}(\theta_i, \hat{\theta}_i, \beta_{-i}(\theta_{-i}^{(j)})) - u_{i, \text{M}}(\theta_i, \hat{p}, \beta_{-i}(\theta_{-i}^{(j)})) } \\
		&+ \abs{ \frac{1}{N} \sum_{j=1}^N u_{i, \text{M}}(\theta_i, \hat{p}, \beta_{-i}(\theta_{-i}^{(j)})) - u_{i, \text{M}}(p, \hat{p}, \beta_{-i}(\theta_{-i}^{(j)})) } \\
		&+ \abs{ \frac{1}{N} \sum_{j=1}^N u_{i, \text{M}}(p, \hat{p}, \beta_{-i}(\theta_{-i}^{(j)})) - u_{i, \text{M}}(p, \beta_i(p), \beta_{-i}(\theta_{-i}^{(j)})) } \\
		&+ \abs{ \frac{1}{N} \sum_{j=1}^N u_{i, \text{M}}(p, \beta_i(p), \beta_{-i}(\theta_{-i}^{(j)})) - u_{i, \text{M}}(\theta_i, \beta_i(p), \beta_{-i}(\theta_{-i}^{(j)})) } \\
		&+ \abs{ \frac{1}{N} \sum_{j=1}^N u_{i, \text{M}}(\theta_i, \beta_i(p), \beta_{-i}(\theta_{-i}^{(j)})) - u_{i, \text{M}}(\theta_i, \beta_i(\theta_i), \beta_{-i}(\theta_{-i}^{(j)})) } \\
		& \leq \abs{ \frac{1}{N} \sum_{j=1}^N u_{i, \text{M}}(p, \hat{p}, \beta_{-i}(\theta_{-i}^{(j)})) - u_{i, \text{M}}(p, \beta_i(p), \beta_{-i}(\theta_{-i}^{(j)})) } \\
		&+ 3\hat{\varepsilon}_{i, \text{disp}}(w_i) + \hat{\varepsilon}_{i, \text{disp}}(L_{\beta_i} w_i).
	\end{align*}
	The statement is complete with an additional application of Theorem \ref{thm:pseudo-dim-guarantee-independent-priors-interim-utility}. That is, in total, three different events with probability $1 - \delta$ need to hold. The first comes from the pseudo-dimension concentration bound of Theorem \ref{thm:pseudo-dim-guarantee-independent-priors-interim-utility}. The two other events are the dispersion guarantees from Assumption \ref{ass:general-dispersion-holds-independent-prior} for balls of width $w_i$ and $L_{\beta_i} \cdot w_i$. The combination of these statements gives with probability at least $1-3\delta$
	\begin{align*}
		& \sup_{\theta_i \in \Theta_i} \hat{\ell}_i(\theta_i, \beta_i(\theta_i), \beta_{-i}) = \sup_{\theta_i, \theta_i^\prime \in \Theta_i} \hat{u}_{i, \text{M}}(\theta_i, \theta_i^\prime, \beta_{-i}) - \hat{u}_{i, \text{M}}(\theta_i, \beta_i(\theta_i), \beta_{-i}) \\
		&\leq \sup_{\theta_i, \hat{\theta}_i \in \mathcal{G}_{w_i}} \frac{1}{N} \sum_{j=1}^N u_{i, \text{M}}(\theta_i, \hat{\theta}_i, \beta_{-i}(\theta_{-i}^{(j)})) - u_{i, \text{M}}(\theta_i, \beta_i(\theta_i), \beta_{-i}(\theta_{-i}^{(j)}))\\
		&+ 4\sqrt{\frac{2d_i}{N} \log\left( \frac{e N}{d_i}\right)} + 2\sqrt{\frac{2}{N} \log\left(\frac{2n}{\delta} \right)} + 3\hat{\varepsilon}_{i, \text{disp}}(w_i) + \hat{\varepsilon}_{i, \text{disp}}(L_{\beta_i}w_i).
	\end{align*}
	
\end{proof}

\section{Proofs for interdependent prior distributions Section~\ref{sec:verifying-equilibria-with-interdependent-priors}} \label{sec:appendix-interdependent-prior-proofs}

This section provides the detailed proofs for the error bounds in approximating the \textit{ex ante} utility loss for interdependent prior distributions.

\subsection{Proof of Theorem \ref{thm:constant-best-response-difference-bounded-over-partition}} \label{sec:appendix-interdependent-proof-tv-distance-over-segments}

The partition $\mathcal{B}_i$ determines which segments of the observation space $\mathcal{O}_i$ can be considered collectively. For each element $B$ within $\mathcal{B}_i$, we identify a constant best-response. We show that the error made by this procedure can be bounded in terms of the total variation distance between prior distributions conditioned on observations from $B$.
We show this for a single segment $B \in \mathcal{B}_i$ first.

% %%%%%%%%%%%%%%%%%%%%%%%%%%%%%%%%%%%%%%%%%%%%%%%% %
% If a restricted of the observation space gives additional information about
% the valuations, then one can improve the bound below.
% \norm{\restr{u_{i, \text{M}}}{\Theta_i \times B \times \mathcal{O}_{-i}}}_\infty
% While this does not give us additional benefit in general, it does for the independent prior case.

\begin{lemma} \label{thm:constant-best-response-difference-bounded-on-single-set}
	Let $B \subset \mathcal{O}_i$ and $\beta_{-i} \in \Sigma_{-i}$ be an opponent strategy profile for agent $i$.
	Then one can bound the largest difference of the \textit{ex interim} best-response utility and the utility of a constant best-response over $B$ by
	\begin{align*}
		&\sup_{o_i \in B} \sup_{b_i \in \mathcal{A}_i} \mathbb{E}_{o_{-i}, \theta_{i}|o_i} \left[u_{i, \text{M}}(\theta_i, b_i, \beta_{-i}(o_{-i})) \right] \\
		&- \sup_{b_i \in \mathcal{A}_i} \mathbb{E}_{\giventhat{\tilde{o}_i, o_{-i}, \theta_i}{ \{o_i \in B\}}} \left[u_{i, \text{M}}\left(\theta_i, b_i, \beta_{-i}(o_{-i}) \right) \right] \leq 2 \norm{u_{i, \text{M}}}_\infty \cdot \sup_{\hat{o}_i, \hat{o}_i^\prime \in B} d_{\text{TV}}\left(F_{\giventhat{\theta_i, o_{-i}}{\hat{o}_i}}, F_{\giventhat{\theta_i, o_{-i}}{\hat{o}_i^\prime}} \right).
	\end{align*}
	Furthermore, the difference between a best-response over the set $\restr{\Sigma_{i}}{B}$ of bidding functions restricted to $B$ and a constant best-response is bounded by
	\begin{align*}
		&\sup_{\beta_i^\prime \in \restr{\Sigma_{i}}{B}} \mathbb{E}_{\giventhat{o_i, o_{-i}, \theta_i}{ \{o_i \in B\}}} \left[u_{i, \text{M}}\left(\theta_i, \beta_i^\prime(o_i), \beta_{-i}(o_{-i}) \right) \right] \\
		&- \sup_{b_i \in \mathcal{A}_i} \mathbb{E}_{\giventhat{o_i, o_{-i}, \theta_i}{ \{o_i \in B\}}} \left[u_{i, \text{M}}\left(\theta_i, b_i, \beta_{-i}(o_{-i}) \right) \right] \leq 2 \norm{u_{i, \text{M}}}_\infty \cdot \sup_{\hat{o}_i, \hat{o}_i^\prime \in B} d_{\text{TV}}\left(F_{\giventhat{\theta_i, o_{-i}}{\hat{o}_i}}, F_{\giventhat{\theta_i, o_{-i}}{\hat{o}_i^\prime}} \right).
	\end{align*}
\end{lemma}

\begin{proof}
	Let $\epsilon > 0$ and $o_i \in B$. Choose $b_i^* \in \mathcal{A}_i$ such that it is within $\epsilon$ of the best-response utility, that is, 
	\begin{align} \label{equ:chosen-eps-constant-interim-best-response}
		\sup_{b_i \in \mathcal{A}_i} \mathbb{E}_{o_{-i}, \theta_i | o_i} \left[u_{i, \text{M}}(\theta, b_i, \beta_{-i}(o_{-i})) \right] - \epsilon \leq \mathbb{E}_{o_{-i}, \theta_i | o_i} \left[u_{i, \text{M}}(\theta, b_i^*, \beta_{-i}(o_{-i})) \right].
	\end{align}	
	Then, 
	\begin{align}
		&\sup_{b_i \in \mathcal{A}_i} \mathbb{E}_{o_{-i}, \theta_{i}|o_i} \left[u_{i, \text{M}}(\theta_i, b_i, \beta_{-i}(o_{-i})) \right] - \sup_{b_i^\prime \in \mathcal{A}_i} \mathbb{E}_{\giventhat{\tilde{o}_i, o_{-i}, \theta_i}{ \{\tilde{o}_i \in B\}}} \left[u_{i, \text{M}}\left(\theta_i, b_i^\prime, \beta_{-i}(o_{-i}) \right) \right] \nonumber \\
		& \overset{\text{Equ. } \ref{equ:chosen-eps-constant-interim-best-response}}{\leq} \mathbb{E}_{o_{-i}, \theta_{i}|o_i} \left[u_{i, \text{M}}(\theta_i, b_i^*, \beta_{-i}(o_{-i})) \right] - \sup_{b_i^\prime \in \mathcal{A}_i} \mathbb{E}_{\giventhat{\tilde{o}_i, o_{-i}, \theta_i}{ \{\tilde{o}_i \in B\}}} \left[u_{i, \text{M}}\left(\theta_i, b_i^\prime, \beta_{-i}(o_{-i}) \right) \right] + \epsilon \nonumber \\
		& \leq \mathbb{E}_{o_{-i}, \theta_{i}|o_i} \left[u_{i, \text{M}}(\theta_i, b_i^*, \beta_{-i}(o_{-i})) \right] - \mathbb{E}_{\giventhat{\tilde{o}_i, o_{-i}, \theta_i}{ \{\tilde{o}_i \in B\}}} \left[u_{i, \text{M}}\left(\theta_i, b_i^*, \beta_{-i}(o_{-i}) \right) \right] + \epsilon  \nonumber \\
		& = \mathbb{E}_{\giventhat{\tilde{o}_i}{{ \{\tilde{o}_i \in B\}}}} \left[\mathbb{E}_{o_{-i}, \theta_{i}|o_i} \left[u_{i, \text{M}}(\theta_i, b_i^*, \beta_{-i}(o_{-i})) \right] - \mathbb{E}_{o_{-i}, \theta_{i}|\tilde{o}_i} \left[u_{i, \text{M}}(\theta_i, b_i^*, \beta_{-i}(o_{-i})) \right] \right] + \epsilon \nonumber \\
		&\overset{\text{Theorem }\ref{thm:bound-distance-of-integrals-by-total-variation}}{\leq} \mathbb{E}_{\giventhat{\tilde{o}_i}{{ \{\tilde{o}_i \in B\}}}} \left[2 \norm{u_{i, \text{M}}}_\infty \cdot d_{\text{TV}}\left(F_{\giventhat{\theta_i, o_{-i}}{o_i}}, F_{\giventhat{\theta_i, o_{-i}}{\tilde{o}_i}} \right)  \right] + \epsilon \nonumber \\
		& \leq  2 \norm{u_{i, \text{M}}}_\infty \mathbb{E}_{\giventhat{\tilde{o}_i}{{ \{\tilde{o}_i \in B\}}}} \left[ \sup_{\hat{o}_i, \hat{o}_i^\prime \in B} d_{\text{TV}}\left(F_{\giventhat{\theta_i, o_{-i}}{\hat{o}_i}}, F_{\giventhat{\theta_i, o_{-i}}{\hat{o}_i^\prime}} \right)   \right] + \epsilon \nonumber \\
		&= 2 \norm{u_{i, \text{M}}}_\infty \cdot \sup_{\hat{o}_i, \hat{o}_i^\prime \in B} d_{\text{TV}}\left(F_{\giventhat{\theta_i, o_{-i}}{\hat{o}_i}}, F_{\giventhat{\theta_i, o_{-i}}{\hat{o}_i^\prime}} \right) + \epsilon. \label{equ:proof-interim-constant-br-tv-distance-interim-case}
	\end{align}
	As $\epsilon$ and $o_i$ were chosen arbitrarily, the first statement follows.
	
	For the second statement, observe that the best-response \textit{ex ante} utility over $B$ is bounded by the largest \textit{ex interim} best-response utility over $B$. More specifically,
	\begin{align}
		&\sup_{\beta_i^\prime \in \restr{\Sigma_{i}}{B}} \mathbb{E}_{\giventhat{o_i, o_{-i}, \theta_i}{ \{o_i \in B\}}} \left[u_{i, \text{M}}\left(\theta_i, \beta_i^\prime(o_i), \beta_{-i}(o_{-i}) \right) \right] \nonumber \\
		& = \sup_{\beta_i^\prime \in \restr{\Sigma_{i}}{B}} \mathbb{E}_{\giventhat{o_i}{ \{o_i \in B\}}} \left[ \mathbb{E}_{o_{-i}, \theta_i | o_i} \left[u_{i, \text{M}}\left(\theta_i, \beta_i^\prime(o_i), \beta_{-i}(o_{-i}) \right) \right] \right] \nonumber \\
		& \leq  \mathbb{E}_{\giventhat{o_i}{ \{o_i \in B\}}}  \left[ \sup_{\beta_i^\prime \in \restr{\Sigma_{i}}{B}} \mathbb{E}_{o_{-i}, \theta_i | o_i} \left[u_{i, \text{M}}\left(\theta_i, \beta_i^\prime(o_i), \beta_{-i}(o_{-i}) \right) \right] \right] \nonumber \\
		&= \mathbb{E}_{\giventhat{o_i}{ \{o_i \in B\}}}  \left[ \sup_{b_i \in \mathcal{A}_i} \mathbb{E}_{o_{-i}, \theta_i | o_i} \left[u_{i, \text{M}}(\theta_i, b_i, \beta_{-i}(o_{-i})) \right] \right] \nonumber \\
		&\leq \sup_{o_i \in B} \sup_{b_i \in \mathcal{A}_i} \mathbb{E}_{o_{-i}, \theta_i | o_i} \left[u_{i, \text{M}}(\theta_i, b_i, \beta_{-i}(o_{-i})) \right]. \label{equ:proof-interim-constant-br-tv-distance-ex-ante-bounded-by-interim}
	\end{align}
	Therefore, using the first statement, we get
	\begin{align*}
		&\sup_{\beta_i^\prime \in \restr{\Sigma_{i}}{B}} \mathbb{E}_{\giventhat{o_i, o_{-i}, \theta_i}{ \{o_i \in B\}}} \left[u_{i, \text{M}}\left(\theta_i, \beta_i^\prime(o_i), \beta_{-i}(o_{-i}) \right) \right] - \sup_{b_i^\prime \in \mathcal{A}_i} \mathbb{E}_{\giventhat{\tilde{o}_i, o_{-i}, \theta_i}{ \{\tilde{o}_i \in B\}}} \left[u_{i, \text{M}}\left(\theta_i, b_i^\prime, \beta_{-i}(o_{-i}) \right) \right] \\
		& \overset{\text{Equ. } \ref{equ:proof-interim-constant-br-tv-distance-ex-ante-bounded-by-interim}}{\leq} \sup_{o_i \in B} \sup_{b_i \in \mathcal{A}_i} \mathbb{E}_{o_{-i}, \theta_i | o_i} \left[u_{i, \text{M}}(\theta_i, b_i, \beta_{-i}(o_{-i})) \right] - \sup_{b_i^\prime \in \mathcal{A}_i} \mathbb{E}_{\giventhat{\tilde{o}_i, o_{-i}, \theta_i}{ \{\tilde{o}_i \in B\}}} \left[u_{i, \text{M}}\left(\theta_i, b_i^\prime, \beta_{-i}(o_{-i}) \right) \right] \\
		& \overset{\text{Equ. } \ref{equ:proof-interim-constant-br-tv-distance-interim-case}}{\leq} 2 \norm{u_{i, \text{M}}}_\infty \cdot \sup_{\hat{o}_i, \hat{o}_i^\prime \in B} d_{\text{TV}}\left(F_{\giventhat{\theta_i, o_{-i}}{\hat{o}_i}}, F_{\giventhat{\theta_i, o_{-i}}{\hat{o}_i^\prime}} \right).
	\end{align*} 
\end{proof}

The previous lemma indicates that the error incurred by employing a constant best-response, as opposed to a functional one over a set $B$, can be managed provided that the conditional distribution does not change too much. This simplifies the utility loss estimation process considerably, as the error introduced by constant best-responses can be bounded by the maximum total variation distance of the conditional distributions for observations from $B$.
The following theorem expands upon this result, applying it across the entire partition $\mathcal{B}_i$ of $\mathcal{O}_i$.

\begin{customthm}{6.3}
	Let $\mathcal{B}_i = \left\{B_1, \dots, B_{N_{\mathcal{B}_i}} \right\}$ be a partition of $\mathcal{O}_i$.
	The difference between a best-response utility over function space to best-responses that are constant for every $B_k$ satisfies
	\begin{align*}
		&\sup_{\beta_i^\prime \in \Sigma_{i}} \tilde{u}_{i, \text{M}}(\beta_i^\prime, \beta_{-i}) - \sup_{b \in \mathcal{A}_i^{N_{\mathcal{B}_i}}} \tilde{u}_{i, \text{M}} \left(\sum_{k=1}^{N_{\mathcal{B}_i}} b_k \mathds{1}_{B_k}, \beta_{-i} \right) \leq 2 \sum_{k=1}^{N_{\mathcal{B}_i}} P(o_i \in B_k) \tau_{i, B_k},
	\end{align*}
	with $\tau_{i, B_k} := \sup_{\hat{o}_i, \hat{o}_i^\prime \in B_k} d_{\text{TV}}\left(F_{\giventhat{\theta_i, o_{-i}}{\hat{o}_i}}, F_{\giventhat{\theta_i, o_{-i}}{\hat{o}_i^\prime}} \right)$.
	If there exists a constant $L_{B_k} > 0$ such that $d_{\text{TV}}\left(F_{\giventhat{\theta_i, o_{-i}}{\hat{o}_i}}, F_{\giventhat{\theta_i, o_{-i}}{\hat{o}_i^\prime}} \right) \leq L_{B_k} \norm{o_i - o_i^\prime}$ for $o_i, o_i^\prime \in B_k$, then $\tau_{i, B_k} \leq L_{B_k} \text{diam}(B_k)$, where $\text{diam}(B_k)$ denotes $B_k$'s diameter.
\end{customthm}

\begin{proof}
	Let $o_i \in \mathcal{O}_i$. Then, there exists a unique $B_k \in \mathcal{B}_i$ such that $o_i \in B_k$. The error between the \textit{ex interim} best-response utility and the constant best-response utility over $B_k$ can be bounded by
	\begin{align*}
		&\sup_{b_i \in \mathcal{A}_i} \mathbb{E}_{o_{-i}, \theta_{i}|o_i} \left[u_{i, \text{M}}(\theta_i, b_i, \beta_{-i}(o_{-i})) \right] - \sup_{b_i \in \mathcal{A}_i} \mathbb{E}_{\giventhat{\tilde{o}_i, o_{-i}, \theta_i}{\left\{\tilde{o}_i \in B_k \right\}}} \left[u_{i, \text{M}}\left(\theta_i, b_i, \beta_{-i}(o_{-i}) \right) \right] \\
		& \overset{\text{Lemma } \ref{thm:constant-best-response-difference-bounded-on-single-set}}{\leq} 2 \norm{u_{i, \text{M}}}_\infty \cdot \sup_{\hat{o}_i, \hat{o}_i^\prime \in B_k} d_{\text{TV}}\left(F_{\giventhat{\theta_i, o_{-i}}{\hat{o}_i}}, F_{\giventhat{\theta_i, o_{-i}}{\hat{o}_i^\prime}} \right) = 2 \norm{u_{i, \text{M}}}_\infty \tau_{i, B_k}.
	\end{align*}
	We rewrite the best-response \textit{ex ante} utilities using the law of total expectation. For the first term follows
	\begin{align*}
		\sup_{\beta_i^\prime \in \Sigma_{i}} \tilde{u}_{i, \text{M}}(\beta_i^\prime, \beta_{-i}) &= \sup_{\beta_i^\prime \in \Sigma_{i}} \mathbb{E}_{o_i, o_{-i}, \theta_i} \left[u_{i, \text{M}}\left(\theta_i, \beta_i^\prime(o_i), \beta_{-i}(o_{-i}) \right) \right] \\
		&= \sup_{\beta_i^\prime \in \Sigma_{i}} \sum_{k=1}^{N_{\mathcal{B}_i}} P \left(o_i \in B_k \right) \mathbb{E}_{o_i, o_{-i}, \theta_i | \left\{o_i \in B_k \right\}} \left[u_{i, \text{M}}\left(\theta_i, \beta_i^\prime(o_i), \beta_{-i}(o_{-i}) \right) \right] \\
		& = \sum_{k=1}^{N_{\mathcal{B}_i}} P \left(o_i \in B_k \right) \sup_{\beta_i^\prime \in \restr{\Sigma_{i}}{B_k}} \mathbb{E}_{o_i, o_{-i}, \theta_i | \left\{o_i \in B_k \right\}} \left[u_{i, \text{M}}\left(\theta_i, \beta_i^\prime(o_i), \beta_{-i}(o_{-i}) \right) \right],
	\end{align*}
	where $\restr{\Sigma_{i}}{B_k}$ denotes the restriction of the bidding strategies to $B_k$.
	We have for the second term
	\begin{align*}
		&\sup_{b \in \mathcal{A}_i^{N_{\mathcal{B}_i}}} \mathbb{E}_{o_i, o_{-i}, \theta_i} \left[u_{i, \text{M}}\left(\theta_i, \sum_{k=1}^{N_{\mathcal{B}_i}} b_k \mathds{1}_{B_k}(o_i), \beta_{-i}(o_{-i}) \right) \right] \\
		& = \sup_{b \in \mathcal{A}_i^{N_{\mathcal{B}_i}}} \sum_{k=1}^{N_{\mathcal{B}_i}} P \left(o_i \in B_k \right) \mathbb{E}_{o_i, o_{-i}, \theta_i | \left\{o_i \in B_k \right\}} \left[u_{i, \text{M}}\left(\theta_i, b_k, \beta_{-i}(o_{-i}) \right) \right] \\
		& = \sum_{k=1}^{N_{\mathcal{B}_i}} P \left(o_i \in B_k \right) \sup_{b_k \in \mathcal{A}_i} \mathbb{E}_{o_i, o_{-i}, \theta_i | \left\{o_i \in B_k \right\}} \left[u_{i, \text{M}}\left(\theta_i, b_k, \beta_{-i}(o_{-i}) \right) \right].
	\end{align*}
	Combing these two transformations gives
	\begin{align*}
		&\sup_{\beta_i^\prime \in \Sigma_{i}} \mathbb{E}_{o_i, o_{-i}, \theta_i} \left[u_{i, \text{M}}\left(\theta_i, \beta_i^\prime(o_i), \beta_{-i}(o_{-i}) \right) \right] \\
		&- \sup_{b \in \mathcal{A}_i^{N_{\mathcal{B}_i}}} \mathbb{E}_{o_i, o_{-i}, \theta_i} \left[u_{i, \text{M}}\left(\theta_i, \sum_{k=1}^{N_{\mathcal{B}_i}} b_k \mathds{1}_{B_k}(o_i), \beta_{-i}(o_{-i}) \right) \right] \\
		& \leq \sum_{k=1}^{N_{\mathcal{B}_i}} P \left(o_i \in B_k \right) \left(\sup_{\beta_i^\prime \in \restr{\Sigma_{i}}{B_k}} \mathbb{E}_{o_i, o_{-i}, \theta_i | \left\{o_i \in B_k \right\}} \left[u_{i, \text{M}}\left(\theta_i, \beta_i^\prime(o_i), \beta_{-i}(o_{-i}) \right) \right] \right. \\
		& - \left. \sup_{b_k \in \mathcal{A}_i} \mathbb{E}_{o_i, o_{-i}, \theta_i | \left\{o_i \in B_k \right\}} \left[u_{i, \text{M}}\left(\theta_i, b_k, \beta_{-i}(o_{-i}) \right) \right] \right)\\
		& \overset{\text{Lemma } \ref{thm:constant-best-response-difference-bounded-on-single-set}}{\leq} 2 \norm{u_{i, \text{M}}}_\infty \sum_{k=1}^{N_{\mathcal{B}_i}} P \left(o_i \in B_k \right) \tau_{i, B_k}.
	\end{align*}
	
	For arbitrary $o_i, o_i^\prime \in B_k$, we have $\norm{o_i - o_i^\prime} \leq \text{diam}(B_k)$. Therefore, if there exists a constant $L_{B_k} > 0$ such that $d_{\text{TV}}\left(g_{B_k}(o_i), g_{B_k}(o_i^\prime) \right) \leq L_{B_k} \norm{o_i - o_i^\prime}$ for $o_i, o_i^\prime \in B_k$, then $\tau_{i, B_k} \leq L_{B_k} \text{diam}(B_k)$.
\end{proof}

\subsection{Proof of Theorems~\ref{thm:interdependent-prior-case-estimate-ex-ante-utility-hoeffding} and \ref{thm:finite-sample-approximation-of-constant-best-response-utility-pseudo-dimension}} \label{sec:appendix-interdependent-proof-sampling-based-estimate}

\begin{customthm}{6.4}
	Let $\beta \in \Sigma$ be a strategy profile. With probability $1 - \delta$ over the draw of the dataset $\mathcal{D}^{\beta}$, we have for every agent $i \in [n]$
	\begin{align*}
		\abs{\tilde{u}_{i, \text{M}}(\beta_i, \beta_{-i}) - \frac{1}{N} \sum_{j=1}^N u_{i, \text{M}}\left(\theta_{i}^{(j)}, \beta_{i}(o_i^{(j)}), \beta_{-i}(o_{-i}^{(j)})\right) } \leq \sqrt{\frac{2}{N}\log\left(\frac{2 n}{\delta}\right)}.
	\end{align*}
\end{customthm}

\begin{proof}
	Fix an agent $i \in [n]$. $u_{i, \text{M}}(\theta_i, \beta_i(o_i), \beta_{-i}(o_{-i}))$ with $\left(\theta_i, \beta_i(o_i), \beta_{-i}(o_{-i}) \right) \sim F^\beta$ is a random variable with a distribution over $[-1, 1]$.
	
	The values $u_{i, \text{M}}\left(\theta_{i}^{(1)}, \beta_{i}(o_i^{(1)}), \beta_{-i}(o_{-i}^{(1)})\right), \dots, u_{i, \text{M}}\left(\theta_{i}^{(N)}, \beta_{i}(o_i^{(N)}), \beta_{-i}(o_{-i}^{(N)})\right)$ are i.i.d. samples from this distribution with $\left(\theta_{i}^{(j)}, \beta_{i}(o_i^{(j)}), \beta_{-i}(o_{-i}^{(j)})\right)$ coming from the dataset $\mathcal{D}^{\beta}$ for $1 \leq j \leq N$.
	By applying Hoeffding's inequality (Theorem~\ref{thm:hoeffding-inequality}), we get with probability at least $1 - \frac{\delta}{n}$
	\begin{align*}
		\abs{\tilde{u}_{i, \text{M}}(\beta_i, \beta_{-i}) - \frac{1}{N} \sum_{j=1}^N u_{i, \text{M}}\left(\theta_{i}^{(j)}, \beta_{i}(o_i^{(j)}), \beta_{-i}(o_{-i}^{(j)})\right) } \leq \sqrt{\frac{2}{N}\log\left(\frac{2 n}{\delta}\right)}.
	\end{align*}
	The statement follows by applying a union bound over the set of agents $[n]$. 
\end{proof}

\begin{customthm}{6.5}
	With probability $1 - \delta$ over the draw of the $n$ sets $\mathcal{D}^{\beta}(\mathcal{B}_1), \dots, \mathcal{D}^{\beta}(\mathcal{B}_n)$, for partitions $\mathcal{B}_i = \left\{B_1, \dots, B_{N_{\mathcal{B}_i}} \right\}$ of $\mathcal{O}_i$ for every agent $i \in [n]$, we have
	\begin{align*}
		&\abs{\sup_{b_i \in \mathcal{A}_i} \mathbb{E}_{\giventhat{o_i, o_{-i}, \theta_i}{\left\{o_i \in B_k \right\}}} \left[u_{i, \text{M}}\left(\theta_i, b_i, \beta_{-i}(o_{-i}) \right) \right] - \sup_{b_i \in \mathcal{A}_i} \frac{1}{N_{B_k}} \sum_{j=1}^{N_{B_k}} u_{i, \text{M}}\left(\theta_{i}^{(j)}, b_i, \beta_{-i}(o_{-i}^{(j)})\right) } \\
		& \leq  \tilde{\varepsilon}_{i, \text{Pdim}}(N_{B_k}),\\
		&\mbox{ with } \tilde{\varepsilon}_{i, \text{Pdim}}\left(N_{B_k}\right) := 2\sqrt{\frac{2 d_i}{N_{B_k}} \log\left(\frac{e N_{B_k}}{d_i} \right) } + \sqrt{\frac{2}{N_{B_k}} \log \left(\frac{n N_{\mathcal{B}_{\text{max}}}}{\delta} \right)}, \mbox{ and } d_i := \text{Pdim} \left(\tilde{\mathcal{F}}_{i, \text{M}} \right).
	\end{align*}
\end{customthm}

\begin{proof}
	Fix an agent $i \in [n]$ and a segment $B_k \in \mathcal{B}_i$. Note that we can write the \textit{ex ante} utility given the event $\left\{o_i \in B_k \right\}$ and bid $b_i \in \mathcal{A}_i$ as
	\begin{align*}
		&\mathbb{E}_{\giventhat{o_i, o_{-i}, \theta_i}{\left\{o_i \in B_k \right\}}} \left[u_{i, \text{M}}\left(\theta_i, b_i, \beta_{-i}(o_{-i}) \right) \right] \\
		&= \mathbb{E}_{o_i, \beta_{-i}(o_{-i}), \theta_i \sim \giventhat{F^{\beta_{-i}}}{\left\{o_i \in B_k \right\}}} \left[u_{i, \text{M}}\left(\theta_i, b_i, \beta_{-i}(o_{-i}) \right) \right].
	\end{align*}
	
	By Theorem \ref{thm:pollard-pac-bound-general-distribution}, we get that with probability at least $1 - \frac{\delta}{nN_{\mathcal{B}_{\text{max}}}}$ over the draw of $\mathcal{D}^{\beta_{-i}}(B_k)$ for all $b_i \in \mathcal{A}_i$,
	\begin{align*}
		&\mathbb{E}_{o_i, \beta_{-i}(o_{-i}), \theta_i \sim \giventhat{F^{\beta_{-i}}}{\left\{o_i \in B_k \right\}}} \left[u_{i, \text{M}}\left(\theta_i, b_i, \beta_{-i}(o_{-i}) \right) \right] - \frac{1}{N_{B_k}} \sum_{j=1}^{N_{B_k}} u_{i, \text{M}}\left(\theta_{i}^{(j)}, b_i, \beta_{-i}(o_{-i}^{(j)})\right) \\
		& \leq 2 \sqrt{\frac{2 d_i}{N_{B_k}} \log\left(\frac{e N_{B_k}}{d_i} \right) } + \sqrt{\frac{2}{N_{B_k}} 
			\log\left(\frac{n N_{\mathcal{B}_{\text{max}}}}{\delta} \right)}.
	\end{align*}
	\begin{comment}
		Therefore, we have with probability at least $1 - \frac{\delta}{n N_{\mathcal{B}_{\text{max}}}}$,
		\begin{align*}
			&\mathbb{E}_{\giventhat{o_i, o_{-i}, \theta_i}{\left\{o_i \in B_k \right\}}} \left[u_{i, \text{M}}\left(\theta_i, b_i, \beta_{-i}(o_{-i}) \right) \right] - \mathbb{E}_{\giventhat{o_i, o_{-i}, \theta_i}{\left\{o_i \in B_k \right\}}} \left[u_{i, \text{M}}\left(\theta_i, \beta_i(o_i), \beta_{-i}(o_{-i}) \right) \right] \\
			&=\mathbb{E}_{\giventhat{o_i, o_{-i}, \theta_i}{\left\{o_i \in B_k \right\}}} \left[u_{i, \text{M}}\left(\theta_i, b_i, \beta_{-i}(o_{-i}) \right) \right] - \frac{1}{N_{B_k}} \sum_{j=1}^{N_{B_k}} u_{i, \text{M}}\left(\theta_{i}^{(j)}, b_i, \beta_{-i}(o_{-i}^{(j)})\right) \\
			&+ \frac{1}{N_{B_k}} \sum_{j=1}^{N_{B_k}} u_{i, \text{M}}\left(\theta_{i}^{(j)}, b_i, \beta_{-i}(o_{-i}^{(j)})\right) - \frac{1}{N_{B_k}} \sum_{j=1}^{N_{B_k}} u_{i, \text{M}}\left(\theta_{i}^{(j)}, \beta_i(o_i^{(j)}), \beta_{-i}(o_{-i}^{(j)})\right) \\
			&+ \frac{1}{N_{B_k}} \sum_{j=1}^{N_{B_k}} u_{i, \text{M}}\left(\theta_{i}^{(j)}, \beta_i(o_i^{(j)}), \beta_{-i}(o_{-i}^{(j)})\right)  - \mathbb{E}_{\giventhat{o_i, o_{-i}, \theta_i}{\left\{o_i \in B_k \right\}}} \left[u_{i, \text{M}}\left(\theta_i, \beta_i(o_i), \beta_{-i}(o_{-i}) \right) \right]\\
			& \leq \frac{1}{N_{B_k}} \sum_{j=1}^{N_{B_k}} u_{i, \text{M}}\left(\theta_{i}^{(j)}, b_i, \beta_{-i}(o_{-i}^{(j)})\right) - u_{i, \text{M}}\left(\theta_{i}^{(j)}, \beta_{i}(o_i^{(j)}), \beta_{-i}(o_{-i}^{(j)})\right) + \tilde{\varepsilon}_{i, \text{Pdim}}(N_{B_k}).
		\end{align*}
	\end{comment}
	Since this holds for all $b_i \in \mathcal{A}_i$, we get the statement for the choice of $i$ and $B_k$. Taking a union bound over all $i \in [n]$ and segments $B_k \in \mathcal{B}_i$ yields the final statement.
\end{proof}

\subsection{Proof of Theorem~\ref{thm:complete-approximation-bound-interdependent-prior}} \label{sec:appendix-interdependent-proof-complete-bound}
\begin{customlemma}{6.7}
	Let $\delta > 0$, $\beta \in \tilde{\Sigma}$ be a strategy profile, and $\text{M}$ be a mechanism. Suppose that for each agent $i \in [n]$ and segment $B_k \in \mathcal{B}_i$, Assumption \ref{ass:bidding-dispersion-holds-interdependent-prior} holds for $w_i >0$ and $v_i(w_i)$. 
	Then, with probability $1- \delta$ over the draw of the sets $\left\{\giventhat{\mathcal{D}^{\beta}(\mathcal{B}_i)}{i \in [n]} \right\}$, agents $i \in [n]$, and segments $B_k \in \mathcal{B}_i$,
	\begin{align*}
		&\abs{\sup_{b_i \in \mathcal{A}_i} \frac{1}{N_{B_k}} \sum_{j=1}^{N_{B_k}} u_{i, \text{M}}\left(\theta_{i}^{(j)}, b_i, \beta_{-i}(o_{-i}^{(j)})\right) -  \max_{b_i \in \mathcal{G}_w} \frac{1}{N_{B_k}} \sum_{j=1}^{N_{B_k}} u_{i, \text{M}}\left(\theta_{i}^{(j)}, b_i, \beta_{-i}(o_{-i}^{(j)})\right)} \\
		&\leq \frac{N_{B_k} - v_{i, B_k}\left(w_{i} \right)}{N_{B_k}} L_{i} w_{i} + \frac{2v_{i, B_k}\left(w_{i} \right)}{N_{B_k}} =: \tilde{\varepsilon}_{i, \text{disp}}(N_{B_k}).
	\end{align*}
\end{customlemma}

\begin{proof}
	Fix agent $i \in [n]$ and $B_k \in \mathcal{B}_i$.
	Let $b_i, b_i^\prime \in \mathcal{A}_i = [0, 1]^m$ with $\norm{b_i - b_i^\prime}_1 \leq  w_{i}$. When considering the following difference for a specific $j \in [N_{B_k}]$
	\begin{align} \label{equ:single-sample-difference-either-lipschitz-or-discontinuous}
		u_{i, \text{M}}\left(\theta_{i}^{(j)}, b_i, \beta_{-i}(o_{-i}^{(j)})\right) - u_{i, \text{M}}\left(\theta_{i}^{(j)}, b_i^\prime, \beta_{-i}(o_{-i}^{(j)})\right),
	\end{align}
	then, either $u_{i, \text{M}}\left(\theta_{i}^{(j)}, \cdot, \beta_{-i}(o_{-i}^{(j)})\right)$ is $L_i$-Lipschitz continuous over $[b_i, b_i^\prime]$ or there is a jump discontinuity. In the first case, we can bound the difference in Equation \ref{equ:single-sample-difference-either-lipschitz-or-discontinuous} by $L_i\norm{b_i - b_i^\prime}_1$, and in the second, we can bound it by $2\norm{u_{i, \text{M}}}_\infty$. While the second bound is trivial, dispersion guarantees that with high probability this case can happen at most $\frac{v_{i, B_k}\left(w_{i} \right)}{N_{B_k}}$ times. Therefore, by the definition of dispersion, we know that with probability $1-\delta$ over the draw of the sets $\left\{\giventhat{\mathcal{D}^{\beta}(B_k)}{B_k \in \mathcal{B}_i, i \in [n]} \right\}$, for mechanism $\text{M}$, agents $i \in [n]$, and segments $B_k \in \mathcal{B}_i$, we have for all $b_i, b_i^\prime \in \mathcal{A}_i = [0, 1]^m$ with $\norm{b_i - b_i^\prime}_1 \leq  w_{i}$ that
	\begin{align*}
		&\abs{\frac{1}{N_{B_k}} \sum_{j=1}^{N_{B_k}} u_{i, \text{M}}\left(\theta_{i}^{(j)}, b_i, \beta_{-i}(o_{-i}^{(j)})\right) - u_{i, \text{M}}\left(\theta_{i}^{(j)}, b_i^\prime, \beta_{-i}(o_{-i}^{(j)})\right)} \\
		&\leq \frac{N_{B_k} - v_{i, B_k}\left(w_{i} \right)}{N_{B_k}} L_{i} w_{i} + \frac{v_{i, B_k}\left(w_{i} \right)}{N_{B_k}}2\norm{u_{i, \text{M}}}_\infty.
	\end{align*}
	Let $b_i \in \mathcal{A}_i$ be arbitrary. By the definition of $\mathcal{G}_w$, there must be a point $p \in \mathcal{G}_w$ such that $\norm{b_i - p}_1 \leq  w_{i}$. Therefore, with probability $1 - \delta$
	\begin{align*}
		&\abs{\frac{1}{N_{B_k}} \sum_{j=1}^{N_{B_k}} u_{i, \text{M}}\left(\theta_{i}^{(j)}, b_i, \beta_{-i}(o_{-i}^{(j)})\right) - \frac{1}{N_{B_k}} \sum_{j=1}^{N_{B_k}} u_{i, \text{M}}\left(\theta_{i}^{(j)}, p, \beta_{-i}(o_{-i}^{(j)})\right)}\\
		&\leq \frac{N_{B_k} - v_{i, B_k}\left(w_{i} \right)}{N_{B_k}} L_{i} w_{i} + \frac{v_{i, B_k}\left(w_{i} \right)}{N_{B_k}}2\norm{u_{i, \text{M}}}_\infty.
	\end{align*}
\end{proof}

\begin{lemma}\label{thm:aux-lemma-approximate-convex-combination-of-bounded-variables}
	Let for every agent $i \in [n]$, $\mathcal{B}_i = \left\{B_1, \dots, B_{N_{\mathcal{B}_i}} \right\}$ be a partition of $\mathcal{O}_i$, $a_1, \dots, a_{N_{\mathcal{B}_i}} \in [0, 2]$, and $\delta \in (0, 1)$. Then, with probability at least $1 - \delta$ over the draw of the dataset $\mathcal{D}$, we have for all agents $i \in [n]$
	\begin{align*}
		\abs{\sum_{k=1}^{N_{\mathcal{B}_i} } \left(P(o_i \in B_k) - \frac{N_{B_k}}{N} \right) a_k} \leq \sqrt{\frac{2}{N} \log\left( \frac{2n}{\delta} \right)}.
	\end{align*}
\end{lemma}

\begin{proof}
	Fix an agent $i \in [n]$. Define the random variable $Y_i:= \sum_{k=1}^{N_{\mathcal{B}_i}} \mathds{1}_{B_k}(o_i) (a_k - 1)$, where $o_i \sim F_{o_i}$. As $\mathcal{B}_i$ is a partition and $a_k \in [0, 2]$, we know $Y_i \in [0, 2]$. We have with probability $1 - \frac{\delta}{n}$
	\begin{align*}
		&\abs{\sum_{k=1}^{N_{\mathcal{B}_i} } \left(P(o_i \in B_k) - \frac{N_{B_k}}{N} \right) a_k} 
		= \abs{\sum_{k=1}^{N_{\mathcal{B}_i} } \mathbb{E}_{o_i}\left[\mathds{1}_{B_k}(o_i) \right] a_k - \sum_{k=1}^{N_{\mathcal{B}_i} }\frac{N_{B_k}}{N} a_k} \\
		& = \abs{\mathbb{E}_{o_i} \left[\sum_{k=1}^{N_{\mathcal{B}_i} } \mathds{1}_{B_k}(o_i) a_k \right] - \frac{1}{N} \sum_{j=1}^N \sum_{k=1}^{N_{\mathcal{B}_i} } \mathds{1}_{B_k}(o_i^{(j)}) a_k  }\\
		&\overset{\text{Theorem }\ref{thm:hoeffding-inequality}}{\leq} \sqrt{\frac{2}{N} \log\left( \frac{2n}{\delta} \right)},
	\end{align*}
	where we used the Hoeffding inequality on the for i.i.d. draws of the random variable $Y_i$. A union bound over the agents $[n]$ completes the proof.
\end{proof}

\begin{customthm}{6.8}
	Let $\delta > 0$ and $\beta \in \tilde{\Sigma}$ be a strategy profile. Suppose that for each agent $i \in [n]$ and segment $B_k \in \mathcal{B}_i$, Assumption \ref{ass:bidding-dispersion-holds-interdependent-prior} holds.
	Then, with probability $1- 4\delta$ over the draw of the sets $\left\{\giventhat{\mathcal{D}^{\beta}(\mathcal{B}_i)}{i \in [n]} \right\}$, agents $i \in [n]$, and segments $B_k \in \mathcal{B}_i$,
	\begin{align*}
		&\tilde{\ell}_i(\beta_i, \beta_{-i}) = \sup_{\beta^{\prime}_i \in \Sigma_i} \tilde{u}_{i, \text{M}}(\beta^{\prime}_i, \beta_{-i}) - \tilde{u}_{i, \text{M}}(\beta_i, \beta_{-i})\\
		&\leq \sum_{k=1}^{N_{\mathcal{B}_i}} \frac{N_{B_k}}{N} \max_{b_i \in\mathcal{G}_{w_i}} \frac{1}{N_{B_k}} \sum_{j=1}^{N_{B_k}} u_{i, \text{M}} (\theta_i^{(j)}, b_i, \beta_{-i}(o_{-i}^{(j)})) - \frac{1}{N} \sum_{l=1}^{N} u_{i, \text{M}} (\theta_i^{(l)}, \beta_i(o_i^{(l)}), \beta_{-i}(o_{-i}^{(l)})) \\
		&+ 2 \sqrt{\frac{2}{N} \log\left(\frac{2n}{\delta} \right)} + \sum_{k=1}^{N_{\mathcal{B}_i}} \frac{N_{B_k}}{N} \min\left\{1, \left(\tau_{i, B_k} + \tilde{\varepsilon}_{i, \text{Pdim}}(N_{B_k}) + \tilde{\varepsilon}_{i, \text{disp}}(N_{B_k}) \right)\right\},
	\end{align*}
	where $\tau_{i, B_k}$, $\tilde{\varepsilon}_{i, \text{Pdim}}(N_{B_k})$, and $\tilde{\varepsilon}_{i, \text{disp}}(N_{B_k})$ are the constants defined in Theorems~\ref{thm:constant-best-response-difference-bounded-over-partition}, \ref{thm:finite-sample-approximation-of-constant-best-response-utility-pseudo-dimension}, and Lemma~\ref{thm:dispersion-bound-single-segment-approximation-uniform-grid}.
\end{customthm}

\begin{proof}
	Fix $i \in [n]$. The \textit{ex ante} utility loss consists of the best-response utility and the \textit{ex ante} utility of the strategy profile $\beta$. We start by approximating the \textit{ex ante} utility of $\beta$. By Theorem~\ref{thm:interdependent-prior-case-estimate-ex-ante-utility-hoeffding}, with probability $1 - \delta$ over the draw of the dataset $\mathcal{D}^{\beta} = \left\{ \left(\theta^{\left(l\right)}, o^{\left(l\right)}, \beta(o^{\left(l\right)}) \right): 1\leq l \leq N \right\}$
	\begin{align} \label{equ:proof-main-theorem-interdependent-prior-ex-ante-utility}
		\abs{\tilde{u}_{i, \text{M}}(\beta_i, \beta_{-i}) - \frac{1}{N} \sum_{l=1}^{N} u_{i, \text{M}} (\theta_i^{(l)}, \beta_i(o_i^{(l)}), \beta_{-i}(o_{-i}^{(l)}))} \leq \sqrt{ \frac{2}{N} \log \left(\frac{2n}{\delta} \right) }.
	\end{align}
	Let's consider the estimation error to the best-response utility next. We can rewrite the best-response \textit{ex ante} utility to
	\begin{align*}
		&\sup_{\beta_i^\prime \in \Sigma_{i}} \mathbb{E}_{o_i, o_{-i}, \theta_i} \left[u_{i, \text{M}}\left(\theta_i, \beta_i^\prime(o_i), \beta_{-i}(o_{-i}) \right) \right] \\
		& = \sum_{k=1}^{N_{\mathcal{B}_i}} P \left(o_i \in B_k \right) \sup_{\beta_i^\prime \in \restr{\Sigma_{i}}{B_k}} \mathbb{E}_{o_i, o_{-i}, \theta_i | \left\{o_i \in B_k \right\}} \left[u_{i, \text{M}}\left(\theta_i, \beta_i^\prime(o_i), \beta_{-i}(o_{-i}) \right) \right].
	\end{align*}
	The inner terms, i.e., the difference of the best-response \textit{ex ante} utility to our estimator over each $B_k \in \mathcal{B}_i$, can be bounded by
	\begin{align} \label{equ:proof-main-theorem-interdependent-prior-br-utility-to-estimator-difference-single-segment-bounded-by-one}
		\sup_{\beta_i^\prime \in \restr{\Sigma_{i}}{B_k}} \mathbb{E}_{o_i, o_{-i}, \theta_i | \left\{o_i \in B_k \right\}} \left[u_{i, \text{M}}\left(\theta_i, \beta_i^\prime(o_i), \beta_{-i}(o_{-i}) \right) \right] - \max_{b_i \in\mathcal{G}_{w_i}} \frac{1}{N_{B_k}} \sum_{j=1}^{N_{B_k}} u_{i, \text{M}} (\theta_i^{(j)}, b_i, \beta_{-i}(o_{-i}^{(j)})) \leq 1.
	\end{align}
	This can be seen by noting that the estimator on the right is bounded below by zero because an agent can guarantee not to be worse off than not participating by bidding the minimal amount. Therefore, we can bound the estimation error for each $B_k$ by one in the worst-case.
	A more meaningful upper bound to the estimation error for the best-response utility can be given by considering four approximation steps. The first one is to consider constant best-responses over a part of the observation space. By Theorem~\ref{thm:constant-best-response-difference-bounded-over-partition}, we have
	\begin{align} \label{equ:proof-main-theorem-interdependent-prior-constant-best-response-error}
		\abs{ \sup_{\beta_i^\prime \in \Sigma_i} \tilde{u}_{i, \text{M}}(\beta_i^\prime, \beta_{-i}) - \sup_{b \in \mathcal{A}_i^{N_{\mathcal{B}_i}}} \tilde{u}_{i, \text{M}}\left(\sum_{k=1}^{N_{\mathcal{B}_i}} b_k \mathds{1}_{B_k}, \beta_{-i} \right)} \leq \sum_{k=1}^{N_{\mathcal{B}_i}} P(o_i \in B_k) \tau_{i, B_k}.
	\end{align}
	The second step is to maximize the empirical mean instead of the expectation. By Theorem~\ref{thm:finite-sample-approximation-of-constant-best-response-utility-pseudo-dimension}, we have with probability $1 - \delta$ over the draw of the datasets $\left\{\giventhat{\mathcal{D}^{\beta}(B_k)}{B_k \in \mathcal{B}_i, i \in [n]} \right\}$, for all agents $i \in [n]$ and segments $B_k \in \mathcal{B}_i$,
	\begin{align} \label{equ:proof-main-theorem-interdependent-prior-pseudo-dimension-bound}
		\begin{split}
			&\abs{ \sup_{b \in \mathcal{A}_i^{N_{\mathcal{B}_i}}} \tilde{u}_{i, \text{M}}\left(\sum_{k=1}^{N_{\mathcal{B}_i}} b_k \mathds{1}_{B_k}, \beta_{-i} \right) - \sum_{k=1}^{N_{\mathcal{B}_i}} P(o_i \in B_k) \sup_{b_i \in\mathcal{A}_i} \frac{1}{N_{B_k}} \sum_{j=1}^{N_{B_k}} u_{i, \text{M}} (\theta_i^{(j)}, b_i, \beta_{-i}(o_{-i}^{(j)})) } \\
			&\leq \sum_{k=1}^{N_{\mathcal{B}_i}} P(o_i \in B_k) \tilde{\varepsilon}_{i, \text{Pdim}}(N_{B_k}).
		\end{split}
	\end{align}
	The third step to enable a search for a best-response is to consider a finite grid $\mathcal{G}_{w_i}$ over $\mathcal{A}_i$, leveraging the concept of dispersion for guarantees. By Lemma~\ref{thm:dispersion-bound-single-segment-approximation-uniform-grid}, we have with probability $1 - \delta$ over the draw of the datasets $\left\{\giventhat{\mathcal{D}^{\beta}(B_k)}{B_k \in \mathcal{B}_i, i \in [n]} \right\}$, for all agents $i \in [n]$ and segments $B_k \in \mathcal{B}_i$,
	\begin{align} \label{equ:proof-main-theorem-interdependent-prior-dispersion-bound}
		\begin{split}
			&\abs{\sum_{k=1}^{N_{\mathcal{B}_i}} P(o_i \in B_k) \left( \sup_{b_i \in\mathcal{A}_i} \frac{1}{N_{B_k}} \sum_{j=1}^{N_{B_k}} u_{i, \text{M}} (\theta_i^{(j)}, b_i, \beta_{-i}(o_{-i}^{(j)})) - \sup_{b_i \in \mathcal{G}_{w_i}} \frac{1}{N_{B_k}} \sum_{j=1}^{N_{B_k}} u_{i, \text{M}} (\theta_i^{(j)}, b_i, \beta_{-i}(o_{-i}^{(j)})) \right) } \\
			& \leq \sum_{k=1}^{N_{\mathcal{B}_i}} P(o_i \in B_k) \tilde{\varepsilon}_{i, \text{disp}}(N_{B_k}).
		\end{split}
	\end{align}
	The fourth approximation step bound the error made by estimation the marginal probabilities $P(o_i \in B_k)$ by $\frac{N_{B_k}}{N}$. For this, define 
	\begin{align*}
		a_k := \max_{b_i \in\mathcal{G}_{w_i}} & \frac{1}{N_{B_k}} \sum_{j=1}^{N_{B_k}} u_{i, \text{M}} (\theta_i^{(j)}, b_i, \beta_{-i}(o_{-i}^{(j)})) \\
		&+ \min \left\{1, \tau_{i, B_k} + \tilde{\varepsilon}_{i, \text{Pdim}}(N_{B_k}) + \tilde{\varepsilon}_{i, \text{disp}}(N_{B_k}) \right\}
	\end{align*}
	Then, we have $a_k \in [0, 2]$ for every $1 \leq k \leq N_{\mathcal{B}_i}$. By applying Lemma~\ref{thm:aux-lemma-approximate-convex-combination-of-bounded-variables}, we have with probability at least $1 - \delta$
	\begin{align} \label{equ:proof-main-theorem-interdependent-marginal-prob-estimation-bound}
		\abs{\sum_{k=1}^{N_{\mathcal{B}_i} } \left(P(o_i \in B_k) - \frac{N_{B_k}}{N} \right) a_k} \leq \sqrt{\frac{2}{N} \log\left( \frac{2n}{\delta} \right)}.
	\end{align}
	
	We combine the above results to give the full statement. We apply Equation \ref{equ:proof-main-theorem-interdependent-prior-ex-ante-utility} to estimate the \textit{ex ante} utility under strategy profile $\beta$. To estimate the best-response utility for each $B_k \in \mathcal{B}_i$, we either apply Equation~\ref{equ:proof-main-theorem-interdependent-prior-br-utility-to-estimator-difference-single-segment-bounded-by-one} for a trivial bound of one or combine Equations~\ref{equ:proof-main-theorem-interdependent-prior-pseudo-dimension-bound} and \ref{equ:proof-main-theorem-interdependent-prior-dispersion-bound} for a potentially stronger upper bound. Finally, we use Equation~\ref{equ:proof-main-theorem-interdependent-marginal-prob-estimation-bound} to justify the estimation of the marginal probabilities $P(o_i \in B_k)$ for every $B_k \in \mathcal{B}_i$. In total, each of the four equations holds with probability $1 - \delta$. By applying a union bound, all four equations hold with probability $1 - 4\delta$. Therefore, by additionally applying Equation \ref{equ:proof-main-theorem-interdependent-prior-constant-best-response-error}, we have with probability $1- 4\delta$ over the draw of the sets $\left\{\giventhat{\mathcal{D}^{\beta}(B_k)}{B_k \in \mathcal{B}_i, i \in [n]} \right\}$ and $\mathcal{D}^{\beta}$, for mechanism $\text{M}$, agents $i \in [n]$, and segments $B_k \in \mathcal{B}_i$,
	\begin{align*}
		&\tilde{\ell}_i(\beta_i, \beta_{-i}) = \sup_{\beta^{\prime}_i \in \Sigma_i} \tilde{u}_{i, \text{M}}(\beta^{\prime}_i, \beta_{-i}) - \tilde{u}_{i, \text{M}}(\beta_i, \beta_{-i})\\
		&= \sup_{\beta^{\prime}_i \in \Sigma_i} \tilde{u}_{i, \text{M}}(\beta^{\prime}_i, \beta_{-i}) - \sup_{b \in \mathcal{A}_i^{N_{\mathcal{B}_i}}} \tilde{u}_{i, \text{M}}\left(\sum_{k=1}^{N_{\mathcal{B}_i}} b_k \mathds{1}_{B_k}, \beta_{-i} \right) \\
		&+ \sup_{b \in \mathcal{A}_i^{N_{\mathcal{B}_i}}} \tilde{u}_{i, \text{M}}\left(\sum_{k=1}^{N_{\mathcal{B}_i}} b_k \mathds{1}_{B_k}, \beta_{-i} \right) - \sum_{k=1}^{N_{\mathcal{B}_i}} P(o_i \in B_k) \sup_{b_i \in\mathcal{A}_i} \frac{1}{N_{B_k}} \sum_{j=1}^{N_{B_k}} u_{i, \text{M}} (\theta_i^{(j)}, b_i, \beta_{-i}(o_{-i}^{(j)})) \\
		&+ \sum_{k=1}^{N_{\mathcal{B}_i}} P(o_i \in B_k) \left( \sup_{b_i \in\mathcal{A}_i} \frac{1}{N_{B_k}} \sum_{j=1}^{N_{B_k}} u_{i, \text{M}} (\theta_i^{(j)}, b_i, \beta_{-i}(o_{-i}^{(j)})) - \sup_{b_i \in \mathcal{G}_{w_i}} \frac{1}{N_{B_k}} \sum_{j=1}^{N_{B_k}} u_{i, \text{M}} (\theta_i^{(j)}, b_i, \beta_{-i}(o_{-i}^{(j)})) \right)\\
		&+ \sum_{k=1}^{N_{\mathcal{B}_i}} P(o_i \in B_k) \max_{b_i \in\mathcal{G}_{w_i}} \frac{1}{N_{B_k}} \sum_{j=1}^{N_{B_k}} u_{i, \text{M}} (\theta_i^{(j)}, b_i, \beta_{-i}(o_{-i}^{(j)})) - \frac{1}{N} \sum_{l=1}^{N} u_{i, \text{M}} (\theta_i^{(l)}, \beta_i(o_i^{(l)}), \beta_{-i}(o_{-i}^{(l)}))\\
		&+\frac{1}{N} \sum_{l=1}^{N} u_{i, \text{M}} (\theta_i^{(l)}, \beta_i(o_i^{(l)}), \beta_{-i}(o_{-i}^{(l)})) - \tilde{u}_{i, \text{M}}(\beta_i, \beta_{-i}) \\
		&\leq \sum_{k=1}^{N_{\mathcal{B}_i}} P(o_i \in B_k) \max_{b_i \in\mathcal{G}_{w_i}} \frac{1}{N_{B_k}} \sum_{j=1}^{N_{B_k}} u_{i, \text{M}} (\theta_i^{(j)}, b_i, \beta_{-i}(o_{-i}^{(j)})) - \frac{1}{N} \sum_{l=1}^{N} u_{i, \text{M}} (\theta_i^{(l)}, \beta_i(o_i^{(l)}), \beta_{-i}(o_{-i}^{(l)})) \\
		&+ \sqrt{\frac{2}{N} \log\left(\frac{2n}{\delta} \right)} +  \sum_{k=1}^{N_{\mathcal{B}_i}} P(o_i \in B_k) \left(\tau_{i, B_k} + \tilde{\varepsilon}_{i, \text{Pdim}}(N_{B_k}) + \tilde{\varepsilon}_{i, \text{disp}}(N_{B_k}) \right)\\
		&= \sum_{k=1}^{N_{\mathcal{B}_i}} P(o_i \in B_k) a_k - \frac{1}{N} \sum_{l=1}^{N} u_{i, \text{M}} (\theta_i^{(l)}, \beta_i(o_i^{(l)}), \beta_{-i}(o_{-i}^{(l)})) + \sqrt{\frac{2}{N} \log\left(\frac{2n}{\delta} \right)}\\
		& \overset{\text{Equ. } \ref{equ:proof-main-theorem-interdependent-marginal-prob-estimation-bound}}{\leq} \sum_{k=1}^{N_{\mathcal{B}_i}} \frac{N_{B_k}}{N} \max_{b_i \in\mathcal{G}_{w_i}} \frac{1}{N_{B_k}} \sum_{j=1}^{N_{B_k}} u_{i, \text{M}} (\theta_i^{(j)}, b_i, \beta_{-i}(o_{-i}^{(j)})) - \frac{1}{N} \sum_{l=1}^{N} u_{i, \text{M}} (\theta_i^{(l)}, \beta_i(o_i^{(l)}), \beta_{-i}(o_{-i}^{(l)})) \\
		&+ 2 \sqrt{\frac{2}{N} \log\left(\frac{2n}{\delta} \right)} + \sum_{k=1}^{N_{\mathcal{B}_i}} \frac{N_{B_k}}{N} \min\left\{1, \left(\tau_{i, B_k} + \tilde{\varepsilon}_{i, \text{Pdim}}(N_{B_k}) + \tilde{\varepsilon}_{i, \text{disp}}(N_{B_k}) \right)\right\}.
	\end{align*}
	
\end{proof}

\end{document}